\documentclass[letterpaper,11pt]{article}

\usepackage[font=small,labelfont=bf]{caption}
\usepackage[utf8]{inputenc}
\usepackage{microtype}
\usepackage{xcolor}
\usepackage{amsthm}
\usepackage{amsmath}
\usepackage{amsthm}
\usepackage{amssymb}
\usepackage{xspace}
\usepackage{graphicx}
\usepackage{fullpage}
\usepackage[colorlinks=false]{hyperref}
\usepackage{url}
\usepackage[capitalise]{cleveref}
\usepackage[shortlabels]{enumitem}
\usepackage{subfig}
\usepackage[export]{adjustbox}
\usepackage{thm-restate}
\usepackage{textpos} 
\usepackage[normalem]{ulem}

\usepackage{libertine}

\newtheorem{theorem}{Theorem}[section]
\newtheorem{lemma}[theorem]{Lemma}
\newtheorem{proposition}[theorem]{Proposition}
\newtheorem{definition}[theorem]{Definition}
\newtheorem{corollary}[theorem]{Corollary}

\newtheorem{claim}[theorem]{Claim}
\newtheorem{observation}[theorem]{Observation}

\newenvironment{innerproof}
 {\proof}
 {\endproof}

\newcommand{\defparproblem}[4]{
	\vspace{3mm}
	\noindent\fbox{
		\begin{minipage}{0.96\linewidth}
			\begin{tabular*}{\linewidth}{@{\extracolsep{\fill}}lr} \textsc{#1} & {\bf{Parameter:}} #3 \\ \end{tabular*}
			{\bf{Input:}} #2 \\
			{\bf{Task:}} #4
		\end{minipage}
	}
	\vspace{2mm}
}

\newcommand{\Oh}{\mathcal{O}}

\newcommand{\cl}{\ensuremath{\mathsf{cl}}}
\newcommand{\sub}{\subseteq}
\newcommand{\sm}{\setminus}

\newcommand{\calC}{{\mathcal{C}}}

\newcommand{\rr}{\ensuremath{\mathbb{R}}}

\newcommand{\tcal}{\ensuremath{\mathcal{T}}}
\newcommand{\dcal}{\ensuremath{\mathcal{D}}}
\newcommand{\pcal}{\ensuremath{\mathcal{P}}}
\newcommand{\Ev}{\vec{E}}
\newcommand{\calP}{\pcal}
\newcommand{\calQ}{{\cal Q}}

\newcommand{\calT}{\tcal}
\newcommand{\oSigma}{\overline{\Sigma}}
\newcommand{\charWord}{\mathsf{char}}

\newcommand{\pdspfull}{{\sc{Planar Disjoint Shortest Paths}}\xspace}
\newcommand{\pdsp}{{\sc{PDSP}}\xspace}
\newcommand{\pdap}{{\sc{PDAP}}\xspace}
\newcommand{\dagg}{\ensuremath{\text{-}\mathsf{dag}}\xspace}
\newcommand{\terset}{terminal-superset\xspace}
\newcommand{\dagcut}{\ensuremath{\text{-}\mathsf{dag}\text{-}\mathsf{cut}}\xspace}
\newcommand{\udagcut}{\ensuremath{\mathsf{dag}\text{-}\mathsf{cut}}\xspace}
\newcommand{\dagcuts}{\ensuremath{\text{-}\mathsf{dag}\text{-}\mathsf{cuts}}\xspace}
\newcommand{\dagring}{\ensuremath{\text{-}\mathsf{dag}\text{-}\mathsf{ring}}\xspace}
\newcommand{\udagring}{\ensuremath{\mathsf{dag}\text{-}\mathsf{ring}}\xspace}
\newcommand{\ring}{\ensuremath{\mathrm{Ring}}\xspace}
\newcommand{\disc}{\ensuremath{\mathrm{Disc}}\xspace}
\newcommand{\splitt}{\mathsf{split}\xspace}
\newcommand{\sameside}{\mathsf{same}\text{-}\mathsf{side}\xspace}

\newcommand{\reduce}{\mathsf{reduce}\xspace}
\newcommand{\load}{\mathsf{load}\xspace}
\newcommand{\minload}{\mathsf{comb}\text{-}\mathsf{load}\xspace}
\newcommand{\topload}{\mathsf{top}\text{-}\mathsf{load}\xspace}
\newcommand{\algload}{\mathsf{alg}\text{-}\mathsf{load}\xspace}
\newcommand{\wind}{\mathsf{WindNum}\xspace}

\newcommand{\padpfull}{{\sc{Planar Annotated Disjoint Paths}}\xspace }

\renewcommand{\leq}{\leqslant}
\renewcommand{\le}{\leqslant}
\renewcommand{\geq}{\geqslant}
\renewcommand{\ge}{\geqslant}

\newcommand{\funding}{M.P. and G.S. were supported by the project BOBR that is funded from the European Research Council (ERC) under the European Union’s Horizon 2020 research and innovation programme with grant agreement No. 948057. In particular, a majority of work on this manuscript was done while G.S. was affiliated with University of Warsaw.
M.W. was supported by Polish National Science Centre SONATA-19 grant 2023/51/D/ST6/00155.}

\title{Planar Disjoint Shortest Paths is Fixed-Parameter Tractable\thanks{\funding}}
\author{
Micha\l{} Pilipczuk\thanks{Institute of Informatics, University of Warsaw, Poland. E-mail: \url{michal.pilipczuk@mimuw.edu.pl}}
\and
Giannos Stamoulis\thanks{Université Paris Cité, CNRS, IRIF, F-75013, Paris, France. E-mail: \url{stamoulis@irif.fr}}
\and 
Micha\l{} W\l{}odarczyk\thanks{Institute of Informatics, University of Warsaw, Poland. E-mail: \url{michal.wloda@gmail.com}}
}
\date{}

\begin{document}

\maketitle

\begin{abstract}
    \!In the {\sc{Disjoint Shortest Paths}} problem one is given a graph~$G$ and a set $\tcal\!=\!\{(s_1,t_1),\ldots,(s_k,t_k)\}$ of $k$ vertex pairs. The question is whether there exist vertex-disjoint paths $P_1,\ldots,P_k$ in $G$ so that each $P_i$ is a shortest path between $s_i$ and $t_i$.
    While the problem is known to be $\mathsf{W}[1]$-hard in general, we show that it is fixed-parameter tractable on planar graphs with positive edge weights.
    Specifically, we propose an algorithm 
    for {\sc{Planar Disjoint Shortest Paths}} with running time $2^{\Oh(k\log k)}\cdot n^{\Oh(1)}$.
    Notably, our parameter dependency is better than state-of-the-art $2^{\Oh(k^2)}$ for the  {\sc{Planar Disjoint Paths}} problem, where the sought paths are not required to be shortest paths.
\end{abstract}

\begin{textblock}{20}(12.2, 6.2)
\includegraphics[width=40px]{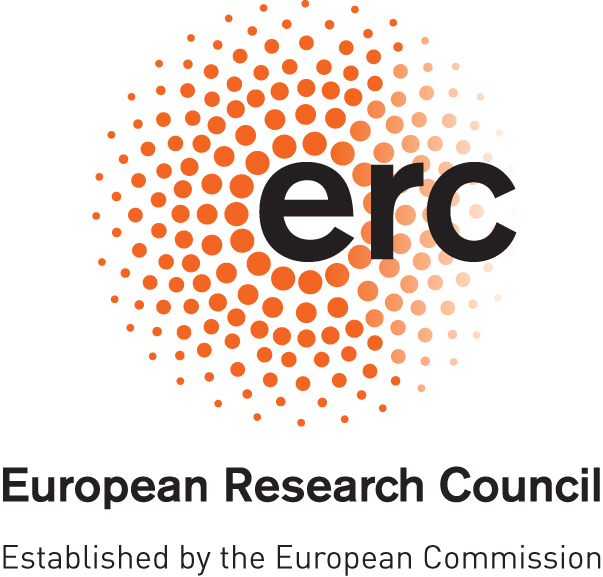}%
\end{textblock}
\begin{textblock}{20}(12.2, 7.2)
\includegraphics[width=40px]{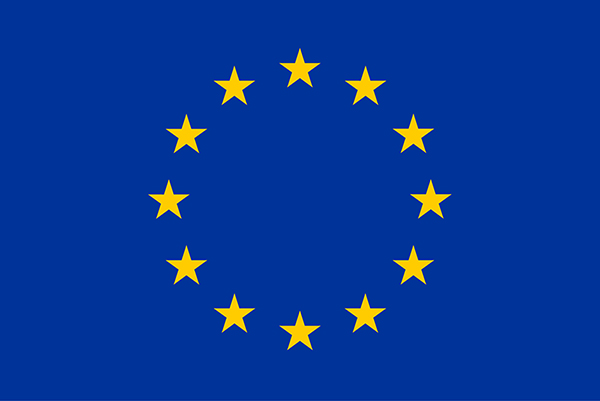}%
\end{textblock}

\thispagestyle{empty}

\newpage

\clearpage
\setcounter{page}{1}

\section{Introduction}

In the {\sc{Disjoint Paths}} problem, we are given a graph $G$ together with a set $\tcal=\{(s_1,t_1),\ldots,(s_k,t_k)\}$ of $k$ pairs of pairwise distinct vertices, called {\em{terminals}}. The task is to decide whether  there are pairwise vertex-disjoint paths $P_1,\ldots,P_k$ in $G$ so that each $P_i$ connects $s_i$ with $t_i$. Being perhaps one of the most important computational problems of topological nature, {\sc{Disjoint Paths}} was studied by Robertson and Seymour within their Graph Minors project, which resulted in establishing fixed-parameter tractability of the problem on general graphs.
More precisely, Robertson and Seymour proposed an $f(k)\cdot n^3$-time algorithm for a computable function $f$~\cite{robertson1995graph}. The dependency on $n$ in the time complexity has been subsequently improved to quadratic~\cite{kawarabayashi2012disjoint}, and very recently to almost linear-time~\cite{KorhonenPS24minor}.

The main idea behind essentially all the contemporary algorithms for {\sc{Disjoint Paths}} is the {\em{irrelevant vertex technique}}: iteratively find vertices that can be safely removed from the instance without changing the answer (hence they are called {\em{irrelevant}}), until the instance becomes so simple that it can be solved directly. Arguing that a given vertex is irrelevant typically relies on delicate rerouting arguments, which makes the technique not applicable to a range of interesting variants of the problem. In fact, some of those variants become computationally much harder, and their treatment requires other~tools.

The variant that we investigate in this paper is {\sc{Disjoint Shortest Paths}}. Here, we are given an instance $(G,\tcal)$ of {\sc{Disjoint Paths}} where $G$ is edge-weighted, and the task is to decide the existence of a solution $P_1,\ldots,P_k$ such that each $P_i$ is a {\em{shortest}} path between $s_i$ and $t_i$. Techniques based on finding irrelevant vertices are rather difficult to apply in the context of {\sc{Disjoint Shortest Paths}}, because rerouting arguments used to reason about the irrelevance of a vertex typically need to modify a solution without paying attention to the lengths of the paths. As a result, the complexity of {\sc{Disjoint Shortest Paths}} still remains quite an uncharted territory. 

{\sc{Disjoint Shortest Paths}} were first studied by Eilam-Tzoreff~\cite{EilamTzoreff98} in 1998, who proved polynomial-time solvability for $k=2$. Later, B\'erczi and Kobayashi~\cite{BercziK17} investigated the directed variant of the problem, and they gave a polynomial-time algorithm for $k=2$ on general digraphs and an $n^{\Oh(k)}$-time algorithm for planar digraphs, thereby putting the directed planar case in the complexity class $\mathsf{XP}$. A major development was delivered by Lochet~\cite{DBLP:conf/soda/Lochet21}, who gave an $\mathsf{XP}$ algorithm, with a running time of $n^{\Oh(k^{5^k})}$, on general undirected graphs with unit weights.
Subsequently, Bentert, Nichterlein, Renken, and Zschoche~\cite{bentert2023using} improved the exponent in the running time to $2^{\Oh(k \log k)}$.
Unfortunately, it was also proved in~\cite{bentert2023using,DBLP:conf/soda/Lochet21} that the problem is $\mathsf{W}[1]$-hard, hence the existence of an FPT algorithm is unlikely. However, so far it was open~\cite{WlodarczykZ23} whether such an algorithm can be designed in the setting of planar graphs.

\paragraph*{Our contribution.} In this work we answer the question stated above in the affirmative by proving the following theorem. Here \pdspfull is the restriction of {\sc{Disjoint Shortest Paths}} to planar undirected graphs
equipped with positive edge weights.

\begin{theorem}\label{thm:main}
    The \pdspfull problem can be solved in time $2^{\Oh(k\log k)}\cdot n^{\Oh(1)}$.
\end{theorem}

Remarkably, the achieved parameter dependency $2^{\Oh(k\log k)}$ is better than state-of-the-art $2^{\Oh(k^2)}$ for {\sc Planar Disjoint Paths}~\cite{cho2023parameterized, LokshtanovMPSZ20}.
This suggests that the shortest-paths variant might be easier on planar graphs, in contrast to the general setting where {\sc Disjoint Paths} is FPT but {\sc Disjoint Shortests Paths} is $\mathsf{W}[1]$-hard.
A similar phenomenon occurs for {\sc Disjoint Paths} on directed graphs with $k=2$: the basic problem is $\mathsf{NP}$-hard~\cite{fortune1980directed} whereas the shortest-paths variant is polynomial-time solvable~\cite{BercziK17}.
Observe that it is important to restrict to only positive edge weights, because allowing for zero weights would give a problem more general than basic {\sc Disjoint Paths.}

It is known that {\sc Disjoint Shortest Paths} on undirected graphs generalizes {\sc Disjoint Paths} on acyclic digraphs~\cite{AkmalWW24,DBLP:conf/soda/Lochet21}.
Moreover, the corresponding reduction preserves planarity so \pdspfull is at least as hard as {\sc Disjoint Paths} on  planar DAGs (\Cref{lem:prelim:planar-dag-reduction}).
As the {\sc Disjoint Paths} problem is fixed-parameter tractable even on general planar digraphs~\cite{CyganMPP13}, such a reduction does not entail any meaningful lower bound.
However, it suggests that \pdspfull may exhibit a behavior typical for problems on planar DAGs, and
provides a valuable hint that guides our analysis (see \Cref{sec:techniques} for  details).

Another consequence of the aforementioned reduction that is the {\sc Planar Edge-Disjoint Shortest Paths} problems is $\mathsf{W}[1]$-hard because {\sc Edge-Disjoint Paths} is $\mathsf{W}[1]$-hard on planar DAGs~\cite{Chitnis23}.
On the other hand,  {\sc Edge-Disjoint Paths} is FPT even on general graphs~\cite{kawarabayashi2012disjoint}.
Hence the requirement on the paths to be shortest makes the problem harder also in the planar edge-disjoint setting, giving yet another contrast to \Cref{thm:main}.

Amiri, Golshani, Kreutzer, and Siebertz~\cite{amiri14vertex} considered {\sc Disjoint Paths} on a subclass of planar DAGs that admit an upward drawing.
They showed that this variant can be solved in time $2^{\Oh(k\log k)}\cdot n$, which improved upon the double-exponential parameter dependency for general planar digraphs~\cite{CyganMPP13}.
They asked whether that result can be extended to all planar DAGs.
As a direct consequence of \Cref{thm:main} and \Cref{lem:prelim:planar-dag-reduction}, we answer their question affirmatively when it comes to parameter dependency.

\begin{corollary}
     The {\sc Disjoint Paths} problem can be solved in time $2^{\Oh(k\log k)}\cdot n^{\Oh(1)}$ on planar DAGs.
\end{corollary}

\paragraph*{Other related work.} Recently, there have been a few other works investigating different aspects of the complexity of {\sc{Disjoint Shortest Paths}}. 
A variant of the problem with congestion was studied by Amiri and Wargalla~\cite{amiri2020congestion}, followed by Akmal and Wein~\cite{akmal2023localtoglobal}. Bentert, Fomin, and Golovach~\cite{BentertFG25tight} studied the natural optimization version, in which one does not need to connect all the terminal pairs, but rather maximize the number of pairs connected. 
Very recently, Akmal, Vassilevska Williams, and Wein~\cite{AkmalWW24} considered the fine-grained complexity of the problem and improved several running times in the cases known to be polynomial-time solvable.

A closely related variant of the {\sc{Disjoint Paths}} problem, where lengths of paths are also considered, is {\sc{Min-Sum Disjoint Paths}}. In this variant, we seek a solution $P_1,\ldots,P_k$ to a given instance $(G,\tcal)$ of {\sc{Disjoint Paths}} that minimizes the sum of the lengths of the paths $P_i$. The complexity of this variant is even more poorly understood than that of {\sc{Disjoint Shortest Paths}}. It follows from the results of \cite{bentert2023using,DBLP:conf/soda/Lochet21} that the problem is $\mathsf{W}[1]$-hard on general graphs, but membership in $\mathsf{XP}$ is currently unknown. Bj\"orklund and Husfeldt gave a randomized polynomial-time algorithm for $k=2$~\cite{BjorklundH19}, but the polynomial-time solvability of the case $k=3$ remains open, even in the setting of planar graphs. Recently, an {\sc{FPT}} algorithm for the problem in a very specific class of planar graphs --- grid graphs --- was proposed by Mari, Mukherjee, Pilipczuk, and Sankowski~\cite{MariMS24}. 
In addition, the problem becomes tractable on planar graphs with certain arrangements of terminals~\cite{datta2018shortest, kobayashi2010shortest, verdiere2011shortest}.

\section{Techniques}
\label{sec:techniques}

To outline the ideas behind our algorithm, we begin by discussing the
$2^{\Oh(k^2)}\cdot n^{\Oh(1)}$-time algorithm for {\sc Planar Disjoint Paths} (not necessarily shortest) by Lokshtanov, Misra, Pilipczuk, Saurabh, and
Zehavi~\cite{LokshtanovMPSZ20}.
To this end, we first need to go back even further in time, to the $n^{\Oh(k)}$-time algorithm by Schrijver from 1994~\cite{schrijver1994finding}, which in fact even applies to the directed setting.
This algorithm is divided into two phases: first, one guesses the {\em homology class} of the solution, out of $n^{\Oh(k)}$ many, and then given a suitable representation of that homology class,
finds a solution in that class (as long as one exists) via a reduction to a specific polynomial-time solvable 2-CSP problem, called {\sc{Homology Feasibility}}. 
Roughly speaking, two solutions $\pcal, \pcal'$ are homologous if $\pcal'$ can be obtained from $\pcal$ (and vice versa) by a series a modifications that ``pull'' some part of a solution through some internal face.
During that process a solution may temporarily lose vertex- or even edge-disjointedness and so it is more convenient to work with homologies of {\em algebraic flows}, objects more general than families of disjoint paths (see \Cref{sec:homology}).
In essence, to enumerate the homology classes, Schrijver considered a Steiner tree $T$ spanning the terminal vertices and analyzed the flow obtained by pushing all the paths onto $T$.
Then, in order to pin down the homology class of such a flow $\pcal$ it is sufficient to specify for each {\em{spinal path}} $Q$ of $T$ the number of times $\pcal$ crosses $Q$, where a spinal path of $T$ is a maximal subpath whose internal vertices have degree $2$.
The maximal number of such passes is $n$ and the number of spinal paths $Q$ to consider is $\Oh(k)$; this yields the $n^{\Oh(k)}$ bound.

\begin{figure}[t]
\centering
\includegraphics[scale=0.6]{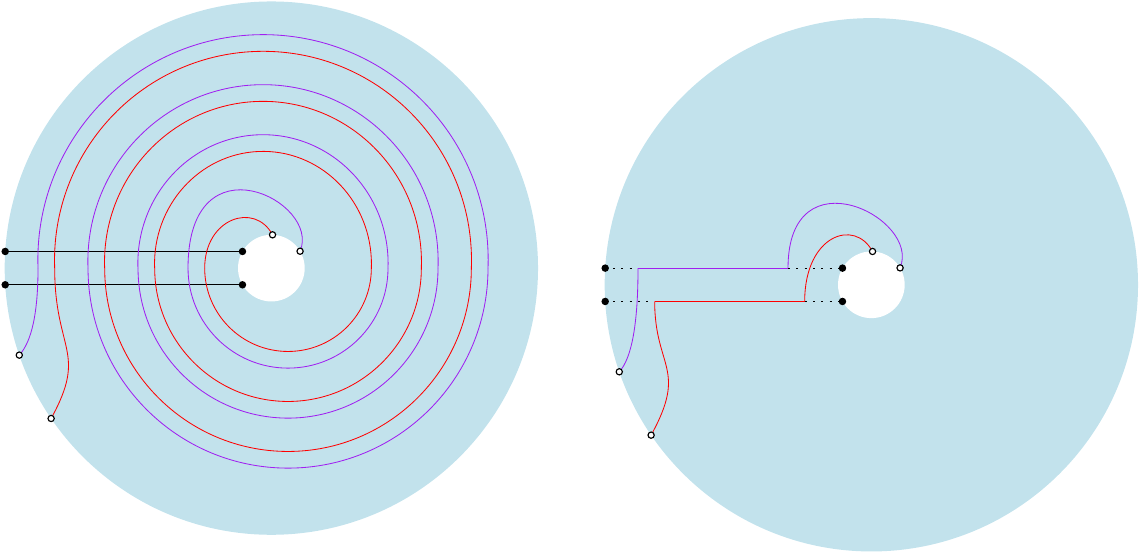}
\caption{
An illustration of the rerouting argument from~\cite{CyganMPP13}.
Left: linkage $\cal{Q}$ is given as the two black horizontal paths and  $\cal{P}$ comprises the purple and the red path.
When considered oriented towards the internal face $C_2$, $\cal{P}$ winds four times clockwise around $C_2$, compared  to $\cal{Q}$.
Right: linkage $\cal{P}'$ has the same endpoints as $\cal{P}$ while its winding number is reduced.
}\label{fig:reroute}
\end{figure}

The main idea behind the $2^{\Oh(k^2)}\cdot n^{\Oh(1)}$-time algorithm of~\cite{LokshtanovMPSZ20} is that it suffices to inspect only $2^{\Oh(k^2)}$ homology classes.
This is challenging because a solution might ``wind'' around some terminal vertex,
forming a deep ``spiral'', and solutions which differ in the depth of such a spiral are not homologous.
Taking into account that there are $\Omega(k)$ potential sources for spirals, we may a priori obtain $n^{\Omega(k)}$ different homologies.
Addressing this issue required incorporating two more techniques.
First, Adler et al.~\cite{AdlerKKLST17} had shown that one can reduce the treewidth $\mathbf{tw}$ of input graph $G$ to $2^{\Oh(k)}$ using the irrelevant vertex technique. 
This bounds the size of a maximal grid appearing in $G$ and simplifies the analysis of spirals.
The second ingredient originates from the fixed-parameter algorithm for {\sc Directed Planar Disjoint Paths}, due to Cygan, Marx, Pilipczuk, and Pilipczuk~\cite{CyganMPP13}.
Consider a directed plane graph $D$ with outer face $C_1$ and a distinguished internal face $C_2$ and let ${\cal Q}$ be a collection of disjoint paths connecting $C_1$ with $C_2$.
Then for any other such collection of paths $\pcal$ of the same size, one can reroute $\pcal$ to almost match the {\em winding number}\footnote{On an intuitive level, the winding number of a curve measures how many times it goes around a certain object on the plane.} of~${\cal Q}$, while preserving the endpoints of the paths in $\pcal$~\cite[Lemma 4.8]{CyganMPP13arxiv}, see \Cref{fig:reroute}.
Effectively, these two ideas reduce the problem to the case where each spinal path $Q$ of the considered Steiner tree $T$ has length $\Oh(\mathbf{tw})$, which in total gives the bound $\mathbf{tw}^{\Oh(k)} = 2^{\Oh(k^2)}$ on the number of enumerated homology classes.

As observed by B\'erczi and Kobayashi~\cite{BercziK17}, the second phase of Schrijver's algorithm can be adapted to work with \pdspfull. That is, the problem can be solved in polynomial time once we guess a correct homology class.
For $s_i,t_i \in V(G)$ we can define a directed acyclic graph, called the $(s_i,t_i)\dagg$, consisting of the oriented edges of $G$ that lie on some shortest $(s_i,t_i)$-path in $G$.
Then it is sufficient to annotate the instance of {\sc Homology Feasibility} with labels indicating that only the edges of the $(s_i,t_i)\dagg$ can accommodate the path $P_i$.
However, the main issue remains: how to bound the number of relevant homology classes as a function of $k$. Here, the main obstacles can be phrased as follows:
\begin{quote}
    The known rerouting arguments break down for the \pdspfull problem, because different edges may lie on shortest paths between different pairs of terminals. (Or more formally, belong to different digraphs $(s_i,t_i)\dagg$, for $i=1,2,\ldots,k$). In particular, there is no treewidth reduction technique available for \pdspfull.
\end{quote}
Therefore, we need to design a completely new strategy to handle the ``deep spirals'' that contribute to the non-FPT number of homology classes that need to be considered.

\begin{figure}[t]
\centering
\includegraphics[scale=0.7]{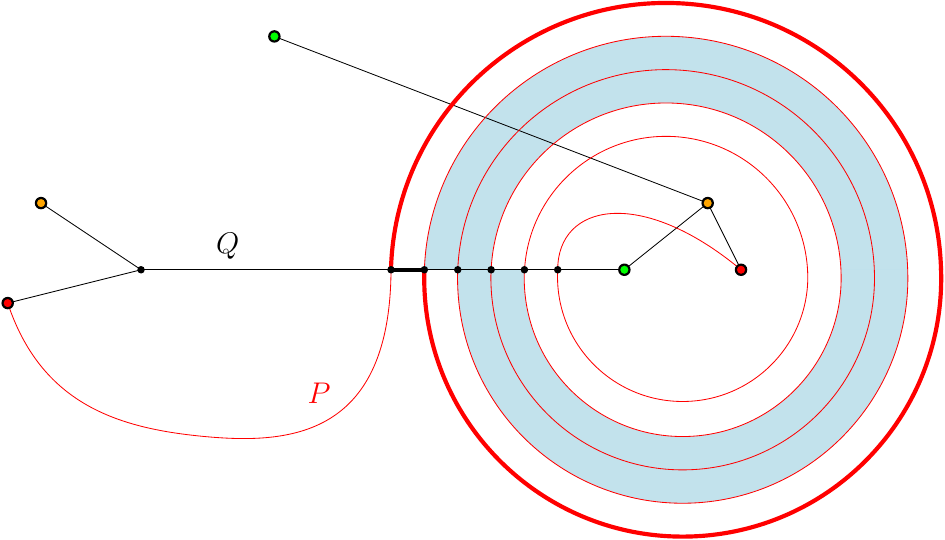}
\caption{
An illustration of the problematic scenario when the solution path $P$ crosses the geodesic spinal path $Q$ in the tree $T$ many times, because of the winding behavior.
The fat segments of $P$ and $Q$ represent subpaths in $G$ of  the same length.
By analyzing the shortest path structure in this part of the graph, we identify a region (in blue) where these structures coincide for all terminal pairs. 
}\label{fig:ring-outline}
\end{figure}

To mitigate this problem, we shall analyze the homology of a solution with respect to a Steiner tree $T$ in $G$ that is made of geodesics\footnote{A {\em{geodesic}} in a graph is a path that is a shortest path between its endpoints.}. That is, every spinal path $Q$ of $T$ must be a geodesic.
Consider some path $P$ from a solution that winds around some terminal with respect to $T$; see \Cref{fig:ring-outline}.
This means that $P$ crosses some spinal path $Q$ of $T$ several times in the same direction.
Since $P$ is also a geodesic, this gives rise to a rather peculiar subgraph of $G$.
More formally, let $u,v \in V(Q) \cap V(P)$ be two consecutive crossing points of $Q$ and $P$.
Then the distance between $u$ and $v$ along $Q$ must be equal to their distance along~$P$.
Consequently, when the number of such crossings is large, then there are multiple non-equivalent (in the topological sense) ways to reach the interior of the spiral from the exterior via a shortest path: on each crossing one can choose to follow either $Q$ or $P$. Now comes the key observation:

\begin{quote}
The situation described above implies the existence of a ring-shaped region $R$ 
that covers almost the entire spiral and in which the structure of shortest $(s_i,t_i)$-paths is the same for each terminal pair $(s_i,t_i)$ separated by the spiral; see \Cref{fig:ring-outline}.
\end{quote}

In other words, the digraphs $(s_i,t_i)\dagg$ and $(s_j,t_j)\dagg$ coincide inside $R$, whenever $s_i,s_j$ lie outside the spiral and $t_i,t_j$ lie inside the spiral, or vice versa.
In addition, it turns out that if $s_i,t_i$ are not separated by the spiral, then no shortest $(s_i,t_i)$-path intersects $R$.
As a consequence, we can neglect the structure of shortest paths inside $R$ and just treat this part of the graph as a DAG on which the problem behaves like {\sc Directed Planar Disjoint Paths}.
This enables us to take advantage of the rerouting argument from~\cite{CyganMPP13} inside the region $R$, to approximately (with constant additive error) find the winding number of the solution within $R$.

In order to turn this idea into an algorithm, we need to engineer the definition of such a ``ring-shaped region'' $R$ (formalized later as a ``$\mathsf{dag}$-$\mathsf{ring}$'') in such a way that:
\begin{itemize}[nosep]
    \item[(a)] it is safe to perform rerouting in $R$,
    \item[(b)] the existence of $R$ is implied by a ``winding behavior'',
    \item[(c)] any solution path can enter $R$ only once and then it must exit $R$ on the other side,
    \item[(d)] there exists a decomposition of the input graph using $\ell = \Oh(k)$ pairwise-disjoint regions $R_1, \dots, R_\ell$ which cannot be extended by any additional region $R_{\ell+1}$,
and
    \item[(e)] such a decomposition can be effectively computed.
\end{itemize}  
These properties ensure that a solution cannot form any large spirals outside the regions $R_1, \dots, R_\ell$.
This enables us to extract the subgraphs where spiraling is possible and analyze them by treating their
shortest paths structures as DAGs. 
We exploit this decomposition to refine the Steiner tree $T$ to only use paths with ``correct'' winding numbers within each $R_i$.
By property~(c), we can count the number $q$ of subpaths that traverse $R_i$ in any solution.
Next, we compute an arbitrary linkage $\mathcal{Q}$ of size $q$ that connects the two sides of $R_i$.
The rerouting argument from the work of Cygan et al.~\cite{CyganMPP13} ensures that 
any solution can be tweaked to almost match (within constant additive error) the winding number of $\mathcal{Q}$ within $R_i$.
Therefore, if we tweak $T$ to go alongside $\mathcal{Q}$, then some solution is guaranteed to have only few ``non-trivial'' crossings with $T$ in this area (\Cref{fig:decomposition}).

\begin{figure}[t]
\centering
\includegraphics[scale=0.9]{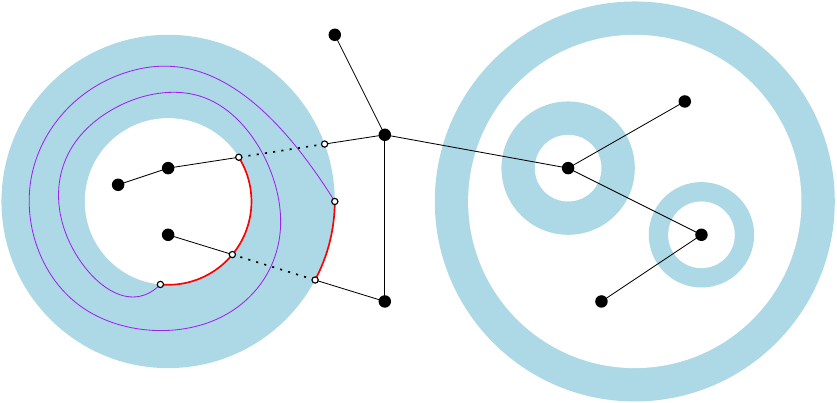}
\caption{
A refinement operation for a Steiner tree $T$ (in black) with respect to regions $R_1, \dots, R_\ell$ (blue).
The purple path $Q$ belongs to a linkage $\cal{Q}$ that connects the two sides of the region $R_1$. 
We want $T$ to wind within $R_1$ in the same fashion as $Q$ to match the behavior of a certain solution.
We remove the subpaths of $T$ that traverse $R_1$ (the dotted segments) and replace them with $Q$.
Next, we insert the red arcs to $T$
in order to maintain connectivity.
These arcs may not correspond to any paths in the graph so we have to augment it with additional edges that are flagged as forbidden to be used by a solution.}\label{fig:decomposition}
\end{figure}

Once we are equipped with such a refined Steiner tree $T$, we show that provided there exists some solution $\pcal$, there is also a solution $\pcal'$ for which number of non-trivial crossings with $T$ is bounded by a function $f(k)$.
By a careful analysis, one can show that there are only $f(k)^{\Oh(k)}$ many homology classes with this property, and so it is sufficient to enumerate and consider only those classes. 
Notably, we are able to estimate $f(k) = k^{\Oh(1)}$,
which gives only $2^{\Oh(k\log k)}$ relevant homology classes; less than $2^{\Oh(k^2)}$ for the standard {\sc Planar Disjoint Paths} problem.
To this end, we take advantage of the fact that geodesic paths can only cross in a monotone fashion.
When $Q$ is a path in $T$ that is disjoint from the regions $R_1, \dots, R_\ell$ and $P$ is a solution path, we can inspect a ``crossing pattern'' of $P$ and $Q$.
As now we can assume that $P \cup Q$ exhibits no ``winding behavior'', it remains to consider the situation when $P$ alternatively crosses $Q$ from the ''left'' and from the ''right''.
Such crossings form structures resembling ``handles'' and
each non-trivial handle must enclose some terminal vertex.
Thanks to the monotonicity of crossings,
the number of non-trivial handles of $P$ on $Q$ can be estimated as $\Oh(k)$ (see \Cref{fig:handles-and-load} on page \pageref{fig:handles-and-load}).
By considering all solution paths $P$ and all paths $Q$ from $T$, we upper bound the total number of such handles by $\Oh(k^3)$.

\paragraph*{Organization of the proof.} 
We begin with \Cref{sec:rings} in which we introduce the concept of a \udagring and establish several properties of this notion.
Then, in \Cref{prop:ring:ring-decomp} we show how to decompose the graph into maximal regions that, on the one hand, cover all the areas where winding is possible and, on the other hand, enjoy a convenient structure.

Next, in \Cref{sec:crossings} we use the notion of {\em homotopy} to count the number of ``non-trivial crossings'' between a path $Q$ and a solution $\pcal$.
In \Cref{prop:outside} we prove that this number is bounded for $Q$ that is disjoint from all the rings and in \Cref{prop:reroute:minload} we analyze the behavior of a path $Q$ that traverses some ring.
This section culminates with \Cref{prop:primalSteiner} in which we exploit the decomposition from \Cref{sec:rings} to construct a Steiner tree spanning the terminals that has few non-trivial crossings with a certain solution $\pcal$.

In \Cref{sec:schrijver} we create a link to the machinery invented by Schrijver~\cite{schrijver1994finding}.
First, we ``shift'' our Steiner tree to the dual graph and introduce the toolbox of homology.
In \Cref{lem:reduceToHom} we show how to identify the homology class of a solution using a collection of words assigned to paths on the Steiner tree.
In \Cref{lem:enumHom} we prove that it suffices to inspect $2^{\Oh(k\log k)}$ such collections, concluding the entire proof.

\section{Preliminaries}

We refer the reader to~\cite{Diestel} for standard graph-theoretic terms not explicitly defined here.
We use letter $G$ to denote undirected graphs and $D$ for directed graphs (digraphs).
A directed acyclic graphs is called a DAG. 
We denote the set of vertices as $V(G)$ (resp. $V(D)$) and the set of edges (resp. arcs) as $E(G)$ (resp. $A(D)$).
For two sets $V_1,V_2 \sub V(G)$, we denote by $E(V_1,V_2)$ the set of edges in $G$ with one endpoint in $V_1$ and the other in $V_2$. 
For an undirected graph $G$, we denote by $\Ev(G)$ to the set of all orientations of its edges, i.e., $\Ev(G) = \{(u,v), (v,u) \mid uv \in E(G)\}$.
For two sets $V_1,V_2 \sub V(G)$, we write $\Ev(V_1,V_2)$ to denote the set $\{(u,v) \mid uv \in E(G), u \in V_1, v\in V_2\}$.
A set $S \sub V(G)$ is a {\em{$(V_1,V_2)$-separator}} in $G$ if $S$ intersects every path connecting a vertex from $V_1$ to a vertex in $V_2$.

\paragraph{Plane graphs and duals.}
A {\em{plane embedding}} of a graph $G$ is given by a mapping from $V(G)$ to $\mathbb{R}^2$ and a mapping that associates with each edge $uv \in E(G)$ a simple curve in the plane connecting the images of $u$ and $v$, so that the curves given by two distinct edges can intersect only in a common endpoint. 
A graph is called {\em{planar}} if it admits a plane embedding.
{A {\em{plane graph}} is a graph with a fixed plane embedding, in which we identify the vertices with their images in the plane.} 
Similarly, we can consider plane multigraphs in which we allow for parallel edges and loops.
Here, the image of a loop is a simple closed curve.
A {\em face} in a plane graph (resp. multigraph) $G$ is a {maximal arc-connected subset of the plane minus the image of~$G$}.
We say that a vertex or an edge lies on a face $f$ if its image  belongs to the closure of $f$.
In every plane graph (resp. multigraph) there is exactly one face of infinite area, referred to as the {\em{outer face}}.
The set of faces of $G$ will be denoted by $F(G)$.

The {\em{dual multigraph}} $G^\star$ has the face set $F(G)$ of $G$ as the vertex set, and for every edge $e$ of $G$ we introduce the {\em{dual edge}} $e^\star$ connecting the two faces lying on the two sides of $e$. (Note that these two faces might coincide in case $e$ is a bridge, in which case $e^\star$ is a loop.) $G^\star$ admits a natural embedding in the plane in which every $f\in F(G)$ is mapped to an arbitrary point chosen within $f$, and dual edges are embedded so that every edge $e\in E(G)$ crosses its dual edge $e^\star$ at a single point.

A directed multigraph $D$ is called planar (resp. plane) if the underlying undirected multigraph is planar (resp. plane).
The dual $D^\star$ of  $D$ is a directed plane  multigraph defined analogously as above but taking into account the orientations of arcs.
An arc $a \in D$ that has face $f_1$ to the left (with respect to its orientation) and $f_2$ to the right corresponds to the arc $a^\star = (f_1,f_2)$ in $D^\star$.

\paragraph{Geometry.}
For a set $X \sub \rr^2$ we denote its closure by $\cl(X)$.
For a simple closed curve $S$ in $\rr^2$ we write $\disc(S)$ to denote the bounded connected component of $\rr^2 \sm S$.
Given two simple closed curves $S,S'$ in the plane, we say that $S$ and $S'$ \emph{cross} if both connected components of $\mathbb{R}^2\setminus S$ intersect $S'$.
For two simple closed curves $S_1, S_2$ that do not cross, we denote by $\ring(S_1, S_2)$ the union of all connected components $X$ of $\rr^2 \sm (S_1 \cup S_2)$ such that $\cl(X)$ intersects both $S_1$ and $S_2$.

\paragraph{Geodesics.}
Given a graph $G$ with positive edge weights $w \colon E(G) \to \rr^+$ and a set $F\subseteq E(G)$,
we use $w(F)$ to denote $\sum_{e\in F} w(e)$.
Also, given $s,t\in V(G)$, we use $d_G(s,t)$ (or just $d(s,t)$ when it does not lead to confusion) to denote the minimum value of $w(E(P))$ over all paths $P$ in $G$ connecting $s$ and $t$ (if no such path exists, we set $d_G(s,t)=\infty$).

A {\em geodesic} in a graph $G$ is a path that is a shortest path between its endpoints.
For a path $P$ and $u,v \in V(P)$ we denote by $P[{u,v}]$ the subpath of $P$ between $u$ and $v$.

\begin{observation}\label{lem:prelim:subpath}
    If $P$ is a geodesic in a graph $G$ and $u,v \in V(P)$ then $P[{u,v}]$ is also a geodesic.
\end{observation}

A path $P$ can be oriented in two ways.
A fixed orientation of $P$ induces a linear order on $V(P)$ denoted~$\le_P$. 
When we refer to an oriented $(x,y)$-path $P$, we consider it with orientation for which $x$ is the first element in the order $\le_P$ and $y$ is the last element.
For $z \in V(P)$ we abbreviate $P[\cdot,z] = P[x,z]$ and $P[z,\cdot] = P[z,y]$.

\begin{observation}\label{lem:prelim:replacement}
    Let $P$ be an oriented shortest $(x,y)$-path in a graph $G$, $u,v \in V(P)$ satisfy $x \le_P u <_P v \le_P y$, and $Q$ be a shortest $(u,v)$-path in $G$.
    Then the concatenation of paths $P[{x,u}] + Q + P[{v,y}]$ also forms a  shortest $(x,y)$-path.
\end{observation}

We state a crucial property of geodesics in graphs with positive edge weights, which will be exploited repeatedly in the analysis of the crossing patterns.

\begin{lemma}[Monotonicity of crossings]\label{lem:prelim:monotone3}
    Let $G$ be an undirected graph with positive edge weights $w \colon E(G) \to \rr^+$.
    Next, let $P, Q$ be oriented geodesics in $G$ and $v_1,v_2,v_3 \in V(P\cap Q)$.
    If $v_1 <_Q v_2 <_Q v_3$ then
    either $v_1 <_P v_2 <_P v_3$
    or $v_1 >_P v_2 >_P v_3$.
\end{lemma}
\begin{proof}
    Suppose that the claim does not hold.
    Then either $v_3$ is an internal vertex of $P[{v_1,v_2}]$
    or $v_1$ is an internal vertex of $P[{v_2,v_3}]$.
    Assume w.l.o.g. the first scenario.
    Note that $P[{v_1,v_2}]$ is a geodesic by \Cref{lem:prelim:subpath}.
    Let $Q'$ be obtained from $Q$ be replacing the subpath $Q[{v_1,v_2}]$ with $P[{v_1,v_2}]$.
    Then $Q'$ should also be a geodesic due to \Cref{lem:prelim:replacement}.
    But the vertex $v_3$ appears twice on $Q'$ so it is not a simple path and hence cannot be a geodesic because the length of any cycle in $G$ is positive.
\end{proof}

\subsection{Problem definition and basic facts}

\defparproblem{Disjoint Shortest Paths}
{An undirected graph $G$ with positive edge weights and a collection $\tcal=\{(s_1,t_1),\dotsm,(s_k,t_k)\}$ of pairs of vertices.}
{$k$}
{Find vertex-disjoint paths $P_1,\dots, P_k$ so that $P_i$ is a shortest $(s_i,t_i)$-path for each $i \in [k]$.}

For comparison, in the {\sc Disjoint Paths} problem we do not require $P_i$ to be a geodesic. 

For an undirected graph $G$ an {\em oriented subgraph} of $G$ is a directed graph with the vertex set $V(G)$ and an arc set contained in $\Ev(G)$.
It will later be convenient to analyze the problem from the following perspective, where oriented subgraphs of $G$ encode the metric structure.

\defparproblem{Disjoint Annotated Paths}
{An undirected graph $G$, a collection $\tcal=\{(s_1,t_1),\dotsm,(s_k,t_k)\}$ of pairs of vertices and a collection $\mathcal{D}=\{D_1,\ldots,D_k\}$ of oriented subgraphs of $G$.}
{$k$}
{Find vertex-disjoint paths $P_1,\dots, P_k$ so that $P_i$ is an oriented $(s_i,t_i)$-path in $D_i$ for each $i \in [k]$.}

In \textsc{Planar Disjoint Shortest Paths} (\pdsp) and \textsc{Planar Disjoint Annotated Paths} ({\sc PDAP}) we assume that the input graph $G$ is planar and that it is provided as a plane graph, that is, with a fixed embedding.

The following reduction from {\sc{Disjoint Paths}} on DAGs to {\sc{Disjoint Shortest Paths}} in undirected graphs is known, see e.g.~\cite[Theorem 2.2]{DBLP:conf/soda/Lochet21}. Here we note that it also works in the planar setting.

\begin{lemma}\label{lem:prelim:planar-dag-reduction}
    There is polynomial-time algorithm that, given a planar acyclic digraph $D$ with a collection of vertex pairs $\tcal_D$, outputs a planar undirected graph $G$ with a collection of vertex pairs $\tcal_G$ so that  $|\tcal_G| = |\tcal_D|$ and $(D,\tcal_D)$ is a yes-instance of {\sc Disjoint Paths} if and only if $(G,\tcal_G)$ is a yes-instance of {\sc Disjoint Shortest Paths}.
\end{lemma}
\begin{proof}
    First, we check if for every $(s_i,t_i) \in \tcal_D$ there exists an $(s_i,t_i)$-path in $D$; if no, we can output a trivial no-instance.
    Next, fix some topological order $v_1,\dots,v_n$ of vertices in $D$.
    Let $G$ be obtained from $D$ by transforming each arc $(v_i,v_j)$ into an undirected edge $v_iv_j$ and subdividing it into $j-i$ unit-length edges.
    Clearly $G$ remains planar.
    Observe that for any $i<j$ there is a 1-1 correspondence between
    directed $(v_i,v_j)$-paths in $D$ and undirected $(v_i,v_j)$-paths in $G$ of length $j-i$, which are shortest $(v_i,v_j)$-paths in $G$.
    Since for each $(s_i,t_i) \in \tcal_D$ the vertex $t_i$ is reachable from $s_i$ in $D$, 
    their distance in $G$ 
    equals their distance in the fixed topological order.
    Hence there is also a 1-1 correspondence between solutions to $(D,\tcal_D)$ and solutions to $(G,\tcal_G)$
\end{proof}

\Cref{lem:prelim:planar-dag-reduction} gives a parameter-preserving reduction from {\sc Disjoint Paths} on planar DAGs to \pdspfull.
Since the first problem is $\mathsf{NP}$-hard~\cite{amiri14vertex}, we know that the latter is as well.

\paragraph{Shortest paths DAGs.}
We next define \emph{shortest paths DAGs} for given pairs of terminals. A slightly different variant of this definition was considered in~\cite{AkmalWW24}.

\begin{definition}[Shortest paths DAGs]
    Let $G$ be an undirected graph with positive edge weights $w \colon E(G) \to \rr^+$ and $s, t \in V(G)$.
    We define $(s,t)\dagg$ as the oriented subgraph of $G$ with arcs defined as follows: $(u,v) \in A((s,t)\dagg)$ if $uv \in E(G)$ and $d_G(s,u) + w(uv) + d_G(v,t) = d_G(s,t)$.
\end{definition}

The following lemma is a straightforward consequence of the definition of shortest paths DAGs (as already observed in~\cite{AkmalWW24}).
\begin{lemma}
    \label{lem:prelim:spaths_in_dags}
    Let $G$ be an undirected graph with positive edge weights $w\colon E(G)\to \rr^+$ and $s,t\in V(G)$.
    Also let $P$ be an oriented $(s,t)$-path in $G$.
    Then $P$ is a geodesic in $G$ if and only if $P$ is a path in $(s,t)$\dagg.
\end{lemma}
\begin{proof}
    For the forward implication observe that every (oriented) edge $(u,v)$ of $P$ belongs to $A((s,t)\dagg)$ by definition and therefore $P$ is a path in $(s,t)$\dagg.
    To prove the second implication, we proceed by induction on the number $\ell$ of edges in $P$.
    The claim clearly holds for $\ell \le 1$ and so we assume $\ell \ge 2$.
    Let $v$ be the second vertex on $P$.
    Since $(s,v) \in A((s,t)\dagg)$
    we have $d_G(s,t) = w(sv) + d_G(v,t)$.
    By the inductive assumption the path $P[v,t]$ is a shortest $(v,t)$-path so its length equals $d_G(v,t)$.
    Consequently, the length of $P$ is $d_G(s,t)$.
\end{proof}

Since we work with positive edge weights, no $(s,t)$-walk in $(s,t)$\dagg may contain a cycle.
This implies that $(s,t)$\dagg is indeed a DAG.

\begin{observation}
\label{obs:prelim:no-dicycles}
    Let $G$ be an undirected graph with positive edge weights $w\colon E(G)\to \rr^+$ and $s,t\in V(G)$. Then $(s,t)$\dagg does not contain directed cycles.
\end{observation}

Next, we make note of the following fact that is a direct consequence of~\Cref{lem:prelim:spaths_in_dags}.

\begin{observation}\label{obs:prelim:annotated}
    Let $(G,\tcal)$ be an instance of \textsc{Disjoint Shortest Paths} and $(G,\tcal, \mathcal{D})$ be an instance of \textsc{Disjoint Annotated Paths} for $\mathcal{D} = \{(s_i,t_i)\dagg \mid i \in [k]\}$.
    Consider a family $\pcal$ of vertex-disjoint paths $P_1,\dots,P_k$ so that $P_i$ is an $(s_i,t_i)$-path.
    Then $\pcal$ is a solution to $(G,\tcal)$  if and only if it is a solution to $(G,\tcal, \mathcal{D})$.
\end{observation}

\subsection{Nice instances}

It will be convenient to restrict ourselves to instances with somewhat regularized structure.
The following definition refers to instances of both problems (\pdsp and \pdap).

\begin{definition}
\label{def:nice-instance}
    An instance $(G,\tcal)$ is called {\em nice} if the following conditions hold.
    \begin{enumerate}[nosep]
        \item $G$ is a connected plane graph.
        \item $\tcal$ is a collection of vertex pairs from $G$ in which no vertex appears more than once.
        \item Every vertex appearing in $\tcal$ has degree 1.

        \vspace{0.2cm}
        Additionally, for \pdsp we require one more condition.
        \vspace{0.2cm}

        \item Each edge $e \in E(G)$ appears as an arc in $(s,t)\dagg$ for some $(s,t) \in \tcal$.
    \end{enumerate}
\end{definition}

We  begin by reducing a general instance of the problem to a nice one.

\begin{lemma}\label{lem:nice-reduction}
There is an algorithm that, given an instance $(G,\tcal)$ of \pdsp, outputs, in polynomial time, one of the following:
\begin{itemize}[nosep]
    \item either a report that $(G,\tcal)$ is a no-instance, or
    \item a collection $\mathcal{I}$ of nice instances of \pdsp,
where $|\mathcal{I}|=\mathcal{O}(k+|G|)$ and
$(G,\tcal)$ is a yes-instance if and only if every instance in $\mathcal{I}$ is a yes-instance. 
\end{itemize} 
\end{lemma}

\begin{proof}
We assume that there is no pair $(s,t)\in\tcal$ such that $s$ and $t$ are in different connected components of $G$ and that no vertex appears in $\tcal$ more than once (since otherwise, we report that $(G,\tcal)$ is a no-instance).
Also, we modify the given graph $G$ as follows:
for every terminal vertex $s$, we add a new vertex $s'$ to $G$, adjacent only to $s$ and replace $s$ with $s'$ in the corresponding pair in $\tcal$. 
We extend the weight function $w$ to assign weight 1 to the new edge $ss'$ (this choice is arbitrary).
Next, we remove all edges of $G$ that do not appear as arcs in any $(s,t)$\dagg, for all $(s,t)\in \tcal$. To do this, for every $u,v\in V(G)$, we compute their distance in $G$ (this can be done in polynomial time) and therefore we can compute the directed graph $(s,t)$\dagg, for each $(s,t)\in \tcal$. Observe that after the removal of these edges, the number of connected components of the graph that contain a terminal vertex can grow by at most $k$.

The algorithm then constructs a collection of instances $\mathcal{I}$, one for each connected component of $G$ that contains some terminal vertex, i.e., $\mathcal{I}$ is the set of instances $(H,\tcal_{H})$, where $H$ is a connected component of $G$ (that contains a terminal vertex) and $\mathcal{T}_H\subseteq \mathcal{T}$ consists of the pairs from $\tcal$ with both vertices in $H$.
For every $(H,\tcal_H)\in\mathcal{I}$, compute a plane embedding of $H$ (this can be done in linear time, see~\cite{HopcroftT74} and~\cite{MehlhornM96}).
Note that thus, each instance in $\mathcal{I}$ satisfies~\Cref{def:nice-instance}\end{proof}

\section{Rings}
\label{sec:rings}

In order to define the ``ring-shaped regions'' as explained in~\Cref{sec:techniques}, 
for every bipartition of the set of vertices appearing in $\tcal$, we need to distinguish the pairs of $\tcal$ that are ``split'' in such partition. In fact, our constructions involve not only the set of terminal vertices but also an additional set of vertices, which correspond to the branching nodes of the Steiner tree composed of geodesics on which we will work on. These notions are formalized in the following definitions.

\begin{definition}
    Let $G$ be a graph and $\mathcal{T}$ be a set of pairs of vertices of $G$.
    We say that a vertex set $\widehat{T} \sub V(G)$
    is a {\em \terset} of $\tcal$ if $\widehat{T}$ contains all vertices appearing in $\tcal$, and possibly some more.
\end{definition}

\begin{definition}[Splitting partitions of terminals]
    Let $G$ be a graph, $\mathcal{T}$ be a set of pairs of vertices of $G$, and $\widehat{T} \sub V(G)$ be a \terset of $\tcal$.
    For a partition  $(X,Y)$ of $\widehat{T}$
    we define $\splitt_\tcal(X, Y)$  as the set of pairs $(s,t)$ such that $s \in X$, $t \in Y$ and either $(s,t)$ or $(t,s)$ belongs to~$\tcal$.
    We say that  $(X,Y)$ is a {\em splitting partition} if $\splitt_\tcal(X, Y) \ne \emptyset$.
    We also define $\sameside(X, Y)$
    as the collection of pairs $(s,t)\in \widehat{T}\times \widehat{T}$ with $s\neq t$,
    for which $\{s,t\} \sub X$ or $\{s,t\} \subseteq Y$.
\end{definition}

\subsection{DAG-cuts}

We next proceed to defining \emph{dag-cuts} for bipartitions of a terminal-superset $\widehat{T}$ of $\tcal$,
which correspond to cuts of the original graph where the cut edges should satisfy the following property: these edges belong to oriented shortest paths for \textsl{every} split pair of $\tcal$ but not for pairs of vertices of $\widehat{T}$ that are not split.

\begin{definition}[DAG-cuts]\label{def:dag-cut}
    Let $(G,\tcal)$ be a nice instance of \pdsp and $(X,Y)$ be a splitting partition of a \terset $\widehat{T}$ of $\tcal$.
    Also, let $W_X,W_Y\subseteq V(G)$ be two disjoint vertex sets with  $X\subseteq W_X$ and $Y\subseteq W_Y$.
    We define a $(W_X,W_Y)$\dagcut as a partition  $(V_X,V_Y)$ of $V(G)$ with the following properties.
    \begin{enumerate}[nosep]
        \item $W_X \sub V_X$ and $W_Y \sub V_Y$.
        \item For each $(v,u)\in \vec{E}(V_X,V_Y)$
            and each  $(t, t') \in \splitt_\tcal(X, Y)$, we have that $(v,u) \in A((t,t')\dagg)$.
        \item For each $(v,u)\in \vec{E}(V_X,V_Y)$
            and each  $(t, t') \in \sameside(X, Y)$, 
            $\{(v,u),(u,v)\}\cap A((t,t')\dagg)=\emptyset$.
    \end{enumerate}    
\end{definition}
See~\Cref{fig:dagcut} for an illustration of a dag-cut.

\begin{figure}[ht]
    \centering
    \includegraphics[width=6cm]{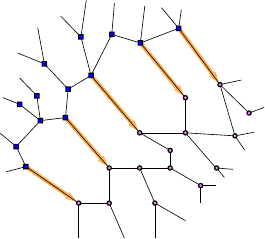}
    \caption{An example of a dag-cut. The vertices depicted as squares correspond to the set $V_X$ while the vertices depicted as discs correspond to the set $V_Y$.}
    \label{fig:dagcut}
\end{figure}

Whenever $P$ is a shortest $(t,t')$-path for some $(t,t') \in \tcal$, $P$ corresponds to a directed $(t,t')$-path in $(t,t')\dagg$.
If $(t, t') \in \sameside( X, Y)$ then $P$ clearly cannot use any edges from an $(X,Y)\dagcut$.
Otherwise, all these edges are oriented in the same direction in $(t,t')\dagg$.
Consequently, $P$ can cross an $(X,Y)\dagcut$ at most once.

\begin{observation}\label{obs:dag-cut:cross}
    Let $(G,\tcal)$ be a nice instance of \pdsp, $(X,Y)$ be a splitting partition of some \terset $\widehat{T}$ of $\tcal$, and $F \sub E(G)$ be a set of edges connecting the two sides of an $(X,Y)\dagcut$.
    Then the following holds:
    \begin{enumerate}
        \item if $P$ is a shortest $(t,t')$-path for some $(t,t')\in\splitt_\tcal(X,Y)$,
        then $|E(P) \cap F| \le 1$, and
        \item if $P$ is a shortest $(t,t')$-path for some $(t,t')\in\sameside(X,Y)$, then $E(P) \cap F = \emptyset$.
    \end{enumerate}
\end{observation}

\paragraph{Relating dag-cuts to cycles in the dual.}
We show that dag-cuts correspond to cycles of the dual graph that separate the vertices of the two different parts of the dag-cut.
In order to show this, we first prove that 
in a dag-cut, the corresponding vertex sets in the partition induce connected subgraphs of $G$.

\begin{lemma}\label{lem:conn-cut}
    Let $(G,\tcal)$ be a nice instance of \pdsp and $(X,Y)$ be a splitting partition of a \terset $\widehat{T}$ of $\tcal$.
    Also, let $W_X,W_Y\subseteq V(G)$ be two disjoint vertex sets with  $X\subseteq W_X$ and $Y\subseteq W_Y$.
    For every $(W_X,W_Y)$\dagcut $(V_X,V_Y)$ in $(G,\tcal)$, it holds that $G[V_X]$ and $G[V_Y]$ are connected.
\end{lemma}
\begin{proof}
    Due to property (1) of~\Cref{def:dag-cut}, we have that $X\subseteq W_X\subseteq V_X$ and $Y\subseteq W_Y\subseteq V_Y$.
    Since $(X,Y)$ is splitting, there is a pair $(s,t)\in \tcal$ such that one vertex belongs to $X$ and the other to $Y$.
    Assume without loss of generality that $s\in X$ and $t\in Y$.

    We first prove that $G[V_Y]$ is connected.
    Let $z$ be a vertex in $V_Y$.
    We show that there is a $(z,t)$-path in $G[V_Y]$.
    Since $G$ is connected, there is a $(z,t)$-path $P$ in $G$.
    We orient $P$ from $z$ to $t$ and this orientation induces a linear order $\le_P$ on $V(P)$.
    Let $v$ be the first (with respect to $\le_P$) vertex of $P$ that is in $V_X$; if such a vertex does not exist then $P$ is the claimed path in $G[V_Y]$.
    Let $e=wv\in E(P)$ such that $w<_P v$ ($w\in V_Y$). 
    By~\Cref{def:dag-cut},
    $(v,w)\in A((s,t)\dagg)$.
    Therefore, there is a shortest $(s,t)$-path $Q$ in $G$ that contains the edges $e$.
    Observe that, by definition, all edges of $Q$ appear in $A((s,t))\dagg$ as arcs oriented towards $t$.
    The latter implies that every vertex of  $Q[w,t]$ belongs to $V_Y$. Thus, the concatenation $P[z,w]+Q[w,t]$ is a $(z,t)$-path in $G[V_Y]$.
    
    To prove connectivity of $G[V_X]$, we work analogously.
    In particular, we pick an arbitrary vertex $z\in V_X$ and the first vertex $v\in V_Y$ in a $(z,s)$-path $P$ and we show the following: if $w$ is the neighbor of $v$ in $P$ (occurring before $v$ in $P$), for every shortest $(s,t)$-path $Q$ passing through $wv$, $V(Q[s,w])$ should be a subset of $V_X$ and therefore the concatenation $P[z,w]+Q[w,s]$ is a $(z,s)$-path in $G[V_X]$.
\end{proof}

Let us now give a formal definition of separability by a cycle in the dual.
Let $G$ be a plane graph, let $A$ and $B$ be two disjoint subsets of $V(G)$.
We say that a cycle $C$ of the dual $G^\star$
\emph{separates} $A$ and $B$ if $A$ and $B$ belong to different connected components of $\mathbb{R}^2\setminus C$.

\begin{lemma}[Folklore]
\label{lem:ring:cycle-dual}
    Let $G$ be a connected plane graph, let $(A,B)$ be a partition of $V(G)$. If $G[A]$ and $G[B]$ are connected subgraphs, then there is a (unique) cycle $C$ in the dual $G^\star$ that separates $A$ and $B$.
\end{lemma}

The existence of cycles in the dual corresponding to dag-cuts in the primal is stated in the following lemma, which is a direct implication of ~\Cref{lem:conn-cut} and~\Cref{lem:ring:cycle-dual}.

\begin{lemma}
\label{lem:ring:cut-cycle}
    Let $(G,\tcal)$ be a nice instance of \pdsp and $(X,Y)$ be a splitting partition of a \terset $\widehat{T}$ of $\tcal$.
    Also, let $W_X,W_Y\subseteq V(G)$ be two disjoint vertex sets with  $X\subseteq W_X$ and $Y\subseteq W_Y$.
    For every $(W_X,W_Y)$\dagcut $(V_X,V_Y)$ in $(G,\tcal)$, there is a (unique) cycle $C$ in the dual $G^\star$ that separates $V_X$ and $V_Y$. 
\end{lemma}

The fact that we can associate every dag-cut to a cycle in the dual (because of~\Cref{lem:ring:cut-cycle}) allows us to define \emph{the cycle} corresponding to each dag-cut, as well as the \emph{ring} bounded by two dag-cuts, defined as follows.
We write $(V,U)  \sqsubseteq (V',U')$ to denote $V\sub V'$.

\begin{definition}
    Let $(G,\tcal)$ be a nice instance of \pdsp and $(X,Y)$ be a splitting partition of a \terset $\widehat{T}$ of $\tcal$.
    Also, let $W_X,W_Y\subseteq V(G)$ be two disjoint vertex sets with  $X\subseteq W_X$ and $Y\subseteq W_Y$. 
    Given a $(W_X,W_Y)$\dagcut $\gamma=(V_{X},V_{Y})$, we use $\mathrm{Cycle}(\gamma)$ to denote the cycle in the dual $G^\star$ that separates $V_X$ and $V_Y$.
    Given two distinct $(W_X,W_Y)$\dagcuts $\gamma_1$ and $\gamma_2$ such that $\gamma_1\sqsubseteq \gamma_2$,
    we use $\ring(\gamma_1,\gamma_2)$ to denote $\ring(\mathrm{Cycle}(\gamma_1),\mathrm{Cycle}(\gamma_2))$.
\end{definition}

Keep in mind that $\mathrm{Ring}(\gamma_1,\gamma_2)$ is an open set.

\paragraph{Maximally pushed dag-cuts.}
We next show that one can consider dag-cuts that are ``maximally pushed'' towards $W_X$ and towards $W_Y$. The proof of the existence of such dag-cuts gives also a way to construct them.

\begin{lemma}\label{lem:maxcuts}
    Let $(G,\tcal)$ be a nice instance of \pdsp and $(X,Y)$ be a splitting partition of a \terset $\widehat{T}$ of $\tcal$.
    Also, let $W_X,W_Y\subseteq V(G)$ be two disjoint vertex sets with  $X\subseteq W_X$ and $Y\subseteq W_Y$.
    Suppose that there exists a $(W_X,W_Y)$\dagcut in $(G,\tcal)$.
    Then, there exist $(W_X,W_Y)$\dagcuts $\gamma_1$ and $\gamma_2$ such that for every other $(W_X,W_Y)$\dagcut $\gamma'$, we have $\gamma_1\sqsubseteq \gamma' \sqsubseteq \gamma_2$. 
    Furthermore, $\gamma_1$ and $\gamma_2$ can be found in polynomial time.
\end{lemma}
\begin{proof}
    Keep in mind that for every $(t,t')\in\splitt_\tcal(X, Y)$, by definition, $t\in X$ and $t'\in Y$.
    Let $F$ be the set of all edges $e$ of $G$ for which there is an orientation $\vec{e}$ so that for every $(t,t')\in\splitt_\tcal(X, Y)$, it holds that $\vec{e} \in A((t,t')\dagg)$.
    Among the edges in $F$, we consider the subset $\tilde{F}\subseteq F$ consisting of all edges $e$ in $F$ for which it holds that for every $(t,t')\in\sameside(X, Y)$,
    no orientation of $e$ belongs to $A((t,t')\dagg)$.
    
    We consider two directed graphs $D$ and $D'$ whose vertex set is $V(G)$ and their arc sets are obtained as follows:
    for every edge $e$ in $E(G)\setminus \tilde{F}$, 
    we add both orientations of $e$ to $A(D)$ and $A(D')$.
    Also, for each edge $uv\in \tilde{F}$,
    if $(u,v)$ is the orientation of $uv$ 
    that belongs to $A((t,t')\dagg)$ for every $(t,t')\in\splitt_\tcal(X,Y)$,
    then add the inverse orientation $(v,u)$ to $A(D)$ and (the former orientation) $(u,v)$ to $A(D')$.
    
    Let $C_X$ (resp. $C_Y$) be the set of all vertices in $V(G)$ that are reachable in $D$ (resp. $D'$) from vertices of $W_X$ (resp. $W_Y$), including $W_X$ (resp. $W_Y$).
    Note that for every $(W_X,W_Y)$\dagcut $(V_X,V_Y)$ in $(G,\tcal)$ it holds that $E(V_X,V_Y)\subseteq \tilde{F}$ and  $C_X\subseteq V_X$, $C_Y\subseteq V_Y$. To see why the latter holds assume towards a contradiction that there is a vertex $z\in C_X\setminus V_X$ (note that $C_X\setminus V_X\subseteq V_Y$).
    Since $z\in C_X$, there is a directed $(w,z)$-path $\vec{P}$ in $D[C_X]$, for some vertex $w\in W_X$.
    Also, since $w\in V_X$ and $z\in V_Y$,
    there is an arc $(v,u)\in A(D)$ that is also an arc of $P$ for which $v\in V_X$ and $u\in V_Y$.
    Since $(V_X,V_Y)$ is a $(W_X,W_Y)$\dagcut, due to~\Cref{def:dag-cut} we get that $uv\in \tilde{F}$ and, in particular, $(v,u)\in A(D)$. This yields a contradiction to the fact that $(u,v)\in A(D)$.
    Thus, given that $C_X\subseteq V_X$ and $C_Y\subseteq V_Y$, the existence of a $(W_X,W_Y)$\dagcut implies that $C_X$ and $C_Y$ are disjoint vertex sets.
    We set $\gamma_1:=(C_X,V(G)\setminus C_X)$ and $\gamma_2:=(V(G)\setminus C_Y, C_Y)$.
    
    To show that $\gamma_1$ and $\gamma_2$ are $(W_X,W_Y)$\dagcuts,
    we first observe that property (1) of~\Cref{def:dag-cut} holds for both $\gamma_1$ and $\gamma_2$, by definition.
    In order to show that properties (2) and (3) of~\Cref{def:dag-cut} are satisfied for $\gamma_1$ and $\gamma_2$, we first observe the following:
    since $G$ is connected and by definition of $D$ and $D'$,
    we have that both $E(C_X,V(G)\setminus C_X)$ and $E(V(G)\setminus C_Y,C_Y)$ are subsets of $\tilde{F}$.
    Moreover, for every edge $vu\in E(G)$ where $v\in C_X$ (resp. $v\in V(G)\setminus C_Y$)
    and $u\in V(G)\setminus C_X$ (resp. $u\in C_Y$), the orientation $(v,u)$ is the one that belongs to $A((t,t')\dagg)$ for every $(t,t')\in\splitt_\tcal(X,Y)$
    as otherwise we would have added $v$ to $C_X$ (resp. $C_Y$).
    This shows that property (2) holds for $\gamma_1$ and $\gamma_2$. 
    For property (3), recall that for every edge $e$ in $\tilde{F}$, it holds that for every $(t,t')\in\sameside(X, Y)$,
    no orientation of $e$ belongs to $A((t,t')\dagg)$; thus this also holds for every edge that is either in $E(C_X,V(G)\setminus C_X)$ or in $E(V(G)\setminus C_Y,C_Y)$.
    
    Also note that for every $(W_X,W_Y)$\dagcut $\gamma'$, we have $\gamma_1\sqsubseteq \gamma' \sqsubseteq \gamma_2$, since 
    $C_X\subseteq V_X'$ and $C_Y\subseteq V_Y'$ for every $(W_X,W_Y)$\dagcut $\gamma'=(V_X',V_Y')$.
    To conclude, we note that finding $\gamma_1$ and $\gamma_2$ reduces to computing $D$, which can be done in polynomial time by computing the $(s,t)$\dagg for every $(s,t)\in\splitt_{\tcal}(X,Y) \cup \sameside(X,Y)$.
\end{proof}

\subsection{DAG-rings}
\label{subsec:dagrings}

Using the notion of dag-cuts, we are now ready to define \emph{dag-rings}, which correspond to pairs of dag-cuts whose cycles bound a ring-shaped region where all edges are in (oriented) shortest paths between split pairs of terminals.

\begin{definition}[DAG-rings]\label{def:dag-rings}
    Let $(G,\tcal)$ be a nice instance of \pdsp and $(X,Y)$ be a splitting partition of a \terset $\widehat{T}$ of $\tcal$. Also, let $W_X,W_Y\subseteq V(G)$ be two disjoint vertex sets with  $X\subseteq W_X$ and $Y\subseteq W_Y$.
    A \emph{$(W_X,W_Y)$\dagring} is a partition $(U_X, U_{\mathsf{mid}},U_Y)$ of $V(G)$ with the following properties.
    \begin{enumerate}[nosep]
        \item $(U_X,U_{\mathsf{mid}}\cup U_Y)$ is a $(W_X,W_Y)$\dagcut.\label{ring:1}
    
        \item $(U_X\cup U_{\mathsf{mid}},U_Y)$ is a $(W_X,W_Y)$\dagcut.\label{ring:2}
        
        \item The set $E(U_X,U_Y)$ is empty.\label{ring:3}
    
        \item For each $e_1 \in E(U_X,U_{\mathsf{mid}})$, $e_2 \in E(U_{\mathsf{mid}},U_Y)$, and each  $(t,t') \in \splitt_\tcal(X, Y)$ there exists a shortest $(t,t')$-path $P$ in $G$ such that $\{e_1, e_2\} \sub E(P)$. \label{ring:4}
    \end{enumerate}
    We say that a pair $(\gamma,\gamma')$ of $(W_X,W_Y)$\dagcuts \emph{represents} a $(W_X,W_Y)$\dagring $(U_X,U_{\mathsf{mid}},U_Y)$ if $\gamma=(U_X,U_{\mathsf{mid}} \cup U_Y)$ and $\gamma'=(U_X \cup U_{\mathsf{mid}},U_Y)$.
\end{definition}
See~\Cref{fig:dagring} for an illustration of a dag-ring.

\begin{figure}[ht]
    \centering
    \includegraphics[width=6cm]{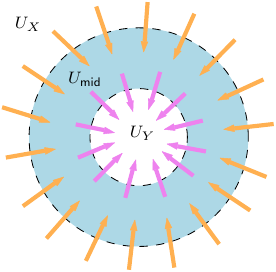}
    \caption{An example of a dag-ring. The two dag-cuts $(\gamma,\gamma')$ representing this dag-ring are illustrated with highlighted orange arcs and violet arcs respectively. The dual cycles $\mathrm{Cycle}(\gamma)$ and $\mathrm{Cycle}(\gamma')$ are illustrated as dashed cycles and the cyan region corresponds to $\ring(\gamma,\gamma')$.}
    \label{fig:dagring}
\end{figure}
The definition of dag-rings implies that no edges in $E(G[U_{\mathsf{mid}}])$ appear as arcs in the $(t,t')$\dagg of any pair $(t,t')\in\sameside(X,Y)$.

\begin{lemma}\label{lem:ring:ring_sameside}
    Let $(G,\tcal)$ be a nice instance of \pdsp, let $(X,Y)$ be a splitting partition of a \terset $\widehat{T}$ of $\tcal$,
    and let $(U_X,U_{\mathsf{mid}},U_Y)$ be an $(X,Y)$\dagring.
    Then for every edge in $E(G[U_{\mathsf{mid}}])$ and every  $(t,t')\in\sameside(X, Y)$, no orientation of $e$ belongs to $A((t,t')\dagg)$.
\end{lemma}
\begin{proof}
    Fix a pair $(t,t')\in\sameside(X,Y)$ and suppose towards a contradiction that $e$ belongs to a shortest $(t,t')$-path, which we denote by $P$.
    Since $X,Y$ are disjoint from $U_{\mathsf{mid}}$, there exists an edge $e'\in E(U_X,U_{\mathsf{mid}})$
    or
    an edge $e''\in E(U_{\mathsf{mid}},U_Y)$
    such that either $e'$ or $e''$ is in $E(P)$,
    a contradiction to the fact that $(U_X, U_{\mathsf{mid}}\cup U_Y)$ and $(U_X \cup U_{\mathsf{mid}},U_Y)$
    are $(X,Y)$\dagcuts.
\end{proof}

We next show that given a dag-ring, 
we can orient the edges of $E(G[U_{\mathsf{mid}}])$ following the orientation of the corresponding shortest-paths DAGs for the split terminal pairs, in order to obtain an acyclic digraph.

Let $\tcal$ be a set of pairs of vertices of a graph $G$ and $(X,Y)$ be a splitting partition of a \terset $\widehat{T}$ of $\tcal$.
Given an edge $e$ of $G$,
we say that an orientation $\vec{e}$ of $e$ is \emph{aligned with $(X,Y)$} if
for every $(t,t')\in\splitt_\tcal(X, Y)$, $\vec{e} \in A((t,t')\dagg)$.

\begin{lemma}\label{lem:ring:dag}
    Let $(G,\tcal)$ be a nice instance of \pdsp, let $(X,Y)$ be a splitting partition of a \terset $\widehat{T}$ of $\tcal$,
    and let $(U_X,U_{\mathsf{mid}},U_Y)$ be an $(X,Y)$\dagring.
    We can orient the edges in $E(G[U_{\mathsf{mid}}])$ to obtain an acyclic digraph $D$ with vertex set $U_{\mathsf{mid}}$ so that every $\vec{e}\in A(D)$ is aligned with $(X,Y)$.
\end{lemma}
\begin{proof}
    We have to show that every edge in $E(G[U_{\mathsf{mid}}])$ has an orientation that is aligned with $(X,Y)$ and that there is no directed cycle in the directed graph $D$ thus obtained.
     
    We first show that for every edge $e\in E(G[U_{\mathsf{mid}}])$ there is an orientation $\vec{e}$ such that for every $(t,t')\in\splitt_\tcal(X,Y)$, $\vec{e} \in A((t,t')\dagg)$.
    Due to~\Cref{lem:ring:ring_sameside}, no orientation of $e$ belongs to $A((t,t')\dagg)$ for any $(t,t')\in\sameside(X,Y)$.
    But the fact that $(G,\tcal)$ is a nice instance implies that $e$ is contained in the edge set of a shortest $(\tilde{t},\tilde{t}')$-path $P$, for some $(\tilde{t},\tilde{t}')\in\splitt_\tcal(X,Y)$.
    Since $X,Y$ are disjoint from $U_{\mathsf{mid}}$, there exists an edge $uv\in E(U_X,U_{\mathsf{mid}})$
    and an edge $wz\in E(U_{\mathsf{mid}},U_Y)$
    such that $\{uv,wz\}\subseteq E(P)$.
    Now observe that by Point~\ref{ring:4} of~\Cref{def:dag-rings}, for every $(t,t')\in\splitt_\tcal(X,Y)$ there is a shortest $(t,t')$-path $Q$ in $G$ such that $\{uv,wz\}\subseteq E(Q)$.
    Due to~\Cref{lem:prelim:replacement},
    $Q[t,u]+P[u,z]+Q[z,t']$ is an (oriented) shortest $(t,t')$-path that contains $e$.
    Therefore, there is an orientation of $e$ that is contained in $A((t,t')\dagg)$ for every $(t,t')\in\splitt_\tcal(X,Y)$.
    
    The fact that $D$ contains no directed cycles follows from the fact that if such a cycle existed, then it should also be a cycle in $(t,t')$\dagg, for every $(t,t')\in\splitt_\tcal(X,Y)$, a contradiction to~\Cref{obs:prelim:no-dicycles}.
\end{proof}

\begin{definition}\label{def:ring:dag-structure}
    Let $(G,\tcal)$ be a nice instance of \pdsp, let $(X,Y)$ be a splitting partition of a \terset $\widehat{T}$ of $\tcal$, and let  $\mathcal{R}=(U_X,U_{\mathsf{mid}}, U_Y)$ be an $(X,Y)$\dagring in $G$.
    Let also $D$ be the digraph from \Cref{lem:ring:dag}, $T_X$ be the set of vertices from $V(D)$ that are adjacent
    to vertices in $U_X$
    and $T_Y$ be the set of vertices from $V(D)$ that are adjacent to vertices in $U_Y$.
    We define the {\em DAG-structure} of $\mathcal{R}$ as the triple $(D,T_X,T_Y)$.
\end{definition}

\paragraph{Maximal dag-rings.}
    Let $(G,\tcal)$ be a nice instance of \pdsp, $\widehat{T}$ be a \terset of $\tcal$, and $(X,Y)$ be a splitting partition of $U$.
    Also, let $W_X,W_Y\subseteq V(G)$ be two disjoint vertex sets with  $X\subseteq W_X$ and $Y\subseteq W_Y$.
    We call a $(W_X,W_Y)$\dagring $(U_X,U_{\mathsf{mid}},U_Y)$  \emph{maximal} if for every other $(W_X,W_Y)$\dagring $(U_X',U_{\mathsf{mid}}',U_Y')$ it holds that $U_{\mathsf{mid}}'\subseteq U_{\mathsf{mid}}$.

    In order to show how to obtain maximal dag-rings, we prove that, given a dag-ring, taking maximally pushed dag-cuts we still obtain a dag-ring.
\begin{lemma}\label{lem:maxring}
    Let $(G,\tcal)$ be a nice instance of \pdsp, $U$ be a \terset of $\tcal$, and $(X,Y)$ be a splitting partition of $U$.
    Also, let $W_X,W_Y\subseteq V(G)$ be two disjoint vertex sets with  $X\subseteq W_X$ and $Y\subseteq W_Y$.
    Consider $(W_X,W_Y)$\dagcuts $\gamma_1 \sqsubseteq \gamma_2 \sqsubseteq \gamma_3 \sqsubseteq \gamma_4$ such that $(\gamma_2,\gamma_3)$ represents a $(W_X,W_Y)$\dagring.
    Then $(\gamma_1,\gamma_4)$ represents a $(W_X,W_Y)$\dagring as well.
\end{lemma}
\begin{proof}
    For every $i\in[4]$, we set $(V_X^i,V_Y^i):=\gamma_i$.
    Keep in mind that for every $i\in[4]$, $V_Y^i=V(G)\setminus V_X^i$.
    Also, we set $U_X:=V_X^1$,
    $U_Y:= V_Y^4$, and $U_{\mathsf{mid}}:= V(G)\setminus (V_X^1\cup V_Y^4)$;
    note that $U_{\mathsf{mid}}=V_X^4\setminus V_X^1=V_Y^1\setminus V_Y^4$ and therefore $\gamma_1=(U_X, U_{\mathsf{mid}}\cup U_Y)$ and $\gamma_4=(U_X\cup U_{\mathsf{mid}}, U_Y)$.
    In order to show that $(\gamma_1,\gamma_4)$ represents a $(W_X,W_Y)$\dagring, by~\Cref{def:dag-rings} we have to show that the set $E(U_X,U_Y)$ is empty and that for each $e_1 \in E(U_X,U_{\mathsf{mid}})$, $e_2 \in E(U_{\mathsf{mid}},U_Y)$, and each  $(t,t') \in \splitt_\tcal(X, Y)$ there exists a shortest $(t,t')$-path $P$ in $G$ such that $\{e_1, e_2\} \sub E(P)$.

    Since $(\gamma_2,\gamma_3)$ represents a $(W_X,W_Y)$\dagring,
    we have that the set $E(V_X^2,V_Y^3)$ is empty. Also, the fact that $\gamma_1\sqsubseteq \gamma_2$ and $\gamma_3\sqsubseteq\gamma_4$ implies that $V_X^1\subseteq V_X^2$ and $V_Y^4\subseteq V_Y^3$.
    Therefore, since $E(V_X^2,V_Y^3)=\emptyset$, we also get $E(V_X^1,V_Y^4)=\emptyset$ and thus $E(U_X,U_Y)=\emptyset$.

    Now, let $e_1=z_1w_1\in E(U_X,U_{\mathsf{mid}})$ and $e_4=z_4w_4\in E(U_{\mathsf{mid}},U_Y)$ and assume that $z_1\in U_X$, $w_1,z_4\in U_{\mathsf{mid}}$, and $w_4\in U_Y$.
    Also let $(t,t')\in\splitt_\tcal(X, Y)$.
    We show that there is a shortest $(t,t')$-path in $G$ that contains both edges $e_1$ and $e_4$.
    Since $U_{\mathsf{mid}}\subseteq V_Y^1$, we have that $e_1\in E(V_X^1,V_Y^1)$.
    Also, the fact that $\gamma_1$ is a dag-cut implies that 
    $(z_1,w_1)\in A((t,t')\dagg)$.
    Therefore, there is a shortest $(t,t')$-path $P$ in $G$ that contains $e_1$.
    Since $\gamma_2$ is a $(W_X,W_Y)$\dagcut, there is an edge $e_2=z_2w_2$ of $P$ with
    $z_2\in V_X^2$ and $w_2\in V_Y^2$.
    Similarly, there is a shortest $(t,t')$-path $P'$ that contains $e_4$ and an edge $e_3\in E(P')$ such that $e_3=z_3w_3$, where $z_3\in V_X^3$ and $w_3\in V_Y^3$.
    Also, the fact that $(\gamma_2,\gamma_3)$ represents a $(W_X,W_Y)$\dagring,
    implies that there is a shortest $(t,t')$-path $Q$ in $G$ such that $\{e_2,e_3\}\subseteq E(Q)$.
    Observe that due to~\Cref{lem:prelim:subpath},
    $Q[t,z_2]$ and $P[t,z_2]$ are shortest $(t,z_2)$-paths in $G$ and  $Q[w_3,t']$ and $P'[w_3,t]$ are shortest $(w_3,t')$-paths in $G$.
    Therefore, due to~\Cref{lem:prelim:replacement},
    we can replace $Q[t,z_2]$ by $P[t,z_2]$ and $Q[w_3,t']$ by $P'[w_3,t]$ in $Q$ and obtain
    the concatenation $P[t,z_2]+Q[z_2,w_3]+P'[w_3,t']$ which is a shortest $(t,t')$-path in $G$ that contains the edges $e_1$ and $e_4$.
\end{proof}

Using~\Cref{lem:maxring}, we reduce the computation of a maximal dag-ring to the computation of maximally pushed dag-cuts.
 
\begin{lemma}\label{lem:ring:max-ring}
    Let $(G,\tcal)$ be a nice instance of \pdsp and $(X,Y)$ be a splitting partition of a \terset $\widehat{T}$ of $\tcal$.
    Also, let $W_X,W_Y\subseteq V(G)$ be two disjoint vertex sets with  $X\subseteq W_X$ and $Y\subseteq W_Y$.
    Suppose that there exists at least one $(W_X,W_Y)$\dagring.
    Then there exists a maximal $(W_X,W_Y)$\dagring, which can be found in polynomial time.
\end{lemma}
\begin{proof}
    Let $(\gamma_1,\gamma_2)$ be a pair of $(W_X,W_Y)$\dagcuts that represent a $(W_X,W_Y)$\dagring with $\gamma_1\sqsubseteq \gamma_2$; such a pair is assumed to exist in the statement of the lemma.
    Also, let $\gamma_1',\gamma_2'$ be two $(W_X,W_Y)$\dagcuts such that for every other $(W_X,W_Y)$\dagcut $\gamma'$, we have $\gamma_1'\sqsubseteq \gamma'\sqsubseteq \gamma_2'$; such $\gamma_1',\gamma_2'$ exist because of~\Cref{lem:maxcuts}.
    Therefore, we have that $\gamma_1'\sqsubseteq \gamma_1\sqsubseteq\gamma_2\sqsubseteq\gamma_2'$, which implies, due to~\Cref{lem:maxring}, that $(\gamma_1',\gamma_2')$ represents a $(W_X,W_Y)$\dagring.
    Maximality of this latter dag-ring follows from the fact that for every
    $(W_X,W_Y)$\dagcut $\gamma'$, we have
    $\gamma_1'\sqsubseteq \gamma'\sqsubseteq\gamma_2'$.
    Also, due to~\Cref{lem:maxcuts}, $\gamma_1'$ and $\gamma_2'$ can be computed in polynomial time.
    Note also that~\Cref{ring:4} of~\Cref{def:dag-rings} can be verified in polynomial time by checking if for every pair of edges $e_1 = (u_1,v_1),e_2 = (u_2,v_2) \in E(G)$ with $u_1\in U_X,v_1\in U_{\mathsf{mid}}$ and  $u_2\in U_{\mathsf{mid}},v_2\in U_Y$, and  every $(t,t') \in \splitt_\tcal(X,Y)$, we have $d_G(t,t')=d_G(t,u_1)+w(e_1)+d_G(v_1,u_2)+w(e_2)+d_G(v_2,t')$.
    The existence of a $(W_X,W_Y)$\dagring implies that this check gives will give a positive answer for every choice of $e_1,e_2,$ and $(t,t')$.
    Conversely, if this check failed it would mean that there exists no $(W_X,W_X)\dagring$ in $G$ at all.
\end{proof}

\subsection{Exhaustive ring decompositions}

We are now ready to define ring decompositions,
which are used in order to define areas where we can tame the winding number of solutions, as explained in~\Cref{sec:techniques}.

\begin{definition}
\label{def:ring:ring-decomp}
    Let $(G,\tcal)$ be a nice instance of \pdsp
    and $\widehat{T}$ be a \terset of $\tcal$.
    A collection of tuples $\mathcal{R} = (X^i,Y^i,\gamma_{1}^{i},\gamma_{2}^{i})_{i=1}^\ell$
    is called a {\em ring decomposition} of  $(G,\tcal,\widehat{T})$ when the following hold.
    \begin{enumerate}
        \item For each $i \in [\ell]$,$(X^i,Y^i)$ is a splitting partition of $\widehat{T}$ and  $(\gamma_{1}^{i},\gamma_{2}^{i})$ represents a $(X^i,Y^i)$\dagring.
        \item For each distinct $i,j \in [\ell]$ the sets $\ring(\gamma_{1}^i, \gamma_{2}^i)$ and $\ring(\gamma_{1}^j, \gamma_{2}^j)$ are disjoint.
    \end{enumerate}
    The \emph{family of partitions} of $\mathcal{R}$ is the set $\{(X^i,Y^i)\colon i\in [\ell]\}$,
    which we denote by $\Pi_{\mathcal{R}}$,
    while the \emph{family of dag-cuts} of $\mathcal{R}$ is the set
    $\bigcup_{i=1}^{\ell}\{\gamma_{1}^i,\gamma_{2}^{i}\}$,
    which we denote by $\Gamma_{\mathcal{R}}$.
    
    A tuple $(X,Y,\gamma,\gamma')$ is an extension of $\mathcal{R}$ if $\mathcal{R} \cup \{(X,Y,\gamma,\gamma')\}$ is also a ring decomposition of $(G,\tcal,\widehat{T})$.
    We call a ring decomposition {\em exhaustive} if it admits no extension.
\end{definition}

\paragraph{Computing exhaustive ring decompositions.}
We next show how to compute extensions of ring decompositions. For this, we will use some facts on crossing of closed curves in the plane.
The next observation follows directly from the fact that the sets $\ring(\gamma_{1}^i, \gamma_{2}^i)$ and $\ring(\gamma_{1}^j, \gamma_{2}^j)$ are asked to be disjoint for each distinct $i,j \in [\ell]$ in~\Cref{def:ring:ring-decomp}, together with the fact that each ring of the form $\ring(\gamma,\gamma')$ is an open set and $\mathrm{Cycle}(\gamma)$ and $\mathrm{Cycle}(\gamma')$ are its two boundaries.

\begin{observation}
\label{obs:ring:noncross}
    Let $(G,\tcal)$ be a nice instance of \pdsp
    and $\widehat{T}$ be a \terset of $\tcal$.
    Also, let $\mathcal{R}$ be a ring decomposition of $(G,\tcal,\widehat{T})$.
    Then for every two $\gamma,\gamma'\in\Gamma_{\mathcal{R}}$,
    $\mathrm{Cycle}(\gamma)$ and $\mathrm{Cycle}(\gamma')$ do not cross (as closed curves on the plane).
\end{observation}

We also use the following known fact for simple closed curves on the plane.

\begin{lemma}[Folklore]\label{lem:ring:crossing-curves}
    Let $S_1, S_2$ be two simple closed curves on the plane.
    Suppose there exist points $x_1,x_2,x_3,x_4 \in \rr^2$ such that $S_1$ separates $\{x_1,x_2\}$ from $\{x_3,x_4\}$ whereas $S_2$ separates $\{x_1,x_3\}$ from $\{x_2,x_4\}$.
    Then $S_1, S_2$ cross.
\end{lemma}

The next lemma shows that the computation of an extension of a ring decomposition for some fixed splitting partition $(X,Y)$ boils down to finding $(W_X,W_Y)$\dagcuts for some vertex supersets $W_X,W_Y$
of $X$ and $Y$, respectively. Moreover, such sets $W_X,W_Y$ can be computed in polynomial time.

\begin{lemma}
\label{lem:ring:w}
    There is a polynomial-time algorithm that, given a ring decomposition $\mathcal{R}$ of $(G,\tcal,\widehat{T})$ and a splitting partition $(X,Y)$ of $\widehat{T}$,
    such that neither $(X,Y)$ nor $(Y,X)$ belongs to the family $\Pi_{\mathcal{R}}$ of partitions of $\mathcal{R}$, 
    outputs one of the following:
    \begin{itemize}[nosep]
        \item either a correct report that there are no $(X,Y)$\dagcuts $\gamma_{1}$ and $\gamma_{2}$ such that $(X,Y,\gamma_{1},\gamma_{2})$ is an extension of $\mathcal{R}$, or
        \item two disjoint sets $W_X, W_Y \sub V(G)$ with $X\subseteq W_X$ and $Y\subseteq W_Y$ such that for every pair $(\gamma_{1},\gamma_{2})$ of $(X,Y)$\dagcuts, $(X,Y,\gamma_{1},\gamma_{2})$ is an extension of $\mathcal{R}$ if and only if $\gamma_{1},\gamma_{2}$ are $(W_X,W_Y)$\dagcuts.
    \end{itemize} 
\end{lemma}
\begin{proof}
    We use $(X^i,Y^i,\gamma_{1}^{i},\gamma_{2}^{i})_{i=1}^\ell$ to denote the ring decomposition $\mathcal{R}$.
    We first check how $(X,Y)$ relates to the other partitions in $\Pi_{\mathcal{R}}$.

    If $(X,Y)$ ``crosses'' some partition $(X',Y')\in \Pi_{\mathcal{R}}$, i.e., both $X$ and $Y$ contain vertices of both $X'$ and $Y'$, then we report that no extension of $\mathcal{R}$ with this partition $(X,Y)$ is possible.
    To see this, let $v_{X,X'}, v_{X,Y'},v_{Y,X'},v_{Y,Y'}$ be vertices of $V(G)$ that are in $X\cap X', X\cap Y',Y\cap X',$ and $Y\cap Y'$, respectively. 
    Also, let 
    $\gamma_{1}'$ and $\gamma_{2}'$ be two $(X',Y')$\dagcuts such that $(X',Y',\gamma_{1}',\gamma_{2}')\in\mathcal{R}$ and suppose that there are $(X,Y)$\dagcuts $\gamma_{1}$ and $\gamma_{2}$ such that $(X,Y,\gamma_{1},\gamma_{2})$ is an extension of $\mathcal{R}$. Then, observe that
    $\mathrm{Cycle}(\gamma_{1}')$  separates $\{v_{X,X'},v_{Y,X'}\}$ and $\{v_{X,Y'},v_{Y,Y'}\}$ while $\mathrm{Cycle}(\gamma_{1})$ separates $\{v_{X,X'},v_{X,Y'}\}$ and $\{v_{Y,X'},v_{Y,Y'}\}$. 
    By~\Cref{lem:ring:crossing-curves},
    $\mathrm{Cycle}(\gamma_{1}')$ and $\mathrm{Cycle}(\gamma_{1})$ cross (as closed curves on the plane), a contradiction to~\Cref{obs:ring:noncross}. Therefore, we can safely report that, in this case, no extension of $\mathcal{R}$ with this partition $(X,Y)$ is possible.

    From now on, we can assume that for every $i\in[\ell]$, 
    either $X^i$ or $Y^i$ \textsl{does not} intersect both $X$ and $Y$.
    In other words,
    one of the following holds: (a) $X^i\subseteq X$ or (b) $Y^i\subseteq Y$ or (c) $Y^i\subseteq X$ or (d) $X^i\subseteq Y$; the cases (b) and (d) are obtained from (a) and (c) by swapping the roles of $X^i$ and $Y^i$. 
    Keep in mind that since $(X,Y),(Y,X)\notin\Pi_{\mathcal{R}}$,
    it cannot happen that $X=X^i$ or $X=Y^i$, for any $i\in[\ell]$; thus, exactly one of the cases among (a)-(d) occurs for each $i\in[\ell]$.
    For every $i\in[\ell]$, let
    \begin{align*}
        Z_X^{i}  :=
        \begin{cases}
        X^i, & \text{if $X^i \subseteq X$,}\\
        Y^i, & \text{if $Y^i\subseteq X$,}\\
        \emptyset, & \text{otherwise,}
        \end{cases}
        &  & \text{and} & & Z_Y^{i}  :=\begin{cases}
        X^i, & \text{if $X^i \subseteq Y$,}\\
        Y^i, & \text{if $Y^i\subseteq Y$,}\\
        \emptyset, & \text{otherwise.}
    \end{cases}
    \end{align*}
    Intuitively, $Z_X^i$ corresponds to the part of the partition $(X^i,Y^i)$ that is contained in $X$ (if such exists), while $Z_Y^{i}$ corresponds to the part that is contained in $Y$ (if such exists).
    
    Also, for every $i\in[\ell]$, if $Z_{X}^i\neq\emptyset$ (resp. $Z_{Y}^i\neq \emptyset$), we set
    $B_X^i$ (resp. $B_Y^i$) to be the set of all vertices of $G$ that are in the same connected component of $\mathbb{R}^2\setminus \ring(\gamma_{1}^i,\gamma_{2}^i)$ as $Z_X^i$ (resp. $Z_Y^i$); if $Z_X^i$ or $Z_Y^i$ is the empty set then we also set $B_X^i$ and $B_Y^i$ to be the empty set, respectively. 
    We finally set $W_X:=X\cup\bigcup_{i=1}^{\ell}B_{X}^i$ and $W_Y:=Y\cup\bigcup_{i=1}^{\ell}B_{Y}^i$.
    It is easy to check that all above can be computed in polynomial time.
    
    We now show that $(X,Y,\gamma_{1},\gamma_{2})$ is an extension of $\mathcal{R}$ if and only if $\gamma_{1},\gamma_{2}$ are $(W_X,W_Y)$\dagcuts.
    For this, it suffices to show that $\ring(\gamma_{1},\gamma_{2})$ is disjoint from the other rings $\ring(\gamma_{1}^i,\gamma_{2}^i)$, $i\in[\ell]$
    if and only if $\gamma_{1}$ and $\gamma_{2}$ are $(W_X,W_Y)$\dagcuts. For every $i\in[\ell]$, let $R_i$ denote $\ring(\gamma_{1}^i,\gamma_{2}^i)$.

    For the forward direction, i.e., that disjointedness of rings implies that $\gamma_{1},\gamma_{2}$ are $(W_X,W_Y)$\dagcuts, we argue as follows.
    Let $\Delta_{X}$ (resp. $\Delta_{Y}$) denote the component of $\mathbb{R}^\mathsf{2}\setminus \mathsf{cl}(\ring(\gamma_{1},\gamma_{2}))$
    that contains $X$ (resp. $Y$).
    Keep in mind that $\Delta_{X}$ and $\Delta_{Y}$ are open sets.
    The fact that for every $i\in[\ell]$, $R_i\cap \ring(\gamma_{1},\gamma_{2})=\emptyset$ implies that every $R_i,i\in[\ell]$ is either a subset of $\Delta_{X}$ or a subset of $\Delta_{Y}$.
    For each $i\in[\ell]$, let $\Delta^i$ be the connected component of $\mathbb{R}^2\setminus R_i$ that is disjoint from $\ring(\gamma_{1},\gamma_{2})$.
    Now notice that either $\Delta^{i}\cap \widehat{T}\subseteq X$ or $\Delta^{i}\cap \widehat{T}\subseteq Y$,
    which implies that either $\Delta^{i}\subseteq \Delta_{X}$ or $\Delta^{i}\subseteq \Delta_{Y}$, respectively.
    Note also that for every $i\in[\ell]$, $B_{X}^i\subseteq \Delta^i$ and $B_{Y}^i\subseteq \Delta^i$; from this we get $W_X\subseteq \Delta_{X}$ and $W_Y\subseteq \Delta_{Y}$. Therefore, if $(U_X,U_{\mathsf{mid}},U_Y)$
    is the $(X,Y)$\dagring represented by $(\gamma_{1},\gamma_{2})$, then $W_X\subseteq U_X$ and $W_Y\subseteq U_Y$.
    This, in turn, implies that 
    $(\gamma_{1},\gamma_{2})$ are also $(W_X,W_Y)$\dagcuts.

    In the other direction,  suppose towards a contradiction, that $\gamma_{1},\gamma_{2}$ are $(W_X,W_Y)$\dagcuts and there is $i\in[\ell]$ such that $R_i\cap \ring(\gamma_{1},\gamma_{2})\neq \emptyset$.
    Since the boundaries of $R_i$ and $\ring(\gamma_{1},\gamma_{2})$ correspond to cycles of the dual $G^\star$,
    there is a vertex of $G$ that belongs to $R_i$ that also belongs to $\ring(\gamma_{1},\gamma_{2})$.
    Consequently, either $\mathrm{Cycle}(\gamma_1^i)$ or $\mathrm{Cycle}(\gamma_2^i)$ crosses $\mathrm{Cycle}(\gamma_1)$ or $\mathrm{Cycle}(\gamma_2)$.
    This implies that
    there is a vertex of $W_X\cup W_Y$ in $\ring(\gamma_{1},\gamma_{2})$.
    This yields a contradiction to the fact that $\mathrm{Cycle}(\gamma_{2})$ and  $\mathrm{Cycle}(\gamma_{1})$ separate $W_X$ and $W_Y$, which holds because of~\Cref{lem:ring:cut-cycle}.
    Therefore, $\mathcal{R} \cup \{(X,Y,\gamma_{1},\gamma_{2})\}$ is a ring decomposition of $(G,\tcal,U)$.
\end{proof}

With~\Cref{lem:ring:w} in hand, we are now ready to show the main result of this subsection, i.e., that given a terminal-superset $\widehat{T}$, an exhaustive ring decomposition can be computed in time $2^{|\widehat{T}|}\cdot n^{\mathcal{O}(1)}$. Moreover, this ring decomposition should consist of $\Oh(|\widehat{T}|)$ many rings.

\begin{proposition}
\label{prop:ring:ring-decomp}
    There exists an algorithm that, given a nice instance $(G,\tcal)$ of \pdsp 
    and a \terset  $\widehat{T}$ of $\tcal$
    with $|\widehat{T}| = r$, runs in time $2^r\cdot n^{\Oh(1)}$, and outputs an exhaustive ring decomposition of $(G,\tcal,\widehat{T})$ containing at most $2r-2$ tuples.
\end{proposition}
\begin{proof}
    We will compute a sequence $(\mathcal{R}_0, \mathcal{R}_1, \dots, \mathcal{R}_\ell)$ of ring decompositions of $(G,\tcal,\widehat{T})$ with the following property.
    For each $i \in [\ell]$ the ring decomposition $\mathcal{R}_i$ comprises $i$ tuples and it is obtained from $\mathcal{R}_{i-1}$ by inserting a tuple $(X^i,Y^i,\gamma_{1}^i,\gamma_{2}^i)$
    which is an extension of $\mathcal{R}_{i-1}$ and maximizes $|\ring(\gamma_{1}^i,\gamma_{2}^i)\cap V(G)|$ (breaking ties arbitrarily).
    We start from an empty ring decomposition $\mathcal{R}_0$ and terminate the process when no more extension is available.
    Then clearly $\mathcal{R}_\ell$ is exhaustive.

    \begin{claim}\label{claim:ring:different}
        For every $i\in[\ell]$, if $(X^i,Y^i,\gamma_{1}^i,\gamma_{2}^i)$ is an extension of $\mathcal{R}_{i-1}$, then neither $(X^i,Y^i)$ nor $(Y^i,X^i)$ belongs to $\Pi_{\mathcal{R}_{i-1}}$. 
    \end{claim}
    \begin{innerproof}
        For every $h\le i$, let $R_h$ denote $\ring(\gamma_{1}^h,\gamma_{2}^h)$.
        Suppose towards a contradiction that either $(X^i,Y^i)$ or $(Y^i,X^i)$ belongs to $\Pi_{\mathcal{R}_{i-1}}$ and let $j<i$ be minimum possible so that $(X^i,Y^i) = (X^j,Y^j)$ or $(X^i,Y^i)=(Y^j,X^j)$.
        Since  $(X^i,Y^i,\gamma_{1}^i,\gamma_{2}^i)$ is an extension of $\mathcal{R}_{i-1}$, we have that $R_i$ and $R_h$ are disjoint, for every $h<i$. Therefore, w.l.o.g. we can assume that $\gamma_{1}^i\sqsubseteq \gamma_{2}^i\sqsubseteq \gamma_{1}^j\sqsubseteq \gamma_{2}^j$ and thus $X^i=X^j$ and $Y^i=Y^j$.
        Note that since $(X^i,Y^i)=(X^j,Y^j)$, $\gamma_{1}^i$ is also  an $(X^j,Y^j)$\dagcut.
        Therefore, due to~\Cref{lem:maxring},
        we get that $(\gamma_{1}^i,\gamma_{2}^j)$ also represents an $(X^j,Y^j)$\dagring.
    
        We next claim that $\ring(\gamma_{1}^i, \gamma_{2}^j)$ is disjoint from the ring $R_{j'}$, for every $j'< j$.
        To see why this holds, suppose towards a contradiction that there is a $j'<j$ such that $R_{j'}$ intersects
        $\ring(\gamma_{1}^i, \gamma_{2}^j)$.
        Due to~\Cref{obs:ring:noncross}, 
        $R_{j'}$ does not intersect $\mathrm{Cycle}(\gamma_{1}^i)$ or $\mathrm{Cycle}(\gamma_{2}^j)$ and therefore $R_{j'}$ is a subset of $\ring(\gamma_{1}^i, \gamma_{2}^j)$.
        Also, since $\gamma_{1}^{j'},\gamma_{2}^{j'}$ are $(X^{j'},Y^{j'})$\dagcuts for some splitting partition $(X^{j'},Y^{j'})$, both components of $\mathbb{R}^2\setminus \mathsf{cl}(R_{j'})$ contain at least one terminal. This, in turn, implies that the sets $X^i$ and $Y^j$ belong to different components of $\mathbb{R}^2\setminus \mathsf{cl}(R_{j'})$ and therefore either $X^{j'}=X^i$ (and $Y^{j'}=Y^{j}=Y^i$) or $X^{j'}=Y^j=Y^i$ (and $Y^{j'}=X^i$),
        a contradiction to minimality of $j$.

        Given that $(\gamma_{1}^i,\gamma_{2}^j)$ represents an $(X^j,Y^j)$\dagring and  
        $\ring(\gamma_{1}^i, \gamma_{2}^j)$ is disjoint from the rings $R_{j'}$, for each $j'< j$, we get that $(X^j,Y^j,\gamma_{1}^i, \gamma_{2}^j)$ is also an extension of $\mathcal{R}_{j-1}$.
        Moreover, since both $R_i$ and $R_j$ are subsets of $\ring(\gamma_{1}^i,\gamma_{2}^j)$, we have $|R_j\cap V(G)|< |\ring(\gamma_{1}^i,\gamma_{2}^j)\cap V(G)|$.
        This contradicts the fact that $(X^j,Y^j,\gamma_{1}^j,\gamma_{2}^j)$ is an extension of $\mathcal{R}_{j-1}$ that maximizes $|\ring(\gamma_{1}^j,\gamma_{2}^j)\cap V(G)|$.
    \end{innerproof}

    \begin{claim}\label{claim:ring:rd_time}
        Given  $\mathcal{R}_{i-1}$ we can compute  $\mathcal{R}_{i}$ in time $2^r\cdot n^{\Oh(1)}$.
    \end{claim}
    \begin{innerproof}
        The algorithm iterates over all splitting partitions $(X,Y)$ of $\widehat{T}$ that are different from $(X^j,Y^j)$ and $(Y^j,X^j)$ for every $j<i$; there are less than $2^r$ many of them. Because of~\Cref{claim:ring:different}, we know that it is sufficient to focus only on partitions that are different from $(X^j,Y^j)$ and $(Y^j,X^j)$ for every $j<i$. For a fixed $(X,Y)$, our goal is to compute two $(X,Y)$\dagcuts $\gamma_{1}$ and $\gamma_{2}$ such that $(X,Y,\gamma_{1},\gamma_{2})$ is an extension of $\mathcal{R}$ while maximizing $|\ring(\gamma_{1},\gamma_{2})\cap V(G)|$, or report that no such extension (with $(X,Y)$) exists. Invoking the algorithm of ~\Cref{lem:ring:w} for $(X,Y)$,    we either correctly report that no extension with $(X,Y)$ is possible, or obtain two disjoint sets $W_X,W_Y\subseteq V(G)$ with $X\subseteq W_X$ and $Y\subseteq W_Y$ such that for every pair $(\gamma_{1},\gamma_{2})$ of $(X,Y)$\dagcuts, $(X,Y,\gamma_{1},\gamma_{2})$ is an extension of $\mathcal{R}$ if and only if $\gamma_{1},\gamma_{2}$ are $(W_X,W_Y)$\dagcuts. Note also that, by definition, a maximal $(W_X,W_Y)$\dagring represented by two $(W_X,W_Y)$\dagcuts $\gamma_{1}',\gamma_{2}'$ is also maximizing $|\ring(\gamma_{1}',\gamma_{2}')\cap V(G)|$. Due to~\Cref{lem:ring:max-ring} such a maximal $(W_X,W_Y)$\dagring can be found in polynomial time.
    \end{innerproof}

    \begin{claim}\label{claim:ring:rd_nb}
        Suppose that the process terminates at iteration $\ell$.
        Then $\ell \le 2r-2$.
    \end{claim}
    \begin{innerproof}
        For every $i\in[\ell]$,
        let $D_i$ be the bounded component of $\mathbb{R}^2\setminus \mathsf{cl}(\ring(\gamma_{1}^i,\gamma_{2}^i))$. Let also $K$ denote the set $\mathbb{R}^2\setminus \bigcup_{i=1}^{\ell}\mathsf{cl}(\ring(\gamma_{1}^i,\gamma_{2}^i))$ and let $D_0$ be the component of $K$ that is not contained in any $D_i$, $i\in[\ell]$. Since for  each distinct $i,j \in [\ell]$ the rings $\ring(\gamma_{1}^i, \gamma_{2}^i)$ and $\ring(\gamma_{1}^j, \gamma_{2}^j)$ are disjoint, we have that if $D_i\cap D_j\neq \emptyset$, then either $D_i\subseteq D_j$ or $D_j\subseteq D_i$.

        We consider a graph $H'$ that has one distinguished vertex $v_0$ corresponding to $D_0$ and a vertex $v_i$ for each set $D_i,i\in[\ell]$.
        Two vertices $v_i,v_j\in V(H)$ for distinct $i,j\in[\ell]$ are adjacent if $D_i\subseteq D_j$ or $D_j\subseteq D_i$, but there is no $h\in[\ell]\setminus\{i,j\}$ such that $D_i \subseteq D_{h}\subseteq D_j$ or $D_j\subseteq D_h\subseteq D_i$, while $v_0$ is adjacent to all vertices $v_i,i\in[\ell]$ for which there is no $h\neq i$ such that $D_i\subseteq D_h$.
        Note that $H$ can be viewed as a tree rooted at $v_0$; every other vertex $v_i$ has exactly one parent, that is the vertex $v_j$ corresponding to the set $D_j$ for which it holds that there is no other $h$ different than $i$ and $j$ so that $D_i\subseteq D_h \subseteq D_j$.

        Let $\ell_1$ denote the number of leaves of $H$ and let $\ell_2$ denote the number of vertices of degree two of $H$.
        We claim that $\ell_1+\ell_2\le r$. Indeed, the leaves of $H$ correspond to disjoint sets $D_i$, each containing at least one element from $\widehat{T}$.
        Moreover, for every vertex $v_i$ of $H$ that has exactly one child $v_j$, we get that $D_i\setminus D_j$ intersects $\widehat{T}$. To see why this holds, observe that if $D_i\setminus D_j$ does not intersect $\widehat{T}$, then $D_i\cap \widehat{T}=D_j\cap \widehat{T}$ and therefore $(X_i,Y_i)$ is equal to either $(X_j,Y_j)$ or $(Y_j,X_j)$, a contradiction to~\Cref{claim:ring:different}.
        In summary, we have $\ell_1+\ell_2$ disjoint sets, each of which contains some element from  $\widehat{T}$.
        
        Therefore, $\ell_1+\ell_2\le |\widehat{T}|=r$. Moreover, using standard arguments, it is easy to see that the number of vertices of degree more than two in $H$ is upper-bounded by $\ell_1-1$ and therefore $|V(H)|\le 2\ell_1+\ell_2-1\le 2r-1$, which implies that $\ell\le 2r-2$.
    \end{innerproof}

    Because of~\Cref{claim:ring:rd_time}, the overall algorithm runs in time $2^{r}\cdot n^{\Oh(1)}$ and because of~\Cref{claim:ring:rd_nb} the number of tuples in the obtained exhaustive ring decomposition is at most $2r-2$.
\end{proof}

\subsection{Winding creates rings}

In this subsection, we show that if we find a candidate solution that winds around a geodesic many times, then this situation gives rise to a dag-ring. This result will be applied in the presence of an exhaustive ring decomposition, in order to rule out large winding outside rings (see~\Cref{prop:outside}).
In order to state this result, we first define (winding) handles of paths.

\begin{definition}
   Let $P, Q$ be paths in a plane graph $G$. We say that $P$ is a {\em handle} of $Q$ if the endpoints of $P$ are internal vertices of $Q$ and $P$ is internally disjoint from $Q$.

   When $P$ is a handle of $Q$, we denote by $C(Q,P)$ the unique cycle contained in $Q \cup P$.
   We say that $P$ is a {\em winding handle} of $Q$ if the endpoints of $Q$ lie in different connected components of $\rr^2 \sm C(Q,P)$.
   Otherwise $P$ is called a {\em regular handle} of $Q$.
\end{definition}

\begin{proposition}\label{prop:winding:dag-ring}
    Let $(G,\tcal)$ be a nice instance of \pdsp, $\widehat{T}$ be a \terset of $\tcal$, and $Q$ be a geodesic in $G$.
    Also, let $P$ be an oriented shortest $(x,y)$-path for some $(x,y) \in \tcal$ and suppose that there exist $h$ disjoint subpaths $P_1, P_2, \dots, P_h$ of $P$, where $h > 10 \cdot (|\widehat{T}| + 1)$, which are located on $P$ in this order and each of them forms a winding handle of $Q$.
    Then there exist:
    \begin{enumerate}[nosep]
        \item indices $\ell, r \in [h]$ such that $\ring(C(Q,P_\ell), C(Q,P_r))$ does not contain any vertices from $\widehat{T}$,
        \item a splitting partition $(X,Y)$ of $\widehat{T}$ with $x \in X$, $y \in Y$, and
        \item a pair $(\gamma, \gamma')$ of $(X,Y)$\dagcuts with
        $\ring(\gamma, \gamma') \sub \ring(C(Q,P_\ell), C(Q,P_r))$
        that represents an $(X,Y)$\dagring.
    \end{enumerate}
\end{proposition}
\begin{proof}
    We present the proof via a series of claims that reveal the structure of the subgraph $Q \cup P$ gradually.
    This allows us to accumulate notation that is used only within this proof.

    For $i \in [h]$ let us denote the two vertices of $V(P_i) \cap V(Q)$ as $v'_i, v''_i$, ordered in such a way that $v'_i <_P v''_i$ (see \Cref{fig:winding-middle}).
    Note that $v''_i <_P v'_{i+1}$.
    By \Cref{lem:prelim:monotone3} we can choose an orientation of $Q$ so that $v_1' <_Q v_1'' <_Q <v_2' <_Q v_2'' <_Q \dots <_Q v_h''$ (that is, the relations $<_Q$ and  $<_P$ coincide on this vertex set).
    Let $x_q, y_q$ be the first and the last vertex of $Q$ with respect to this orientation.

    By the definition of a winding handle, for each $i \in [h]$ the cycle $C(Q,P_i)$ separates the plane $\rr^2$ into two connected components, one containing $x_q$ and one with $y_q$.
    Let $R_i \sub \rr^2$ denote the connected component of $\rr^2 \sm C(Q,P_i)$ containing $x_q$.

    \begin{claim}\label{claim:winding:containment}
        For each $i \in [h-1]$ it holds that $R_i \sub R_{i+1}$.
    \end{claim}
    \begin{innerproof}
        Because $v_i'' <_Q  v_{i+1}'$ and $v_i'' <_P  v_{i+1}'$ we infer that $C(Q,P_i)$ and $C(Q,P_{i+1})$ are disjoint.
        Next, $Q[x_q,v_i']$ is disjoint from $C(Q,P_{i+1})$ so $v_i' \in R_{i+1}$ and consequently the entire cycle $C(Q,P_i)$ is contained in $R_{i+1}$.
        This means that one of the connected components of $\rr^2 \sm C(Q,P_i)$ must be contained in $R_{i+1}$ as well.
        This component cannot contain $y_q$ (which lies outside $R_{i+1}$) so it must be the one containing $x_q$, that is, $R_i$.
    \end{innerproof}

    Note that $\cl(R_i) = R_i \cup C(Q,P_i)$ and
   that $V(C(Q,P_{i}))$ forms a $(V(C(Q,P_{i-1})), V(C(Q,P_{i+1})))$-separator
   for every $i \in [2,h-1]$.
    Next, we establish a crucial property of $Q \cup P$ that will facilitate applying rerouting arguments.

    \begin{claim}\label{claim:winding:middle}
        Let $i,j \in [h]$ be such that $j \ge i + 2$
        and $u \in V(C(Q,P_i))$, $w \in V(C(Q,P_{j}))$.
        Then there exists a shortest $(u,w)$-path in $G$ that goes through $v_i''$ and $v_{j}'$.
    \end{claim}
    \begin{innerproof}
        First suppose that both $u,w$ belong to $V(P)$.
        By \Cref{lem:prelim:subpath}, the subpath $P[u,w]$ is a shortest $(u,w)$-path.
        We have $u \le_P v_i'' <_P v_j' \le_P w$
        so $P[u,w]$ is the sought path. 
        The case when $u,w \in V(Q)$ is analogous.

        It remains to consider $u \in V(P)$, $w \in V(Q)$; again the symmetric case is analogous.
        Let $J$ be some shortest  $(u,w)$-path.
        Note that $V(C(Q,P_{i+1}))$ is a $(u,w)$-separator
        so there is $z \in V(J) \cap V(C(Q,P_{i+1}))$.

        Case (1): $z \in V(P)$. 
        By \Cref{lem:prelim:replacement} the concatenation $J' = P[u,z] + J[z,w]$ is a shortest $(u,w)$-path.
        But $u \le_P v_i'' <_P z$ so 
        $v_i'' \in V(J')$.
        Next, we take advantage of the fact that $v_i'' \in V(Q)$ and $v_i'' <_Q v_j' \le_Q w$.
        The concatenation $J'[u, v_i''] + Q[v_i'',w]$
         is a shortest $(u,w)$-path the goes through both $v_i''$ and $v_j'$.
         See \Cref{fig:winding-middle} for an illustration.

        Case (2): $z \in V(Q)$. 
        Similarly as above, we obtain that $J'' = J[u,z] + Q[z,w]$ is a shortest $(u,w)$-path.
        We have $z <_Q v_{j}' \le_Q w$ so $v_{j}' \in V(J'')$.
        Finally, we consider the geodesic $P[u,v_j'] + J''[v_j',w]$ which visits $v_i'' \in V(P[u,v_j'])$.
        The claim follows. 
    \end{innerproof}

\begin{figure}[t]
    \centering
    \includegraphics[scale=0.7]{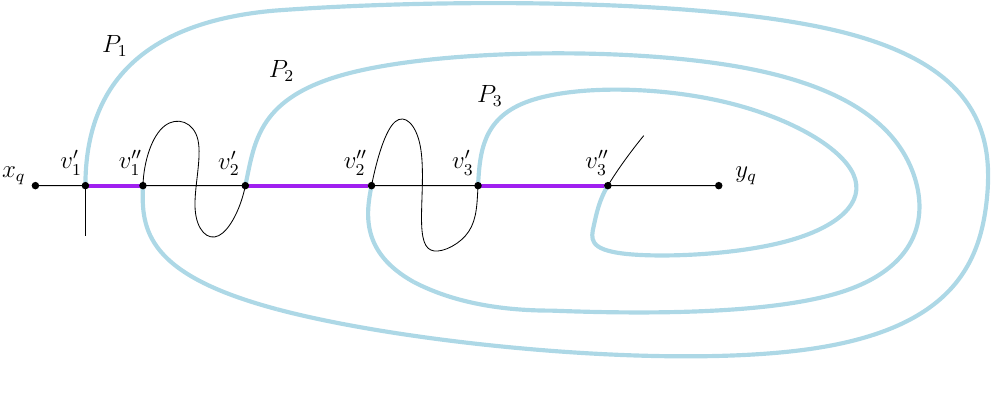}
    \includegraphics[scale=0.7]{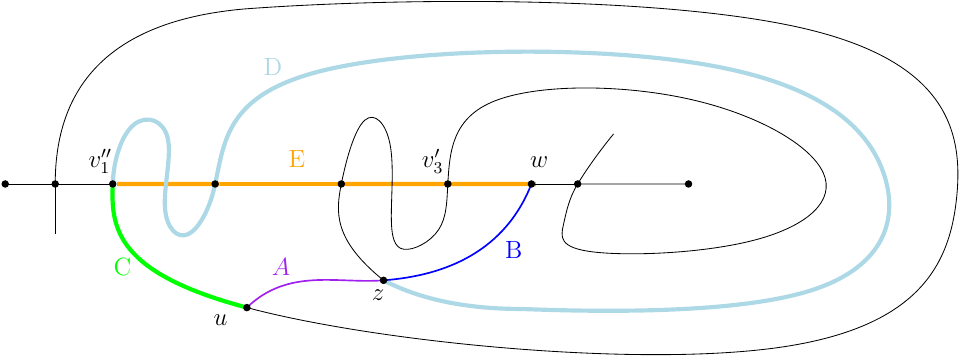}
    \caption{
    Top: The path $Q$ is horizontal while the path $P$ winds around $y_q$.
    The subpaths of $P$ highlighted in blue are the winding handles $P_1,P_2,P_3$.
    Together with the highlighted subpaths of $Q$ they form cycles $C(Q,P_i)$.
    Bottom: An illustration to \Cref{claim:winding:middle}, Case (1). Here $i = 1,\, j = 3$ and the shortest $(u,w)$-path $J$ is depicted as $A+B$.
    Both $A$ and $C+D$ are shortest $(u,z)$-paths, so $C+D+B$ is also a shortest $(u,w)$-path.
    Next, $E$ is a shortest $(v_1'',w)$-path so $C+E$ is a shortest $(u,w)$-path that satisfies the claim.}
    \label{fig:winding-middle}
\end{figure}

    As a consequence of this claim, we show that shortest paths cannot go ``back and forth''. 
    \begin{claim}\label{claim:winding:noreturn}
        Let $i \in [h-2]$ and $u,w \in V(C(Q,P_i))$.
        If $J$ is a shortest $(u,w)$-path then $V(J) \cap V(C(Q,P_{i+2})) = \emptyset$.
        Similarly, when $u,w \in V(C(Q,P_{i+2}))$ then 
        $V(J) \cap V(C(Q,P_{i})) = \emptyset$.
    \end{claim}
    \begin{innerproof}
        Consider the first case.
        Suppose for the sake of contradiction that 
        there exists $z \in V(J) \cap V(C(Q,P_{i+2}))$.
        Clearly,  $z \ne v_i''$ because the latter vertex lies on $C(Q,P_{i})$.
        By \Cref{claim:winding:middle}, there exists a shortest $(u,z)$-path that goes through $v_i''$.
        Because the edge weights are positive, $d_G(u,z)$ must be strictly greater than $d_G(u,v_i'')$.
        Analogously, it holds that $d_G(w,z) > d_G(w,v_i'')$.
        Since $z$ lies on a  shortest $(u,w)$-path, we infer that $d_G(u,w) = d_G(u,z) + d_G(w,z) > d_G(u,v_i'') + d_G(w,v_i'')$.
        This contradicts the triangle inequality and therefore such $z$ cannot exists.
        
        The case when $u,w \in V(C(Q,P_{i+2}))$ and $z \in V(J) \cap V(C(Q,P_{i}))$ is analogous but this time we exploit the fact that the path $J$ goes through $v_{i+2}'$, again because of~\Cref{claim:winding:middle}.
    \end{innerproof}

    Next, we
    identify an interval $[\ell,r] \sub [1,h]$ for which $R_r \sm R_\ell$ does not contain any elements from $\widehat{T}$.

    \begin{claim}\label{claim:winding:interval}
        There exists a partition $(X,Y)$ of $\widehat{T}$ with $x \in X$, $y \in Y$ and an interval $[\ell, r] \sub [h]$ such that $r - \ell = 10$ and $R_i \cap \widehat{T} = X $  for each $i \in [\ell, r]$.
    \end{claim}
    \begin{innerproof}
        By \Cref{claim:winding:containment} we have $R_i \sub R_j$ for any $1 \le i < j \le h$  so $|R_i \cap \widehat{T}| \le |R_j \cap \widehat{T}|$.
    The sequence $a_i = |R_i \cap \widehat{T}|$ can increase at most $|\widehat{T}|$ times.
    Since $h > 10 \cdot (|\widehat{T}| + 1)$ there exist indices $\ell, r \in [h]$ such that $a_\ell = a_r$ and $r = \ell + 10$.
    Then $R_i \cap \widehat{T}$ coincides for each $i \in [\ell, r]$.
    We set $X = R_r \cap \widehat{T}$, $Y = \widehat{T} \sm X$.

    Now we argue that $x \in X$.
    If $x \not\in X$ then $x \not\in R_r$ and the subpath 
    $P[x,v_\ell']$ would have to intersect $C(Q,P_r)$.
    Suppose there exists $z \in V(P[x,v_\ell']) \cap V(C(Q,P_r))$.
    Then  $x \le_P z <_P v_\ell'$ but on the other hand $v_\ell' <_P v_r' \le_P z$ what yields a contradiction.
    The argument that $y \in Y$ is symmetric.
    \end{innerproof}  

    Next, we observe that when two vertices are either both ``outside'' or both ``inside'' then no shortest paths between them can intersect $C(Q,P_{\ell+2})$ or $C(Q,P_{r-2})$. 

    \begin{claim}\label{claim:winding:sameside}
        Let $x_1,x_2 \in X$ (resp. $x_1,x_2 \in Y$) and $J$ be a shortest $(x_1,x_2)$-path.
        Then $J$ is disjoint from $C(Q,P_{\ell+2})$
        (resp. $C(Q,P_{r-2}))$.
    \end{claim}
    \begin{innerproof}
        Consider $x_1,x_2 \in X$.
        Note that $V(C(Q,P_\ell))$ is a $(X, V(C(Q,P_{\ell+2}))$-separator. If $V(J)$ intersected $V(C(Q,P_{\ell+2}))$ then $J$ would contain a subpath (still geodesic) between two vertices from $V(C(Q,P_{\ell}))$ visiting a vertex from $V(C(Q,P_{\ell+2}))$.
        This is impossible due to \Cref{claim:winding:noreturn}.
        The case with  $x_1,x_2 \in Y$ is analogous.
    \end{innerproof}

Next, we show that given two vertices lying on the ``closed ring'' between $C(Q,P_{r-2})$
and $C(Q,P_{\ell+2})$, the corresponding arc appears in the shortest paths DAGs for (oriented) shortest paths between split pairs of terminals.
    \begin{claim}\label{claim:winding:reroute}
        Let $u,w$ be vertices of $G$ which lie in $\cl(R_{r-2} \sm R_{\ell+2})$ and $x_1 \in X,\, y_1 \in Y$.
        If $(u,w) \in A((x_1,y_1)\dagg)$ then there exists a shortest $(x_1,y_1)$-path that visits $v_\ell'',u,w,v_r'$, in this order.
    \end{claim}
    \begin{innerproof}
        Consider a shortest $(x_1,y_1)$-path $J$ that goes through the edge $uw \in E(G)$ with $u <_J w$.
        The path $J[x_1,u]$ must visit a vertex from $V(C(Q,P_\ell))$ and one from $V(C(Q,P_{\ell+2}))$.
    We apply \Cref{claim:winding:middle} to this subpath to find a geodesic with the same endpoints that visits $v_\ell''$.
    By the replacement argument (\Cref{lem:prelim:replacement}) we obtain a shortest $(x_1,u)$-path $J'$ that goes through $v_\ell''$.
    We exploit the same argument to obtain that $J' + J[u,y_1]$ is a shortest $(x_1,y_1)$-path that goes through  $v_\ell'',u,w$, in this order.
    Analogously we can tweak the subpath between $w$ and $y_1$ to visit~$v_r'$.
    \end{innerproof}

    \begin{claim}\label{claim:winding:coincide}
        Let $u,w$ be vertices of $G$ which lie in  $\cl(R_{r-2} \sm R_{\ell+2})$ and
        $x_1,x_2 \in X,\, y_1,y_2 \in Y$.
        Then $(u,w) \in A((x_1,y_1)\dagg)$ if and only if $(u,w) \in A((x_2,y_2)\dagg)$.
    \end{claim}
    \begin{innerproof}
    \Cref{claim:winding:reroute} implies that when 
    $(u,w) \in A((x_1,y_1)\dagg)$ then $(u,w) \in A((v_\ell'',v_r')\dagg)$.
    The opposite implication holds as well: we can modify any shortest $(x_1,y_1)$-path to go through $v_\ell''$ and $v_r'$ (by the same claim) and then we can replace the subpath from $v_\ell''$ to $v_r'$ with any shortest $(v_\ell'', v_r')$-path.
    In summary, the condition $(u,w) \in A((x_1,y_1)\dagg)$ is equivalent to $(u,w) \in A((v_\ell'',v_r')\dagg)$ which is equivalent to $(u,w) \in A((x_2,y_2)\dagg)$.
    \end{innerproof}

    Now we orient the edges lying in the closure of $(R_{r-2} \sm R_{\ell+2})$ to be consistent with any shortest path that goes through this region, when oriented towards $Y$.
    \Cref{fig:winding-dag-ring} illustrates where particular objects are located with respect to $P_\ell, \dots, P_r$.

    \begin{claim}\label{claim:winding:dag}
        Let $U = V(G) \cap \cl(R_{r-2} \sm R_{\ell+2})$.
        We can orient the edges of $G[U]$ to obtain an acyclic digraph $D$ such that:
        \begin{enumerate}
            \item For each $(u,w) \in A(D)$ and $(x_1,y_1) \in \splitt_\tcal(X,Y)$ it holds that  $(u,w) \in A((x_1,y_1)\dagg)$.
            \label{item:winding:dag:split}
            \item For each $(u,w) \in A(D)$ and $(x_1,y_1) \in \sameside(X,Y)$ it holds that neither $(u,w)$ nor $(w,u)$ belongs to $A((x_1,y_1)\dagg)$.
            \label{item:winding:dag:sameside}
 \end{enumerate}
    \end{claim}
    \begin{innerproof}
        By \Cref{claim:winding:sameside} when $(x_1,y_1) \in \sameside(X,Y)$ then no shortest $(x_1,y_1)$-path can go through either $u$ or $w$, what will imply condition (\ref{item:winding:dag:sameside}).
        Since the instance is nice, every edge $uw \in E(G)$ belongs to some shortest $(x_1,y_1)$-path, for some $(x_1,y_1) \in \tcal$.
        This must happen for some $(x_1,y_1) \in \splitt_\tcal(X,Y)$.
        But \Cref{claim:winding:coincide} implies that then this happens for every $(x_1,y_1) \in \splitt_\tcal(X,Y)$, with the same orientation of $(u,w)$.
        Hence we can construct $D$ by considering the subgraph of $(x,y)\dagg$ (recall that $(x,y) \in \splitt_\tcal(X,Y)$) induced by $U$.
        Then $D$ is clearly acyclic as a subgraph of an acyclic digraph; see~\Cref{obs:prelim:no-dicycles}.
    \end{innerproof}

    \begin{claim}\label{claim:winding:cut}
        Let $i \in [\ell+2, r-4]$.
        The set $U$ admits a partition $(U_X,U_Y)$
        such that
        $U \cap \cl(R_i) \sub U_X$,
        $U \sm R_{i+2} \sub U_Y$
        and for each $(u,w) \in \Ev(U_X,U_Y)$
        it holds that $(u,w) \in A(D)$.
    \end{claim}
    \begin{innerproof}
        Note that the assumption $i \in [\ell+2, r-4]$ is needed to ensure that 
        $V(G) \cap \cl(R_{i+2} \sm R_i) \sub U$.
        
        Suppose that such a partition does not exists.
        Then $D$ admits a directed path from 
        $U \sm R_{i+2}$ to $U  \cap \cl(R_i)$.
        This path must connect a vertex from
        $V(C(Q,P_{i+2}))$ to one in $V(C(Q,P_i))$.
        By \Cref{claim:winding:dag}(\ref{item:winding:dag:split}) this means that there exists a shortest $(x,y)$-path $J$ that visits a vertex $z_2 \in V(C(Q,P_{i+2}))$ and then a vertex $z_3 \in V(C(Q,P_i))$.
        But before reaching $z_2$, since $x$ belongs to $R_{i}$,
        it must also visit some vertex $z_1 \in V(C(Q,P_i))$.
        Then $J[z_1,z_3]$ is a geodesic that goes through $z_2 \in V(C(Q,P_{i+2}))$.
        This is impossible due to \Cref{claim:winding:noreturn}.       
    \end{innerproof}

\begin{claim}\label{claim:winding:dagcut}
    For each $i \in [\ell+2, r-4]$
    there exists an $(X,Y)\dagcut$ $(V_X,V_Y)$ in $G$ such that
    $V(G) \cap \cl(R_i) \sub V_X$ and
        $V(G) \sm R_{i+2} \sub V_Y$.
\end{claim}
\begin{innerproof}
    We apply \Cref{claim:winding:cut} 
    to obtain a partition $(U_X,U_Y)$ of $U$
    and consider the edge set $S = E_G(U_X,U_Y)$.
    We define $V_X$ as the reachability set of $V(G) \cap \cl(R_i)$ in the graph $G \sm S$ and $V_Y = V(G) \sm V_X$.
    From \Cref{claim:winding:interval} we know that $X \sub V_X$.

    We argue that $V(G) \sm R_{i+2} \sub V_Y$, which will imply $Y \sub V_Y$.
    Suppose that this is not the case and there exists a path $J$ from some $u \in V(G) \cap \cl(R_i)$ to $w \in V(G) \sm R_{i+2}$ in $G \sm S$.
    The path $J$ must contain a subpath from $V(C(Q,P_i)) \sub U_X$  to $V(C(Q,P_{i+2})) \sub U_Y$ contained in $G[U]$.
    But such a path must enter $U_Y$ at some point and so it must use an edge from $S$; a contradiction.
    
    Next, consider some $(u,w) \in \Ev(V_X,V_Y)$.
    We have $u,w \in U$ and $(u,w) \in A(D)$.
    The definition of the digraph $D$ (\Cref{claim:winding:dag}) ensures that $(u,w) \in A((x_1,y_1)\dagg)$ when $(x_1,y_1) \in \splitt_\tcal(X,Y)$
    and $(u,w),(w,u) \not\in A((x_1,y_1)\dagg)$ when $(x_1,x_1) \in \sameside(X,Y)$.
    Hence $(V_X,V_Y)$ obeys the requirements of \Cref{def:dag-cut}.
\end{innerproof}

\begin{figure}[t]
\centering
\includegraphics[scale=0.8]{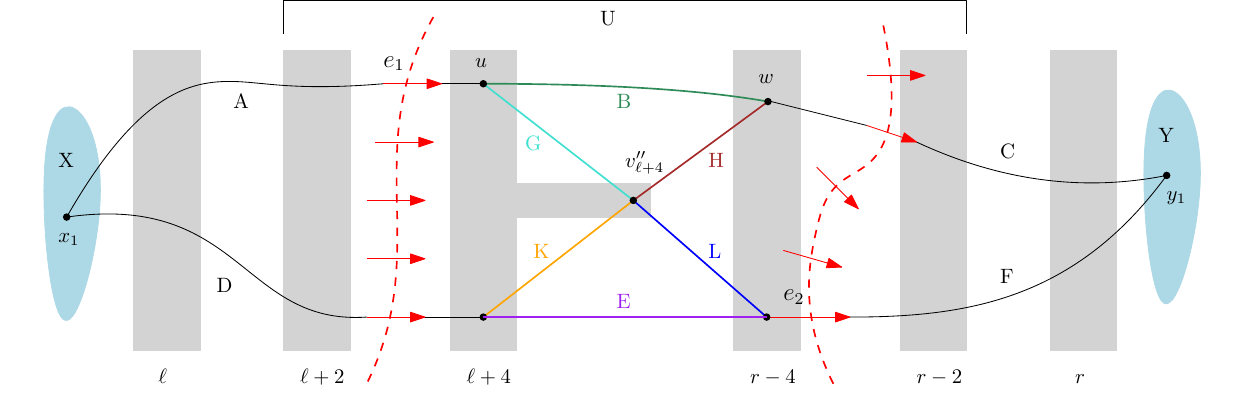}
\caption{
The final argument in the proof of \Cref{prop:winding:dag-ring}.
The embedding of $G$ is discarded here: the gray blocks depict the sets $V(C(Q,P_i))$ and the two cuts from \Cref{claim:winding:dagcut} are drawn in red.
The shortest $(x_1,y_1)$-path $J_1$ passing through $e_1$ is represented as $A+B+C$ while the path $J_2$  passing through $e_2$ is represented as $D+E+F$.
By \Cref{claim:winding:middle} the paths $A+G+H+C$ and $D+K+L+F$ are also shortest.
As a consequence, $A+G+L+F$ is a shortest $(x_1,y_1)$-path and it passes through $e_1$ and $e_2$.
}\label{fig:winding-dag-ring}
\end{figure}

    We are ready to finish the proof of the proposition.
    We apply \Cref{claim:winding:dagcut} to $i = \ell + 2$ and $i = r - 4$ and~obtain: 
    \begin{enumerate}
        \item an $(X,Y)\dagcut$ $\gamma_1 = (V^1_X,V^1_Y)$  with 
    $V(G) \cap \cl(R_{\ell+2}) \sub V^1_X$,  $V(G) \sm R_{\ell+4} \sub V^1_Y$, 
    and
        \item an $(X,Y)\dagcut$ $\gamma_2 = (V^2_X,V^2_Y)$ with 
    $V(G) \cap \cl(R_{r-4}) \sub V^2_X$, $V(G) \sm R_{r-2} \sub V^2_Y$.
    \end{enumerate}
    These cuts are depicted in \Cref{fig:winding-dag-ring}.
    We want to show that $(V^1_X, V^1_Y \cap V^2_X, V^2_Y)$ forms an $(X,Y)\dagring$.

    Since $r = \ell + 10$, the set $V^1_Y \cap V^2_X$ is non-empty (it contains $V(C(Q,P_{\ell+5}))$) and moreover 
    there can be no edges between $V^1_X$ and $V^2_Y$; this yields condition (\ref{ring:3})  of \Cref{def:dag-rings}.
    Consider some $e_1 \in E_G(V^1_X,V^1_Y)$, $e_2 \in E_G(V^2_X,V^2_Y)$, and $(x_1,y_1) \in \splitt_\tcal(X,Y)$.
    We need to justify that there exists a shortest $(x_1,y_1)$-path that goes through both $e_1$ and $e_2$.
    
    Let $J_1,J_2$ be some shortest $(x_1,y_1)$-paths that go respectively through $e_1$ and $e_2$; they exist because each of $e_1,e_2$ has an orientation that belongs to $D$
     and $D$ is subgraph of $(x_1,y_1)\dagg$ (\Cref{claim:winding:dag}).
    Note that the endpoints of $e_1$
    are contained in $\cl(R_{\ell+4})$ while the endpoints of $e_2$ are contained in $\cl(\rr^2 \sm R_{r-4})$.
    Let $u \in V(J_1) \cap V(C(Q,P_{\ell+4}))$ and $w \in V(J_1) \cap V(C(Q,P_{r-4}))$ be vertices on $J_1$ that are visited after passing through $e_1$ (possibly $u$ can be an endpoint of $e_1$).
    Since $r = \ell + 10$ we can apply \Cref{claim:winding:middle} for $i = \ell+4$ and $j=r - 4$ (which satisfy $j = i+2$) to obtain a shortest $(u,w)$-path $J'$ that visits  $v_{\ell+4}''$.
    By \Cref{lem:prelim:replacement} we know that $J_1[x_1,u] + J' + J_1[w,y_1]$ is also a geodesic.
    So now we can assume that $J_1$ visits $v_{\ell+4}''$
    after passing through $e_1$.
    We perform an analogous modification to $J_2$, so we can assume that $J_2$ visits $v_{\ell+4}''$ before passing through $e_2$.
    Next, we take advantage of \Cref{lem:prelim:replacement} to infer that $J_1[x_1,v_{\ell+4}''] + J_2[v_{\ell+4}'',y_1]$ is a shortest $(x_1,y_1)$-path, which passes through both $e_1$ and $e_2$.
    This proves that $(V^1_X, V^1_Y \cap V^2_X, V^2_Y)$ is indeed an $(X,Y)\dagring$.
    
    Finally, recall that $\mathrm{Cycle}(\gamma_1)$ is a cycle in the dual graph $G^\star$ enclosing $V^1_X$.
    This is a superset of all vertices in $\cl(R_{\ell})$ and a subset of $R_{r}$.
    Hence $C(Q,P_\ell)$ and $C(Q,P_r)$ belong to different connected components of $\rr^2 \sm \mathrm{Cycle}(\gamma_1)$.
    The same holds for $\gamma_2$.
    Therefore, $\ring(\gamma_{1}, \gamma_{2}) \sub \ring(C(Q,P_\ell), C(Q,P_r))$.
    By \Cref{claim:winding:interval} we know that $\ring(C(Q,P_\ell), C(Q,P_r)) = R_r \sm \cl(R_\ell)$ contains no vertices from $\widehat{T}$.
    This concludes the proof of \Cref{prop:winding:dag-ring}.
\end{proof}

\section{Estimating the number of crossings}
\label{sec:crossings}

\subsection{Homotopy}

In this section we formalize the idea of counting the number of ``non-trivial crossings'' between a path $Q$ and a solution $\pcal$, mentioned in \Cref{sec:techniques}.
To this end, we  employ the concept of homotopy which, similarly as homology, detects similarities among topological objects.
However, we find homotopy to be a more intuitive language to conduct the reasoning, hence we will rely on it.
First, we identify a kind of handles which correspond to ``trivial crossings''.

\begin{definition}[Empty regular handle]
    Let $(G,\tcal)$ be a nice instance of \pdsp (resp. \pdap).
    For a path $Q$ in $G$ and a regular handle $H$ of $Q$ we define $D(Q,H)$ as the connected component of $\rr^2 \sm C(Q,H)$ that does not contain the endpoints of $Q$.
    We say that handle $H$ is {\em empty} if $D(Q,H)$ does not contain any terminal from $\tcal$.
\end{definition}

We can use an empty regular handle $H$ of $Q$ to ``pull'' the path $Q$ through $D(Q,H)$, obtaining a path that is in some sense topologically equivalent with respect to the terminal set.
The following definition formalizes this idea using the notion of homotopy.
We remark that our definition of homotopy deviates slightly from the one used in~\cite{LokshtanovMPSZ20}.

\begin{definition}[Homotopy]\label{def:homotopy}
    Let $(G,\tcal)$ be a nice instance of \pdsp (resp. \pdap), $Q$ be a path in $G$ and $H$ be an empty regular handle of $Q$.
    Next, let $q_1,q_2$ be the endpoints of $Q$ and $r_1,r_2$ be the endpoints of $H$ so that $q_1 <_Q  r_1 <_Q r_2 <_Q q_2$ for some orientation of $Q$.
    We define the path $Q^H$ as the concatenation $Q[q_1,r_1] + H + Q[r_2,q_2]$.
    We say that $Q^H$ has been obtained from $Q$ by a {\em pull operation}.
   
    We say that path $Q_2$ is {\em homotopic}
    to $Q_1$ if $Q_2$ can be obtained from $Q_1$ by a (possibly empty) series of pull operations.
\end{definition}

Note that our notion of homotopy preserves the endpoints of path $Q$ as well the first and the last edge on $Q$, because the attachments of a handle are always internal vertices of $Q$.

\begin{lemma}\label{lem:homotopy:equivalence}
    Homotopy forms an equivalence relation.
\end{lemma}
\begin{proof}
    Transitivity is obvious.
    Symmetry  follows from the fact that when a $(v_1,v_2)$-path $H$ is an empty regular handle of $Q$ then $Q[v_1,v_2]$ is an empty regular handle of $Q^H$.
\end{proof}

We are in position to formalize the idea of counting non-trivial crossings between a path $Q$ and a solution $\pcal$.
Intuitively, a crossing should be considered trivial when it can be avoided by some continuous shift of $Q$, which is formalized using homotopy.
Effectively we would like to count only those crossings that cannot be avoided by any tweak of $Q$.

For a subgraph $Q$ of $G$ we define $\partial(Q)$ as the subset of edges in $E(G) \sm E(Q)$ with at least one endpoint in $V(Q)$.
Note that when $Q$ is a path then any chord of $Q$ belongs to $\partial(Q)$ even though both its endpoints belong to $V(Q)$.

\begin{definition}\label{def:homotopy:load}
    Let $(G,\tcal)$ be a nice instance of \pdsp (resp. \pdap), $\pcal = (P_1,\dots,P_k)$ be some solution to $(G,\tcal)$, 
    and $Q$ be some path in $G$. 
    We define $\load(\pcal,Q)$ as the number of edges incident to $Q$ that are used by~$\pcal$, that is,
    \[ \load(\pcal,Q) = \sum_{i=1}^k \left|E(P_i) \cap \partial(Q)\right|. \]
    
    Furthermore, we define the {\em combinatorial load} $\minload(\pcal,Q)$ as the minimum $\load(\pcal,\widehat Q)$ over all paths $\widehat Q$ that are homotopic to $Q$. 
\end{definition}

\begin{figure}[t]
\centering
\includegraphics[scale=1.0]{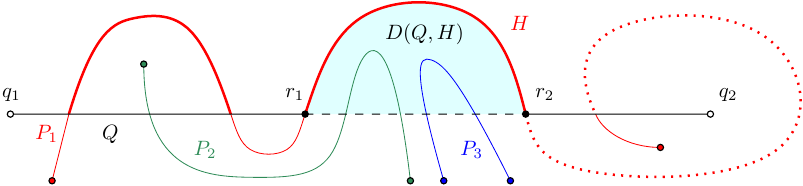}
\caption{
There are three $P_1$-handles of $Q$: the two marked in bold are regular handles and one winding handle is dotted.
The handle $H$ is empty as the set $D(Q,H)$ does not contain any terminals.
The path $Q^H = Q[q_1,r_1] + H + Q[r_2,q_2]$ is homotopic to $Q$.
We have $\load(\pcal, Q) = 20$ and $\load(\pcal, Q^H) = 10$.
}\label{fig:handles-and-load}
\end{figure}

In \Cref{sec:topo-load} we will extend the definition of load to the case where $Q$ is a dual path and introduce topological and algebraic counterparts of combinatorial load.

For paths $P, Q$ we say that a handle $H$ of $Q$ is a \emph{$P$-handle} whenever $H$ is a subpath of $P$.
For a path family $\pcal$ we say $H$ is a \emph{$\pcal$-handle} when it is a $P$-handle for some $P \in \pcal$.

\begin{lemma}\label{lem:homotopy:decrease}
    Let $G,\tcal,\pcal$ be as in \Cref{def:homotopy:load}
    and $Q$ be a geodesic in $G$.
    Next, let $H$ be a $\pcal$-handle of $Q$.
    Then $\load(\pcal, Q^H) < \load(\pcal,Q)$.
\end{lemma}
\begin{proof}
    Let $v_1 <_Q v_2$ be the endpoints of $H$.
    Consider any edge $e \in E(\pcal) \cap \partial(Q^H)$.
    The edge $e$ cannot be incident to an internal vertex of $H$ because $H$ is a subpath of a path from $\pcal$.
    Suppose that $e \in E(Q)$.
    Then $e$ is an edge between one of the endpoints of $H$ and a vertex $v_3$ satisfying $v_1 <_Q < v_3 <_Q v_2$.
    But this contradicts the monotonicity of crossings (\Cref{lem:prelim:monotone3}).
    Hence $e \in \partial(Q)$ and $e$ also contributes to $\load(\pcal,Q)$.
    On the other hand, $H$ contains at least one edge (note that $H$ might possibly consist of a single edge) that contributes to $\load(\pcal,Q)$ but not to $\load(\pcal,Q^H)$.
    The lemma follows.
\end{proof}

The following lemma will facilitate estimation of $\load(\pcal,Q)$.

\begin{lemma}\label{lem:homotopy:handles-load}
    Let $P, Q$ be paths in a graph $G$ and $\ell$ be the number of $P$-handles of $Q$.
    Then $|E(P) \cap \partial(Q)| \le 2\ell + 6$.
\end{lemma}
\begin{proof}
    Consider  $e \in E(P) \cap \partial(Q)$.
    There are 3 scenarios.

    \begin{itemize}[nosep]
        \item $e$ is incident to one of the endpoints of $Q$; there are at most 4 such edges.
        \item $e$ belongs to one of two subpaths of $P$ between an endpoint of $P$ and the closest vertex that is an internal vertex of~$Q$; there are at most 2 such edges.
        \item $e$ belongs to a $P$-handle of $Q$; there are $2\ell$ such edges.
    \end{itemize}
    The lemma follows.
\end{proof}

\subsection{Crossings outside rings}

We now explain  how an exhaustive ring decomposition comes in useful for estimating the number of ``non-trivial crossings''.
We will use \Cref{prop:winding:dag-ring} to justify that a geodesic $Q$ disjoint from all the rings in the decomposition cannot have too many winding $\pcal$-handles for any solution $\pcal$.
We will also bound the number of non-empty regular $\pcal$-handles using the following fact.

\begin{lemma}\label{claim:homotopy:disjoint}
        Let $Q, P$ be geodesics in graph $G$.
        Next, let $H_1, H_2$ be two distinct  regular $P$-handles of $Q$.
        Then $D(Q,H_1) \cap D(Q,H_2) = \emptyset$.
    \end{lemma}
    \begin{proof}
        From the definition of a handle we obtain that $H_1,H_2$ are internally disjoint subpaths of $P$.
        They may potentially intersect only at a common endpoint.
        Consider some orientation of $P$.
        Let $v_1,u_1$ be the endpoints of $H_1$ and $v_2,u_2$ be the endpoints of $H_2$ ordered in such a way that $v_1 <_P u_1 \le_P v_2 <_P u_2$.
        By \cref{lem:prelim:monotone3} we know that we can orient $Q$ so that $v_1 <_Q u_1 \le_Q v_2 <_Q u_2$.
        Let $q_2$ be the last vertex on $Q$ with respect to this orientation.
        Observe that $Q[u_2,q_2] \cap C(Q,H_1) = \emptyset$ so $u_2$ belongs to the same connected component of $\rr^2 \sm C(Q,H_1)$ as $q_2$.
        Consequently, $u_2 \not\in D(Q,H_1)$.
        Since both $Q[v_2,u_2]$ and $P[v_2,u_2]$ are disjoint from $C(Q,H_1) \sm \{u_1\}$ we get that $C(Q,H_2) \cap D(Q,H_1) = \emptyset$.
        Therefore, either $D(Q,H_1)$ is fully contained in $D(Q,H_2)$ or they are disjoint.
        By a symmetric argument we obtain $C(Q,H_1) \cap D(Q,H_2) = \emptyset$ which rules out the first option.
        The lemma follows.
    \end{proof}

We say that a path $Q$ intersects an $(X,Y)\dagring$ $(\gamma,\gamma')$ if $Q \cap \ring(\gamma,\gamma') \ne \emptyset$.

\begin{proposition}\label{prop:outside}
    Let $(G,\tcal)$
    be a nice instance of \pdsp and $\pcal = (P_1,\dots,P_k)$ be some solution to $(G,\tcal)$.
    Next, let $\widehat{T}$ be some \terset of $\tcal$, $\mathcal{R}$ be an exhaustive ring decomposition for $(G,\tcal,\widehat{T})$ and $\widehat Q$ be a geodesic in $G$ which does not intersect any \udagring from $\mathcal{R}$.
    Then $\minload(\pcal,\widehat Q) \leq \Oh(|\widehat{T}|^2)$.
\end{proposition}
\begin{proof}
We begin by showing that the properties of $\widehat Q$ are preserved by certain pull operations.
\begin{claim}
    Let $H$ be an empty regular $\pcal$-handle of $\widehat Q$.
    Then $\widehat Q^H$ is a geodesic and it does not intersect any \udagring from $\mathcal{R}$.
\end{claim}
\begin{innerproof}
    The path $H$ is a geodesic as a subpath of a path from $\pcal$.
    Then \Cref{lem:prelim:replacement} implies that $\widehat Q^H$ is a geodesic.
    Suppose for the sake of contradiction that $H$ intersects
    $\ring(\gamma,\gamma')$ where $(\gamma,\gamma')$ represents some \udagring from $\mathcal{R}$.
    By \Cref{obs:dag-cut:cross}(1), $H$ can cross $\gamma$ at most once hence $\gamma$ separates the endpoints of $H$.
    But these endpoints lie on $\widehat Q$ so some edge in $\widehat Q$ belongs to $\gamma$.
    This contradicts the assumption that $\widehat Q$ is disjoint from $\ring(\gamma,\gamma')$.
    As a consequence, we infer that $H$ as well as $\widehat Q^H$ are disjoint from all the rings in $\mathcal{R}$.   
\end{innerproof}

While $\widehat Q$ has an empty regular $\pcal$-handle $H$, we replace $\widehat Q$ by $\widehat Q^H$.
The claim above ensures that the conditions of the proposition are maintained.
By \Cref{lem:homotopy:decrease} we know that this process must terminate at some point.
Let $Q$ denote the path obtained in the end of this process; then we know that $Q$ has no empty regular $\pcal$-handles.
Since $Q$ is homotopic to $\widehat Q$ we have $\minload(\pcal,\widehat Q) \le \load(\pcal,Q)$.
Therefore, is suffices to count the number of $\pcal$-handles of $Q$ that are either  non-empty regular handles or winding handles.

    \begin{claim}\label{claim:outside:nonempty}
        For each $i \in [k]$ there are at most $2k$ non-empty regular $P_i$-handles of $Q$.
    \end{claim}
    \begin{innerproof}
        Let $P_i^1,\dots,P_i^t$ be the collection of non-empty regular $P_i$-handles of $Q$.
        By \Cref{claim:homotopy:disjoint}  the sets $D(Q,P_i^j)_{j=1}^t$ are pairwise-
        disjoint.
        Next, non-emptiness means that each set $D(Q,P_i^j)$ contains some terminal from $\tcal$.
        There are only $2k$ terminals which yields an upper bound on~$t$.
    \end{innerproof}

We move on to the analysis of the winding handles.
Let $\widehat k = |\widehat T|$. Note that $2k \le \widehat k$.

\begin{claim}\label{claim:outside:winding}
        For each $i \in [k]$ there are at most $20\cdot(\widehat{k}+1)$  winding $P_i$-handles of $Q$.
    \end{claim}
    \begin{innerproof}
        Suppose, for the sake of contradictions, that there are $20\cdot(\widehat{k}+1)+1$ winding $P_i$-handles of $Q$.
        Then there is a family of $10\cdot(\widehat{k}+1)+1$ many such handles that are vertex-disjoint.

        Let $P_i[Q]$ be the subpath of $P_i$ between its first and last intersection with $Q$, inclusively.
        We apply \Cref{prop:winding:dag-ring}
        to obtain that there exists a region $U \sub \rr^2$,
        two winding handles $P_i^1, P_i^2 \sub P_i[Q]$ of $Q$, and a splitting partition $(X,Y)$ of $\tcal$
        such that (i) the boundary of $U$ is $C(Q,P_i^1) \cup C(Q,P_i^2)$, (ii) $U$~contains no vertex from $\widehat T$, (iii) there exists an $(X,Y)\dagring$ $(\gamma_1,\gamma_2)$ with $\ring(\gamma_1,\gamma_2) \sub U$.

        We argue that $\ring(\gamma_1,\gamma_2)$ is disjoint from all the rings in $\mathcal{R}$.
        Recall that $Q$ is disjoint from all the rings in $\mathcal{R}$.
        By \Cref{obs:dag-cut:cross}(1) whenever the path $P_i$ enters any \udagring through a \udagcut $\gamma$ then it cannot cross $\gamma$ again.
        Since no \udagcut from $\mathcal{R}$ separates any vertices on $Q$, the path $P_i[Q]$ must be also disjoint from all the rings in $\mathcal{R}$.
        Since $P_i^1, P_i^2$ are subpaths of $P_i[Q]$, we infer that $C(Q,P_i^1), C(Q,P_i^2)$ are also disjoint from all the rings in $\mathcal{R}$.

        Suppose for the sake of contradiction that $U$ intersects  $\ring(\gamma_3,\gamma_4)$ where $(\gamma_3,\gamma_4)$ represents a \udagring from $\mathcal{R}$.
        Since $\ring(\gamma_3,\gamma_4)$ is disjoint from the boundary of $U$, it must be that $\ring(\gamma_3,\gamma_4) \sub U$.
        Let $\Delta_3$ denote $\disc(\mathrm{Cycle}(\gamma_3))$.
        Recall that $\gamma_3,\gamma_4$ are $(X',Y')\dagcuts$ for some splitting partition $(X',Y')$ so in particular $\Delta_3$ contains at least one terminal.
        Consequently, $\Delta_3$ cannot be fully contained in $U$, which contains no terminals.
        Likewise, the complement of $\Delta_3$ also cannot be contained in $U$.
        Because the boundary of $U$ is a union of two connected sets $C(Q,P_i^1), C(Q,P_i^2)$, exactly one of them must be contained in $\Delta_3$.
        This means that $\gamma_3$ separates a pair of vertices from $Q$ and this contradicts the assumption that $Q$ is disjoint from $\ring(\gamma_3,\gamma_4)$.

        We have established that no \udagring from $\mathcal{R}$ intersects $U$.
        Consequently, $(X,Y,\gamma_1,\gamma_2)$ forms an extension to the ring decomposition $\mathcal{R}$.
        This contradicts the assumption that $\mathcal{R}$ is exhaustive. 
        The claim follows.
    \end{innerproof}
    
    We are in position to estimate $\load(\pcal,Q)$.
    By Claims \ref{claim:outside:winding} and \ref{claim:outside:nonempty}, the number of $P_i$-handles of $Q$ is at most $20\cdot(\widehat{k}+1) + 2k$, for each $i \in [k]$.
    Then by \Cref{lem:homotopy:handles-load} we have $|E(P_i) \cap \partial(Q)| \le 40\cdot(\widehat{k}+1) + 4k + 6 \le 46\cdot(\widehat{k}+1)$.
    We conclude that $\minload(\pcal,\widehat Q) \le \load(\pcal,Q) \le k \cdot 46\cdot(\widehat{k}+1) = \Oh(\widehat{k}^2)$.
\end{proof}

\subsection{Rerouting inside a ring}

In order to deal with paths traversing a \udagring, we will rely on a rerouting argument by Cygan et al.~\cite{CyganMPP13}, which in turn is based on the work by Ding et al.~\cite{ding1993disjoint} on packing disjoint cycles in a toroidal graph.
First, we need to summon the definition of the {\em winding number}.

\begin{definition}[{\cite{CyganMPP13arxiv}}]
    We say that a (directed or undirected) graph $G$ is embedded in a {\em rooted ring} $(C_1,C_2,W)$ when the following conditions hold:
    \begin{itemize}[nosep]
        \item $C_1$ is the outer face of $G$ and $C_2$ is an interior face of $G$, and
        \item $W$ is a dual path in $G$ connecting $C_1$ to $C_2$, oriented from $C_1$ to $C_2$.
    \end{itemize}
    
We call $W$ a {\em reference curve}.
For an oriented path $P$ we define its {\em winding number} $\wind(P,W)$ as
the number of signed crossings with $W$:
each time $P$ crosses $W$ from left to right (with respect to the fixed orientations of $P$ and $W$) we add +1, and when $P$ crosses $W$ from right to left, we add -1.
When $P$ connects a vertex $v_1$ on $C_1$ to a vertex $v_2$ on $C_2$, on default we consider it oriented from $v_1$ to $v_2$.
\end{definition}

We now state the rerouting argument that will allow us to bound the winding numbers of the solution paths inside a DAG-ring.
Intuitively speaking, when $\mathcal{Q} = (Q_1,\dots,Q_s)$, $\pcal = (P_1,\dots,P_s)$ are linkages between vertices on $C_1$ and $C_2$, then $\pcal$ can be modified to have almost the same winding as $\mathcal{Q}$ while preserving the endpoints of $\pcal$.

\begin{lemma}[{\cite[Lemma 4.8]{CyganMPP13arxiv}}]
\label{lem:reroute:reference}
    Let $D$ be a digraph embedded in a rooted ring $(C_1,C_2,W)$.
    The four sequences of vertices of $D$ listed below are given in a clockwise order.
\begin{itemize}
    \item Let $p^1_1,\dots,p^1_s$ be vertices on $C_1$  and let  $p^2_1,\dots,p^2_s$ be vertices on $C_2$.
     \item Let $q^1_1,\dots,q^1_s$ be vertices on $C_1$  and let  $q^2_1,\dots,q^2_s$ be vertices on $C_2$.
    \item Let $P_1, \dots, P_s$ be a set of pairwise vertex-disjoint paths in $D$ contained between $C_1$ and $C_2$ such
    that the endpoints of $P_i$ are $p^1_i$ and $p^2_i$.
    \item Let $Q_1,\dots , Q_s$ be a set of pairwise vertex-disjoint paths in $D$ contained between $C_1$ and $C_2$ such
    that the endpoints of $Q_i$ are $q^1_i$ and $q^2_i$.
    \item $P_i$ and $Q_i$ go in the same direction, i.e., $P_i$ goes from $p^1_i$ to $p^2_i$ iff. $Q_i$ goes from $q^1_i$ to $q^2_i$.
\end{itemize}
Then there is set $P'_1, \dots, P'_s$
of pairwise vertex-disjoint paths in $D$ between $C_1$ and $C_2$ such that $P'_i$ and $P_i$
 have the same start/end vertices and $|\wind(P'_i,W) - \wind(Q_i,W)| \le 6$.
\end{lemma}

The condition that $P_i$ and $Q_i$ go in the same direction is immaterial for our work because we will only consider directed paths going from $C_1$ to $C_2$.

\subsection{Crossings inside a ring}

Recall the notion of a DAG-structure from \Cref{def:ring:dag-structure}.

\begin{proposition}\label{prop:reroute:minload}
    Let $(G,\tcal)$ be a nice instance of \pdsp, $\pcal$ be some solution to $(G,\tcal)$, $(X,Y)$ be a splitting partition of a \terset $\widehat{T}$, and  $(\gamma, \gamma')$ be an $(X,Y)$\dagring with a DAG-structure $(D, T_{X}, T_{Y})$.
    Furthermore, let
    $\ell = |\splitt_\tcal(X, Y)|$ and $Q_1, Q_2, \dots, Q_\ell$ be a family of disjoint directed $(T_{X}, T_{Y})$-paths in $D$.
    Then there exists a solution $\pcal'$ that coincides with $\pcal$ outside $\ring(\gamma, \gamma')$ and
    $\minload(\pcal',Q_1) \leq \Oh(k)$.
\end{proposition}
\begin{proof}
    Let $C_{X},C_{Y}$ denote the faces of $D$ enclosed by $\mathrm{Cycle}(\gamma)$ and $\mathrm{Cycle}(\gamma')$.
    Next, let $W$ be one of the dual $(C_{X},C_{Y})$-paths going parallel to $Q_1$.
    Clearly $\wind(Q_i,W) = 0$ for each $i \in [\ell]$.
    Then $D$ is embedded in the rooted ring $(C_{X},C_{Y},W)$.
    By \Cref{obs:dag-cut:cross}(1) we know that there are exactly $\ell$ paths in $\pcal$ that intersect $\ring(\gamma, \gamma')$ and each intersection is a single $(T_{X}, T_{Y})$-path.

    \begin{claim}\label{claim:reroute:replacement}
        Consider $P \in \pcal$ that traverses $\ring(\gamma, \gamma')$, oriented in such a way that it first crosses $\gamma$. Let $v_1 <_P v_2$ be two vertices from $V(P)$ that lie in $\ring(\gamma, \gamma')$.
        Next, let $\widehat{P}$ be a directed $(v_1,v_2)$-path in $D$.
        Then $P[\cdot,v_1] + \widehat{P} + P[v_2,\cdot]$ is a geodesic.
    \end{claim}
    \begin{innerproof}
        Let $t,t'$ be endpoints of $P$.
        Then $P$ is an (oriented) shortest $(t,t')$-path if and only if $P$ is a directed path in $(t,t')\dagg$.
        \Cref{lem:ring:dag} implies that $D$ is a subgraph of $(t,t')\dagg$ and so $\widehat{P}$ is a directed path in $(t,t')\dagg$. The claim follows.
    \end{innerproof}

    First consider $\ell = 1$.
    Then there exactly one index $i \in [k]$ for which $P_i$ intersects the ring and this intersection is a single $(T_X, T_Y)$-subpath of $P_i$.
    If $P_i$ does not intersect $Q_1$ then $\load(\pcal,Q_1) = 0$.
    Otherwise, let $v_1,v_2$ be the first and the last vertices on $P_i$ that belong to $V(Q_1)$.
    Let $P'_i = P_i[\cdot,v_1] + Q[v_1,v_2] + P_1[v_2,\cdot]$
    and $\pcal'$ be obtained from $\pcal$ by replacing $P_1$ with $P'_1$.
    Since $Q[v_1,v_2]$ is a path in $D$, \Cref{claim:reroute:replacement} implies that $P'_1$ is a geodesic.
    Hence $\pcal'$ is a valid solution and $\load(\pcal,Q_1) \le 2$.

    For now on we assume $\ell \ge 2$.
    For simplicity, assume w.l.o.g. that $\splitt_\tcal(X,Y)$ comprises the pairs $(s_1,t_1),\dots,(s_\ell,t_\ell)$.
    For $i \in [\ell]$ let $\widehat{P}_i$ denote the maximal subpath of $P_i$ contained in $\ring(\gamma, \gamma')$; then $\widehat{P}_i$ is a directed $(T_{X}, T_{Y})$-path in $D$.
    We apply \Cref{lem:reroute:reference} to the digraph $D$ and the disjoint-paths families $(\widehat{P}_i)_{i=1}^\ell$ and $(Q_i)_{i=1}^\ell$ to obtain a disjoint-paths family $(P'_i)_{i=1}^\ell$ that have the same endpoints as $(\widehat{P}_i)_{i=1}^\ell$ and $|\wind(P'_i,W)| \le 6$ for each $i \in [\ell]$.
    By \Cref{claim:reroute:replacement} replacing the subpath $\widehat{P}_i$ with $P'_i$ in $P_i$ results in a geodesic. 
    Moreover, this modification cannot spoil the disjointedness of paths.
    Let us keep the variable $\pcal$ to denote this new solution.

    When $P$ is an oriented path and $H$ is a $P$-handle of $Q$ we consider the orientation of the cycle $C(Q,H)$ that obeys the orientation of $P$.
    We refer to such oriented cycle as $\vec{C}(Q,H)$

    \begin{claim}\label{claim:reroute:regular}
        Let $Q$ be a path in $G$ contained fully in $\ring(\gamma, \gamma')$.
        Let $H$ be a regular $\pcal$-handle of $Q$.
        Then the handle $H$ is empty and entirely contained in $\ring(\gamma,\gamma')$, and $\wind(Q^H,W) = \wind(Q,W)$. 
    \end{claim}
    \begin{innerproof}
        By \Cref{obs:dag-cut:cross}(1) the intersection of each $P \in \pcal$ with $\ring(\gamma, \gamma')$ can be only a single subpath of~$P$.
        Hence $H$ is entirely contained in  $\ring(\gamma, \gamma')$ and also $D(Q,H) \sub \ring(\gamma, \gamma')$.
        Since an $(X,Y)\dagring$ contains no terminals from $\tcal$, we infer that $H$ is indeed an empty handle.

        Now we argue that $\wind(Q^H,W) = \wind(Q,W)$.
        Let $v_1 <_Q v_2$ denote the endpoints of~$H$.
        Observe that $\wind(\vec{C}(Q,H),W) = 0$ because the number of times $W$ enters $D(Q,H)$ equals the number of time $W$ exits  $D(Q,H)$ and these two kinds of crossings contribute opposite values to the winding number.  
        We have $\wind(H,W) - \wind(Q[v_1,v_2],W) = \wind(\vec{C}(Q,H),W) = 0$.
        Hence $\wind(Q^H,W) = \wind(Q,W) - \wind(Q[v_1,v_2],W) + \wind(H,W) = \wind(Q,W)$.
    \end{innerproof}

    To estimate $\minload(\pcal,Q_1)$ we proceed similarly as in \Cref{prop:outside}.
    Observe that $Q_1$ is a geodesic in $G$ because it is a directed path in $D$.
    We begin with $Q = Q_1$ and while $Q$ has a regular $\pcal$-handle $H$ we replace $Q$ with $Q^H$.
    By \Cref{claim:reroute:regular} this handle must be empty and the invariant $\wind(Q,W) = 0$ is maintained.
    In addition, \Cref{lem:prelim:replacement} implies that this operation preserves the property of being a geodesic.
    By \Cref{lem:homotopy:decrease} this process must terminate; from now on let $Q$ denote the path obtained in the end. 
    Then $Q$ is homotopic to $Q_1$ and $Q$ has no regular $\pcal$-handles.    
    
    \begin{claim}\label{claim:reroute:winding-one}
    Let $H$ be a winding handle of $Q$ with endpoints $v_1 <_Q v_2$ and $W'$ be some dual $(C_{X},C_{Y})$-path.
    Then $\wind(\vec{C}(Q,H),W') \in \{-1,1\}$ and this value does not depend on the choice of $W'$.
    \end{claim}
    \begin{innerproof}
        Let $U$ be the connected component of $\rr^2 \sm C(Q,H)$ that contains the endpoint of $Q$ lying on $C_{X}$.
        By the definition of a winding handle, $U$ does not contain the other endpoint of $Q$.
        Observe that whenever $W'$ crosses $C(Q,H)$ then it alternatively exits and enters $U$, and such crossings contribute opposite values to $\wind(\vec{C}(Q,H),W')$.
        Consequently, $\wind(\vec{C}(Q,H),W') \in \{-1,0,1\}$.
        The value 0 cannot be achieved because the number of crossings is odd.
        Then  $\wind(\vec{C}(Q,H),W') = 1$ if and only if 
         $U$  is the connected component of $\rr^2 \sm C(Q,H)$ that lies to the right of $C(Q,H)$ with respect to its orientation; and $\wind(\vec{C}(Q,H),W') = -1$ otherwise.
    \end{innerproof}
    
    We say that a $\pcal$-handle $H$ is a {\em clockwise} handle when $\wind(\vec{C}(Q,H),W) = -1$ and a {\em counter-clockwise} handle otherwise.

    \begin{claim}\label{claim:reroute:add}
        Let $H$ be a path in $D$ that forms a clockwise (resp. counter-clockwise) winding handle of $Q$.
        Then $\wind(Q^H,W) = \wind(Q,W) - 1$ (resp. $\wind(Q^H,W) = \wind(Q,W) + 1$).
    \end{claim}
    \begin{innerproof}
     Let $v_1 <_Q v_2$ be the endpoints of $H$. We consider this path oriented from $v_1$ to $v_2$.
        We have 
        \begin{align*}
         \wind(Q^H,W) & = \wind(Q[\cdot,v_1],W) + \wind(H,W) + \wind(Q[v_2, \cdot],W) \\
        & = \wind(Q,W) - \wind(Q[v_1,v_2],W) + \wind(H,W) \\
        & = \wind(Q,W) + \wind(\vec{C}(Q,H),W).
        \end{align*}
        The last term is either 1 or -1 depending on whether $H$ is a clockwise handle.
    \end{innerproof}

\begin{figure}[t]
\centering
\includegraphics[scale=0.75]{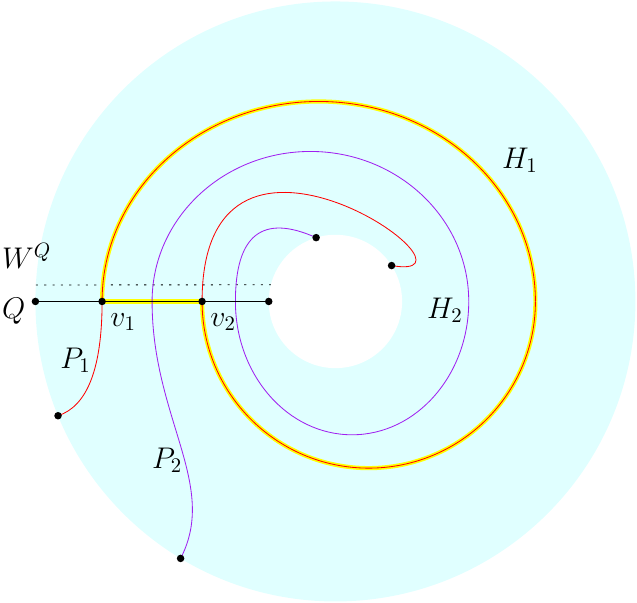}
\caption{
An illustration to \Cref{claim:reroute:all-the-same}.
The $P_1$-handle $H_1$ of $Q$ is clockwise.
The cycle $C(Q,P_1)$ is highlighted.
The reference curve $W^Q$ is crossed by $\vec{C}(Q,H_1)$ exactly once.
The path $P_2$ must cross the cycle $C(Q,H_1)$ through $Q[v_1,v_2]$ and directly afterwards, within handle $H_2$, it must cross $W^Q$ in the same direction as $H_1$.
}\label{fig:clockwise}
\end{figure}

    \begin{claim}\label{claim:reroute:all-the-same}
         Either all the winding $\pcal$-handles of $Q$ are clockwise or all are counter-clockwise.
    \end{claim}
    \begin{innerproof}
        Here be exploit the assumption $\ell \ge 2$.
        Recall that $Q$ has no regular $\pcal$-handles.
        Let $H_1$ be a winding $P_1$-handle of $Q$ with endpoints $v_1 <_Q v_2$.
        Since $C(Q,H_1)$ separates $C_{X}$ and $C_{Y}$ every other path $P_i$, $i \in [2,\ell]$ must go through $Q[v_1,v_2]$.
        Since $\ell \ge 2$, there is at least one such path.
        Suppose w.l.o.g. that it is $P_2$ that enters $Q[v_1,v_2]$ at a point closest to~$v_1$.
        Consider a winding $P_2$-handle $H_2$ that starts at a vertex in $Q[v_1,v_2]$.
        
        We claim that $\wind(\vec{C}(Q,H_1),W) = \wind(\vec{C}(Q,H_2),W)$.
        By \Cref{claim:reroute:winding-one} the winding number of $\vec{C}(Q,H_i)$ is independent from the choice of the reference curve $W$.
        Let $W^Q$ be any of  the two dual $(C_{X},C_{Y})$-paths that go parallel to $Q$.
        Since a handle of $Q$ can have at most two crossings with $W^Q$ and we know that this number is odd for a winding handle, we infer that each cycle $C(Q,H_i)$ has exactly one crossing with $W^Q$.
        Let us now fix $W^Q$ to be this dual path that is crossed first by $H_1$ (with respect to the orientation inherited from $P_1$).
        Then also the first edge of $H_2$ crosses $W^Q$ and does it in the same direction as~$H_1$ (see \Cref{fig:clockwise}).
        Hence $\wind(\vec{C}(Q,H_1),W^Q) = \wind(\vec{C}(Q,H_2),W^Q)$.

        Observe that the restriction to indices 1,2 is immaterial and 
        the argument above implies that each pair of consecutive $\pcal$-handles of $Q$ share the same orientation.
        Hence all of them are clockwise or none is.
    \end{innerproof}

    \begin{claim}\label{lem:reroute:prefix}
        Let $P$ be a $(T_X,T_Y)$-path in $D$ with a non-empty intersection with $Q$ and let $v_1,v_2$ be the first and the last common vertex of $P$ and $Q$.
        Consider $P' = Q[\cdot,v_1] + P[v_1,v_2] + Q[v_2,\cdot]$.
        Then $|\wind(P,W) - \wind(P',W)| \le 2$.
    \end{claim}
    \begin{innerproof}
        Let $J$ be the oriented cycle obtained by concatenating $Q[\cdot,v_1]$ with reversed $P[\cdot,v_1]$ and a curve $S$ within $C_X$ that connects the endpoint of $P$ to the endpoint of $Q$.
        Since $W$ crosses $J$ alternatively, we have $\wind(J,W) \in \{-1,0,1\}$.
        Next, the curve $S$ does not cross $W$ hence $\wind(Q[\cdot,v_1],W) - \wind(P[\cdot,v_1],W) = \wind(J,W)$ and so these two winding numbers differ by at most 1.
        Consequently, replacing this prefix of $P$ with the prefix of $Q$ modifies the winding number of the path by at most~1.
        The same argument applies to the modification of a suffix.
    \end{innerproof}

Suppose for the sake of contradiction 
that for some $i \in [\ell]$ the number of $P_i$-handles of $Q$ is greater than 8.
Recall that all of them must be winding handles.
Let ${Q}_{\mathrm{wind}}$ be obtained form $Q$ by pulling all those handles, i.e., we iteratively replace $Q$ with $Q^H$ for each $P_i$-handle $H$.
By Claims~\ref{claim:reroute:add} and~\ref{claim:reroute:all-the-same} we obtain $|\wind({Q}_{\mathrm{wind}},W)| \ge 9$.
In turn, \Cref{lem:reroute:prefix} says that $|\wind({Q}_{\mathrm{wind}},W) - \wind(P_i,W)| \le 2$ because these paths may differ only at the very beginning and the very end.
But the former application of \Cref{lem:reroute:reference} has guaranteed that $|\wind(P_i,W)| \le 6$; a contradiction.

We conclude that for each $i \in [\ell]$ the number of $P_i$-handles of $Q$ is at most  8.
By \Cref{lem:homotopy:handles-load} we estimate $|E(P_i) \cap \partial(Q)| \le 22$.
Hence $\minload(\pcal,Q_1) \le \load(\pcal,Q) \le 22\cdot \ell \le 22\cdot k$.
\end{proof}

\subsection{Geodesic Steiner tree}
\label{sec:primalSteiner}

We are in position to construct a Steiner tree that will work as a blueprint for enumerating homology classes, as discussed in \Cref{sec:techniques}.
We start with an arbitrary Steiner tree spanning the terminals and iteratively refine it using an exhaustive ring decomposition.
Then we will take advantage of Propositions~\ref{prop:outside} and~\ref{prop:reroute:minload} to estimate the number of ``non-trivial crossings'' between the tree and a certain solution.

Before proceeding to the statement of the next result, we define \textsl{principal vertices} and \textsl{spinal paths} of a tree $T$. Let $G$ be a graph, let $B\subseteq V(G)$ and let $T$ be a tree in $G$ such that $B\subseteq V(T)$ and every leaf of $T$ belongs to $B$.
We call a vertex of $T$ \emph{principal} if it either belongs to $B$ or has degree more than two. A \emph{spinal path} of $T$ is a maximal path in $T$ whose all internal vertices are non-principal.

\begin{proposition}\label{prop:primalSteiner}
    There is an algorithm that, given
     a nice instance $(G,\tcal)$ of \pdsp with $|\tcal| = k$, runs in time $2^{\Oh(k)}\cdot n^{\Oh(1)}$, and returns either a report that $(G,\tcal)$ has no solution, or a nice instance $(G',\tcal',\mathcal{D}')$ of {\sc Planar Disjoint Annotated Paths} and a tree $T$ in $G'$ such that 
     the following conditions hold.
     \begin{enumerate}
        \item $|\tcal'| = k$.
        \item $T$ spans the terminals of $\tcal'$ in $G'$ and $T$ consists of $\mathcal{O}(k)$ spinal paths.
        \item $(G,\tcal)$ admits a solution if and only if $(G',\tcal',\mathcal{D}')$ admits a solution $\pcal'$ such that for every spinal path $Q$ in $T$ it holds that $\minload(\pcal',Q) = \Oh(k^3)$.
     \end{enumerate}
\end{proposition}
\begin{proof}
    The algorithm works as follows.
    Considering an arbitrary ordering of the terminals of $\tcal$, we find a shortest path $P$ between the first two vertices in the ordering and we set $H:=P$.
    Then, for each next terminal $t$ in the ordering we find a shortest path $P'$ between $V(H)$ and $t$, we update $H$ by appending the path $P'$, and proceed to the next terminal in the ordering.
    In the end of this procedure, the obtained graph $H$ is a tree where every terminal from $\tcal$ appears as a principal vertex and consists of at most $2k$ spinal paths, each of them being a geodesic in $G$.
    Let $\widehat{T}$ be the set of principal vertices of $H$ and note that it is a \terset of $\tcal$. Observe that $H$ can be computed in polynomial time.
    
    Next, we invoke the algorithm of~\Cref{prop:ring:ring-decomp} to compute an exhaustive ring decomposition $(X^i,Y^i,\gamma_{1}^{i},\gamma_{2}^{i})_{i=1}^\ell$ of $(G,\tcal,\widehat{T})$, where $\ell\le 2|\widehat{T}|-2$.
    This can be done in time $2^{|\widehat{T}|}\cdot n^{\Oh(1)}$.
    For every $i\in[\ell]$, let $(U_{X^i},U_{\mathsf{mid}}^i,U_{Y^i})$ be the $(X^i,Y^i)$\dagring represented by $(\gamma_{1}^{i},\gamma_{2}^{i})$.
    
    We construct the graph $G'$ as follows:
    \begin{itemize}[nosep]
        \item for every $i\in[\ell]$, subdivide once every edge $e$ that belongs to $E(U_{X^i},U_{\mathsf{mid}}^i)$ or $E(U_{\mathsf{mid}}^i,U_{Y^i})$, by introducing a new vertex $v_e$,
        
        \item for every $i\in[\ell]$, connect the vertices $v_e$ with $e\in E(U_{X^i},U_{\mathsf{mid}}^i)$ in a cycle, following the order that the corresponding edges are crossed by the cycle $\mathrm{Cycle}(\gamma_{1}^i)$ of the dual graph $G^\star$ . Similarly, connect the vertices $v_e$ with  $e\in E(U_{\mathsf{mid}}^i,U_{Y^i})$ in a cycle. We denote these newly introduced cycles by $C_1^i$ and $C_2^i$.
    \end{itemize}
    We denote by $F$ the set of all edges subdivided in the above procedure.
    Since all terminals of $\tcal$ appear also in $G'$, we can set $\tcal':=\tcal$. We also set $\mathcal{D}'$ to be the collection of oriented subgraphs obtained from $\mathcal{D}$ after subdividing edges in $F$. Note that since $(G,\tcal)$ is nice, $(G',\tcal',\mathcal{D}')$ is nice as well.

    It now remains to construct the tree $T$ in $G'$ and show that the two last conditions of the statement hold.
    Before this, let us comment that there is a 1-1 correspondence between the shortest paths DAGs as well as the solutions in $(G,\tcal)$ with the corresponding objects in $(G',\tcal')$.
    Indeed, for every $i\in[k]$, the $(s_i,t_i)$\dagg in $G$ is obtained from the $(s_i,t_i)$\dagg in $G'$ by replacing the edges incident to each vertex $v_e$ with the edge $e$. 
    This implies that if $\mathcal{D}$ is the set of the shortest paths DAGs $(s_i,t_i)\dagg$ in $G$, for $i \in [k]$, then $(G,\tcal,\mathcal{D})$ has a solution if and only if $(G',\tcal',\mathcal{D}')$ has a solution.
    In particular if we associate each family $\mathcal{P}$ of vertex-disjoint paths $P_1,\ldots,P_k$ in $G$, so that $P_i$ is an $(s_i,t_i)$-path, with the corresponding family $\mathcal{P}'$ of paths in $G'$ obtained after the aforementioned subdivisions, then $\mathcal{P}$ is a solution to $(G,\tcal,\mathcal{D})$ if and only if $\mathcal{P}'$ is a solution to $(G',\tcal',\mathcal{D}')$.
    This, combined with~\Cref{obs:prelim:annotated}, implies that $\mathcal{P}$ is a solution to $(G,\tcal)$ if and only if $\mathcal{P}'$ is a solution to $(G',\tcal',\mathcal{D}')$.
    
    In order to construct the tree $T$, we will replace the part of $H$ that lies inside rings of the decomposition by a collection of paths. 
    In order to maintain the property that the obtained graph $T$ is a tree, we use the cycles $C_1^i,C_2^i$ and ``rerouting'' in each ring will be done in a way that the load of each of the obtained paths is small. 

    Let us first show how to construct the part of $T$ inside each ring. For this, we focus on the graph $G$ and we compute the DAG-structure $(D_i,T_{X^i},T_{Y^i})$ (\Cref{def:ring:dag-structure}) for the $(X^i,Y^i)$\dagring $(U_{X^i},U_{\mathsf{mid}}^i,U_{Y^i})$ for every $i\in[\ell]$. For every $i\in[\ell]$, we ask for a collection of $|\splitt_\tcal(X^i,Y^i)|$ many vertex-disjoint paths in $D_i$ from $T_{X^i}$ to $T_{Y^i}$ and we set $Q^i$ be one of these paths. Note that if such linkage does not exist, then we can safely report that $(G,\tcal)$ is a no-instance. Observe that the computation of the DAG-structures and the paths $Q^i$ can be done in polynomial time. Next, because of~\Cref{prop:reroute:minload} 
    the existence of any solution implies the existence of a solution $\pcal$
    so that for every $i\in[\ell]$ it holds that
    \begin{equation}
        \minload(\pcal,Q^i) = \Oh(k).\label{eq:load:inn}
    \end{equation}
    We are now ready to define the graph $T$, which is obtained from $H$ by applying the following steps.
    \begin{itemize}[nosep]
        \item Consider the subdivision $H'$ of $H$ in $G'$ obtained by replacing each edge in $uv\in F\cap E(H)$ with $uv_e, v_ev$.
        \item For every $i\in[\ell]$, remove from $H'$ all vertices in $V(H')\cap U_{\mathsf{mid}}^i$ and add the cycles $C_1^i,C_2^i$ and the path $Q^i$. We denote the obtained graph by $\tilde{H}$ and observe that it is a connected graph that spans all terminals in $\tcal'$.

        \item Then, for every pair of distinct $i,j\in[\ell]$, consider all subpaths of $T$ whose one endpoint is in $V(C_1^i)\cup V(C_2^i)$ and its other endpoint is in $V(C_1^j)\cup V(C_2^j)$ 
        and have no internal vertex in $\bigcup_{i\in[\ell]}U_{\mathsf{mid}}^i$. If there are more than one such paths for $i,j$, then discard the internal vertices of all but one of these paths from $\tilde{H}$.
        
        \item Finally, find a maximal set $J\subseteq \bigcup_{i\in[\ell]} E(C_1^i)\cup E(C_2^i)$ such that $\tilde{H}-J$ is a tree, set $T:=\tilde{H}-J$. 
    \end{itemize}
    
    Consider the vertices $v_e\in V(T)$, for $e\in F$, that are non-principal vertices of $H'$ and have a degree greater than two in $T$. Note that these vertices together with the principal vertices of $H'$ form the set of principal vertices of $T$.  
    
    Next, note that every spinal path $Q$ of $T$ is a concatenation of paths of the following kinds:
    \begin{enumerate}[nosep]
        \item[(a)] subpaths of spinal paths of $H'$ that are disjoint from $\bigcup_{i\in[\ell]}U_{\mathsf{mid}}^i$, or
        \item[(b)] paths whose edge set is a subset of $E(C_1^i)$ or $E(C_2^i)$ for some $i\in[\ell]$, or
        \item[(c)] paths which after the removal of their endpoints result to $Q^i$ for some $i\in[\ell]$.
    \end{enumerate}
    Also notice that, by construction, every spinal path of $H$ is a subpath of a shortest $(t,t')$-path for some $t,t'\in\widehat{T}$. Therefore, due to~\Cref{obs:dag-cut:cross}, every spinal path of $H$ that intersects $U_{\mathsf{mid}}^i$ for $i\in[\ell]$, has exactly one edge in $E(U_{X^i},U_{\mathsf{mid}}^i)$ and exactly one edge in $E(U_{\mathsf{mid}}^i,U_{Y^i})$.
    Therefore, for each $i\in[\ell]$,
    if $v$ is a (new) principal vertex of $T$ appearing on $V(C_1^i)$ then it is the endpoint of a subpath of $H'$ that is internally disjoint from $\bigcup_{h\in [\ell]}U_{\mathsf{mid}}^h$
    and whose one endpoint is $v$ and its other endpoint is either in $\widehat{T}$ or in $V(C_1^j)\cup V(C_2^j)$ for some $j\neq i$.
    Observe that, by construction, the total number of (newly) introduced principal vertices of $T$ is upper-bounded by a linear function of $|\widehat{T}|+\ell$ and therefore the number of spinal paths of $T$ is $\mathcal{O}(k)$. This completes the proof of condition (2) of the statement for $T$.
    Moreover, since every spinal path of $H$ can enter each ring at most once, every spinal path of $T$ is a concatenation of $\mathcal{O}(k)$ paths of one of the kinds (a), (b), and (c) given above.
    
    We next show that condition (3) holds for $T$.
    \begin{claim}
        $(G,\tcal)$ admits a solution if and only if $(G',\tcal', \dcal')$ admits a solution $\mathcal{P}'$ so that for every spinal path $Q$ in $T$ it holds that $\minload(\mathcal{P}',Q)=\mathcal{O}(k^3)$  
    \end{claim}
    \begin{innerproof}
        Let $Q$ be a spinal path of $T$.
        Recall that $Q$ is the concatenation of $\mathcal{O}(k)$ paths of the three kinds described above. For every such subpath $Q'$ of $Q$, 
        we distinguish three cases, corresponding to the three kinds (a, (b), (c) mentioned above.
        In all these cases there is a 1-1 correspondence between solutions $\pcal$ to $(G,\tcal)$ and solutions $\pcal'$ to $(G',\tcal',\mathcal{D}')$.
        Moreover, due to~\cref{eq:load:inn}, we can assume that for every solution $\mathcal{P}$ to $(G,\tcal)$, it holds that $\minload(\pcal,Q^i)=\mathcal{O}(k)$.

        For case (a), i.e., if $Q'$ is disjoint from $\bigcup_{i\in[\ell]}U_{\mathsf{mid}}^i$, then let $\bar{Q}'$ be the path obtained from $Q'$ after the removal of (the unique) vertex $v_e, e\in F$ that is an endpoint of $Q$. Observe that $\bar{Q}'$ is a subpath of a spinal path of $H$ and therefore a geodesic in $G$. Also, by construction for every solution $\mathcal{P}$ to $(G,\tcal)$, it holds that a path from $\mathcal{P}$ has an edge incident to $\bar{Q}'$ if and only if a path from $\mathcal{P}'$ (in $G'$) has an edge incident to $Q'$ and therefore $\minload(\mathcal{P}',Q')=\minload(\mathcal{P},\bar{Q}')$. Recall that $\bar{Q}'$ is disjoint from $\bigcup_{i\in[\ell]}U_{\mathsf{mid}}^i$ and therefore, due to \Cref{prop:outside}, for every solution $\mathcal{P}$ to $(G,\tcal)$, it holds that $\minload(\pcal,\bar{Q}') = \Oh(k^2)$. Therefore, in this case, $\minload(\mathcal{P}',Q')=\mathcal{O}(k^2)$.

        In case (b) the edge set of $Q'$ is a subset of $E(C_1^i)$ or $E(C_2^i)$ for some $i\in[\ell]$. 
        By \Cref{obs:dag-cut:cross} we know that every solution $\mathcal{P}'$ to $(G',\tcal',\dcal')$ contains at most $k$ vertices of the form $v_e,e\in E(U_{X^i},U_{\mathsf{mid}}^i)$ (resp. $e\in E(U_{\mathsf{mid}}^i,U_{Y^i})$). Therefore, in this case, $\minload(\mathcal{P}',Q')=\mathcal{O}(k)$.

        In case (c), we have that by removing the endpoints of $Q'$ we obtain $Q^i$ for some $i\in[\ell]$. In this case, we argue as follows.
        First note that, by construction for every solution $\mathcal{P}$ to $(G,\tcal)$, it holds that a path from $\mathcal{P}$ has an edge incident to $Q^i$ if and only if a path from $\mathcal{P}'$ has an edge incident to $Q'$ and therefore $\minload(\mathcal{P}',Q')=\minload(\mathcal{P},Q^i)$.
        By~\cref{eq:load:inn}, we have that $\minload(\mathcal{P},Q^i)=\mathcal{O}(k)$ and thus $\minload(\mathcal{P}',Q')=\mathcal{O}(k)$.

        Now suppose that $Q$ is the concatenation of paths $L_1,\ldots,L_p$, where $p=\mathcal{O}(k)$ and for each $i\in[p]$ the path $L_i$  corresponds to one of the cases (a), (b), (c).
        Then, for each $i\in[p]$, $\minload(\mathcal{P}',L_i)=\mathcal{O}(k^2)$.  
        It follows directly from the definition that the combinatorial load of a concatenation of two paths is at most the sum of their combinatorial loads.
        Therefore, $\minload(\mathcal{P}',Q) \le \sum_{i\in[p]} \minload(\mathcal{P},L_i) = \mathcal{O}(k^3)$.  
    \end{innerproof}

\end{proof}

\section{Retrieving a solution via homology}
\label{sec:schrijver}

In this section we complete the proof of \cref{thm:main} by showing how to reduce the considered problem, with the help of \cref{prop:primalSteiner}, to an FPT number of subproblems conforming to the algebraic framework of Schrijver~\cite{schrijver1994finding}. We remark that while the idea of searching for a solution to the {\sc{Disjoint Paths}} problem of a fixed topological shape using the auxiliary algebraic problem of {\sc{Homology Feasibility}} is due to Schrijver, the concrete implementation of the reduction needed here, and in particular the fine bounds on the number of homology classes that need to be considered, are new to this work.

\subsection{Group-labelled graphs}\label{sec:group-prelims}

First, we need to introduce an algebraic language for speaking about the topology of sets of paths in plane-embedded graphs. This language --- of group-labeled graphs --- was used by Schrijver in his work on vertex-disjoint paths in planar digraphs~\cite{schrijver1994finding}. As mentioned,  we will eventually  use some of his results.

\paragraph*{Oriented walks.}Throughout this section we fix $G$ to be a connected multigraph together with an embedding in the plane. Note that in $G$ we allow both parallel edges (edges connecting the same pair of endpoints) and loops (edges connecting a vertex with itself, embedded as closed curves).
Recall that the set of faces of $G$ is denoted by $F(G)$.

An {\em{oriented walk}} in $G$ is a sequence $\vec W=(a_1,a_2,\ldots,a_p)$ of elements of $\Ev(G)$ such that the tail $a_{i+1}$ is the head of $a_i$, for all $i\in \{1,\ldots,p-1\}$. The {\em{reversal}} of $\vec W$ is $\vec W^{-1}=(a_p^{-1},\ldots,a_1^{-1})$. Oriented walks in the dual graph $G^\star$ are defined in the same way. An oriented walk is an {\em{oriented path}} if it visits every vertex of the graph at most once. Every path $P$ in $G$ gives rise to two oriented paths, called the {\em{orientations}} of $P$.

\paragraph*{Groups and labellings.} We will use the multiplicative notation for groups: if $\Gamma$ is a group, the result of applying the binary group operation to elements $a,b\in \Gamma$ will be denoted by $a\cdot b$ or just $ab$, the inverse of an element $a\in \Gamma$ will be denoted by $a^{-1}$, and the identity element of $\Gamma$ will be denoted by $1_{\Gamma}$. 

Suppose $\Sigma$ is a finite alphabet. We define $\oSigma=\{a,a^{-1}\colon a\in \Sigma\}$ and a mapping $\reduce\colon \oSigma^\star\to \oSigma^\star$ as follows: for $w\in \oSigma^\star$, $\reduce(w)$ is obtained by exhaustively removing subwords of the form $aa^{-1}$ or $a^{-1}a$ for $a\in \Sigma$. (It is well-known and easy to see that the result of such exhaustive removal is unique.) Thus, the image of $\reduce$ can be endowed with an associative binary operation $\cdot$ defined as $u\cdot w=\reduce(uw)$, with identity element being the empty word and the inverse operation defined by reverting the word and replacing every symbol by its inverse. Thus, the image of $\reduce$ becomes the {\em{free group}} on generators $\Sigma$, which we will denote by $\langle \Sigma\rangle$. All groups used in this paper will be of this form.

For a group $\Gamma$, a {\em{$\Gamma$-labelling}} of $G$ is a function $\lambda\colon \Ev(G)\to \Gamma$ satisfying
$$\lambda(a^{-1})=(\lambda(a))^{-1}\qquad\textrm{for each arc }a\in \Ev(G).$$
If $\vec W=(a_1,a_2,\ldots,a_p)$ is an oriented walk in $G$, then we define
$$\lambda(\vec W)=\lambda(a_1)\cdot \lambda(a_2)\cdot \ldots\cdot \lambda(a_p).$$
The labelling $\lambda$ naturally gives rise to a $\Gamma$-labelling $\lambda^\star$ of $G^\star$, defined as
$$\lambda^\star(a^\star)=\lambda(a)\qquad\textrm{for each arc }a\in \Ev(G).$$
Then $\lambda^\star$ can be extended to walks in $G^\star$ in the same way as this was done for $G$.

\paragraph*{Characteristic words.} Suppose now that $(G,\calT,\dcal)$ is an instance of {\sc Planar Disjoint Annotated Paths}, where $\calT=\{(s_1,t_1),\ldots,(s_k,t_k)\}$.
We define the alphabet $\Sigma=\Sigma_\calT=\{1,2,\ldots,k\}$, and we let $\Gamma=\Gamma_\calT=\langle \Sigma\rangle$ be the free group with generators $\Sigma$. Then, given a solution $\calP=(P_1,\ldots,P_k)$ to $(G,\calT,\dcal)$, where every $P_i$ is considered an oriented path that starts in $s_i$ and ends in $t_i$, we define the $\Gamma$-labelling $\lambda_\calP$ of $G$ as follows: for an arc $a\in \Ev(G)$, set
\begin{itemize}[nosep]
\item $\lambda_\calP(a)=i$ if some $P_i$ traverses $a$ respecting its orientation;
\item $\lambda_\calP(a)=i^{-1}$ if some $P_i$ traverses $a$ disrespecting its orientation; and
\item $\lambda_\calP(a)=1_\Gamma$ if no $P_i$ traverses $a$.
\end{itemize}
It is straightforward to check that thus, $\lambda_{\calP}$ is a $\Gamma$-labelling of $G$.

Now, if $\vec W$ is an oriented walk in the dual $G^\star$, then we define the {\em{characteristic word}} of $\calP$ with respect to $\vec W$ as
$$\charWord(\calP,\vec W)=\lambda^\star_\calP(\vec W).$$
Intuitively speaking, $\charWord(\calP,\vec W)$ gathers information about the crossings of $\vec W$ by the paths of $\calP$ when $\vec W$ is traversed along its orientation. Every counter-clockwise crossing with $P_i$ contributes with~$i$, every clockwise crossing with $P_i$ contributes with $i^{-1}$, and consecutive occurrences of $i$ and $i^{-1}$ cancel out. Note that $\charWord(\calP,\vec W^{-1})=\charWord(\calP,\vec W)^{-1}$.

An oriented walk is an {\em{oriented cycle}} if it begins and ends at the same vertex, and otherwise all the vertices traversed by it are pairwise different and different from the beginning. An oriented cycle is {\em{non-separating}} if one of the two regions in which it divides the plane does not contain any terminals from $V_\calT$. We will use the following straightforward claim, which can be imagined as a discrete analogue of Green's~Theorem. 

\begin{lemma}\label{lem:Green}
    If $(G,\tcal,\dcal)$ is a nice instance of {\sc{Planar Annotated Disjoint Paths}} and $C$ is an oriented non-separating cycle in $G^\star$, then $\charWord(\calP,C)=1_{\Gamma_\tcal}$.
\end{lemma}
\begin{proof}
    Choose a starting vertex anywhere on $C$ and let $a_1,\ldots,a_p$ be the sequence of consecutive arcs of $C$. Further, let $s_1,\ldots,s_q$ be the sequence obtained from $a_1,\ldots,a_p$ by replacing every arc $a$ with $\lambda_{\calP}^\star(a)$, and removing all the entries equal to $1_{\Gamma_\tcal}$; thus $\charWord(\calP,C)=s_1s_2\ldots s_q$. The segments of paths of $\calP$ contained in the region delimited by $C$ that does not contain terminals naturally define a matching on the entries of the sequence $s_1,\ldots,s_q$, which matches opposite symbols (i.e. $i$ with $i^{-1}$). Since the paths of $\calP$ are pairwise disjoint, replacing within every matched pair the left element by $($ and the right element by $)$, turns the sequence $s_1,\ldots,s_q$ into a well-formed bracket expression. Since the matched brackets correspond to symbols that cancel each other, the whole expression $s_1s_2\ldots s_q$ evaluates in $\Gamma_\tcal$ to $1_{\Gamma_\tcal}$.
\end{proof}

\subsection{Topological load and algebraic load}
\label{sec:topo-load}

Before we continue, we need to make a reinterpretation of the notion of the minimum load of a path $Q$ with respect to a family of paths $\calP$ (see \cref{def:homotopy:load}). Intuitively, it will be more convenient to work with a definition that considers $Q$ fixed and applies transformations to $\calP$, rather than the other way round, as was the case in \cref{def:homotopy:load}. Also, we need $Q$ to be a path in the dual graph $G^\star$, rather than in $G$ itself.

\paragraph*{Load of dual paths.} First, we need to extend the notions of (minimum) load also to paths in the dual graph $G^\star$. This is done in a natural manner as follows. Suppose $(G,\tcal)$ is a nice instance of \pdsp and $Q$ is a path in $G^\star$. We define (empty) handles of $Q$ as before, except that a handle of $Q$ is a path in~$G^\star$ instead of in $G$. This gives an obvious extension of the notions of pull operations and of homotopy to paths in $G^\star$ (see \cref{def:homotopy}). Now, for a path $Q$ in $G^\star$ and a solution $\calP$ of $(G,\tcal)$, we define $\load(\calP,Q)$ to be the number of edges of $Q$ that cross an edge of $\bigcup_{P\in \calP} E(P)$, and again $\minload(\calP,Q)$ is defined as the minimum $\load(\calP,\widehat{Q})$ among paths $\widehat{Q}$ homotopic to $Q$.

\paragraph*{Topological load.} First, we introduce a notion of load that is almost equivalent to that of \cref{def:homotopy:load}, but relies on untangling the path family $\calP$ rather than modifying the path $Q$. The notion is defined in topological terms, so we need a few definitions.

A {\em{curve}} in the plane is a homeomorphic image of the interval $[0,1]$. Two curves $P,Q$ are in {\em{general position}} if $P\cap Q$ consists of a finite number of points, none of which is an endpoint of $P$ or $Q$, and for every point $x\in P\cap Q$ there is an open neighborhood $N$ of $x$ such that $N-P$ has two connected components, both intersecting $Q$, and $N-Q$ has two connected components, both intersecting $P$ (intuitively, $P$ and $Q$ cross at $x$ so that they appear in an interlacing manner around $x$).

For a positive integer $k$, a {\em{$k$-configuration}} is a pair $\calC=(\calP,Q)$, where $Q$ is a curve, $\calP$ is a family of $k$ pairwise disjoint curves, and $Q$ is in general position with each $P\in \calP$. Thus, $\calC$ naturally gives rise to a plane multigraph $H_\calC$ defined as follows:
\begin{itemize}[nosep]
    \item The vertex set of $H_\calC$ consists of the endpoints of the curves of $\calP\cup \{Q\}$ and the intersection points of $Q$ with the paths of $\calP$.
    \item Let a {\em{segment}} be a maximal subcurve of $Q$ or any $P\in \calP$ that does not contain any vertex of $H_\calC$ in its interior; there are naturally $Q$-segments (contained in $Q$) and $\calP$-segments (contained in paths of $\calP$). For every segment $S$, introduce an edge $e_S$ to $H_\calC$ connecting the two endpoints of $S$; thus $H_\calC$ naturally has $Q$-edges and $\calP$-edges. The segment $S$ is deemed to be the embedding of $e_S$ in the plane, thus making $H_\calC$ a plane multigraph. 
\end{itemize}

Now, we consider the following operation of {\em{simplification}}: as long as $H_\calC$ contains a face of length~$2$ homeomorphic to an open disc, say with vertices $u,v$, replace the three-edge path consisting of $Q$-edges incident to $u$ and $v$ with a single $Q$-edge with the same endpoints, and replace the three-edge path consisting of $\calP$-edges incident to $u$ with a single $\calP$-edge with the same endpoints, so that the new edges do not intersect; see \Cref{fig:simplification}. It is straightforward to verify that the plane multigraph obtained by applying the simplification operation exhaustively is defined uniquely and does not depend on the order of making simplifications; we call it $\widetilde{H}_\calC$. We define the {\em{topological load}} of configuration $\calC$, denoted $\topload(\calC)$ as the number of vertices of $\widetilde{H}_\calC$ that are not the endpoints of the paths of $\calP\cup \{Q\}$.

\begin{figure}[t]
\centering
\includegraphics[scale=0.4]{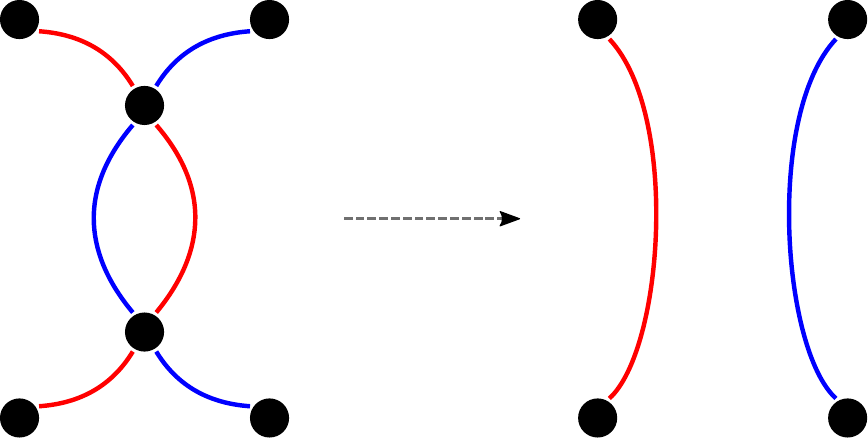}
\caption{
A visualization of the simplification operation. $Q$-edges are depicted in blue and $\calP$-edges are depicted in red.
}\label{fig:simplification}
\end{figure}

The following lemma links the topological load with the notion of combinatorial load considered in \cref{def:homotopy:load}.

\begin{lemma}\label{lem:topcomb}
    Let $(G,\calT,\dcal)$ be an instance of \padpfull, $\calP$ be a solution to $(G,\calT,\dcal)$, and $Q$ be a path in $G^\star$. Then $(\calP,Q)$ is a configuration and $\topload(\calP,Q)\leq \minload(\calP,Q)$.
\end{lemma}
\begin{proof}
    That $(\calP,Q)$ is a configuration is obvious. For the claimed inequality, the key observation is the following: if $\widehat{Q}$ is homotopic to $Q$, then $\widetilde{H}_{(\calP,Q)}=\widetilde{H}_{(\calP,\widehat{Q})}$. To see this statement, first note that it suffices to give a proof under the assumption that $Q$ and $\widehat{Q}$ differ by one pull operation, which moreover involves an empty handle that delimits a single face of $G^\star$. Indeed, every pull operation can be presented as a composition of pull operations as above. Now, under this assumption, it is easy to see that either $H_{(\calP,Q)}$ and $H_{(\calP,\widehat{Q})}$ are equal, or they differ by one simplification operation. In both cases, we have $\widetilde{H}_{(\calP,Q)}=\widetilde{H}_{(\calP,\widehat{Q})}$.

    Coming back to the proof of the claimed inequality, let $\widehat{Q}$ be the path homotopic to $Q$ such that $\load(\calP,\widehat{Q})=\minload(\calP,Q)$. By the observation of the previous paragraph, we have $\widetilde{H}_{(\calP,Q)}=\widetilde{H}_{(\calP,\widehat{Q})}$, hence $\topload(\calP,Q)=\topload(\calP,\widehat{Q})$. Finally, we clearly have $\load(\calP,\widehat{Q})\geq \topload(\calP,\widehat{Q})$, because $H_{(\calP,\widehat{Q})}$ has at least as many vertices as $\widetilde{H}_{(\calP,\widehat{Q})}$. By combining the three relations above we conclude that $\topload(\calP,Q)\leq \minload(\calP,Q)$, as claimed.
\end{proof}

\paragraph*{Algebraic load.} Next, we introduce another notion of load. This time, the notion is tailored to the toolbox of Schrijver, and therefore is algebraic in nature.

Let $(G,\calT,\dcal)$ be a nice instance of \padpfull, $\calP$ be a solution to this instance, and $Q$ be a path in $G^\star$. We define the {\em{algebraic load}} of $Q$ with respect to $\calP$, denoted $\algload(\calP,Q)$, as
$$\algload(\calP,Q)=|\charWord(\calP,\vec Q)|,$$
where $\vec Q$ is any orientation of $Q$. (The words $\charWord(\calP,\vec Q)$ for the two orientations are inverses of each other, hence they have the same length.) We have the following easy observation.

\begin{lemma}\label{lem:algtop}
    For every nice instance $(G,\calT,\dcal)$ of \padpfull, solution $\calP$ to this instance, and path $Q$ in $G^\star$, we have $\algload(\calP,Q)\leq \topload(\calP,Q)$.
\end{lemma}
\begin{proof}
    Let $\vec Q$ be any orientation of $Q$ and let $w$ be the word over $\overline{\Sigma}_\tcal$ consisting of the indices or their inverses of consecutive paths of $\calP$ crossed along $\vec Q$, so that $\charWord(\calP,\vec Q)=\reduce(w)$. Observe that in the definition of the topological load of the configuration $(\calP,Q)$, every consecutive operation of simplification corresponds to removing some pair of consecutive symbols of the form $ii^{-1}$ or $i^{-1}i$ from $w$. Therefore, the number of non-endpoint vertices left in $\widetilde{H}_{(\calP,Q)}$ after applying simplifications exhaustively cannot be smaller than the number of symbols left in $w$ after applying the cancellations exhaustively (i.e., applying the $\reduce(\cdot)$ operator). It follows that $\algload(\calP,Q)=|\reduce(w)|\leq \topload(\calP,Q)$, as claimed.
\end{proof}

We note that the inequality provided by \cref{lem:algtop} cannot be reversed: it is possible that the algebraic load of a path $Q$ is significantly smaller than its topological load. We leave the easy task of finding examples to the reader.

\subsection{Skeletons}

We now gather all the observations obtained so far to prove a key proposition: given an instance of \padpfull, we can compute a ``skeleton'' --- a tree in the dual connecting all the terminals --- with the property that if there is a solution, then there is also a solution that crosses this skeleton a small number of times. The number of crossings will be expressed in terms of the topological load. The reader might think of skeletons as of the terminal-connecting trees provided by \cref{prop:primalSteiner}, except for technical modifications that make them live in the dual of the graph, rather than in the primal.

We proceed to formal details.
Let $(G,\calT,\dcal)$ be a nice instance of \padpfull. Recall that then each terminal  $t$ from $\calT$ has degree $1$ in $G$, hence there is a single face $f_t\in F(G)$ incident to $t$. Let $F_\calT$ be the set of those {\em{terminal faces}} $f_t$ for terminal vertices $t$ from $\tcal$. A {\em{skeleton}} of $(G,\calT,\dcal)$ is an inclusion-wise minimal tree $K$ that is a subgraph of $G^\star$ and contains all the faces of $F_\calT$. A vertex $f$ of $K$ is called {\em{principal}} if it belongs to $F_\calT$ or has degree at least $3$ in $K$, and {\em{non-principal}} otherwise. A {\em{spinal}} path of $K$ is a maximal path in $K$ whose all internal vertices are non-principal. The set of all spinal paths of $K$ will be denoted by $\calQ(K)$, and we let $\vec{\calQ}(K)$ be the set of all orientations of spinal paths of $K$. Note that if $|\calT|=k$, then $|\calQ(K)|\leq 4k-3$ and $|\vec \calQ(K)|\leq 8k-6$.

Now if $\calP$ is a solution to $(G,\calT,\dcal)$, then the {\em{topological load}} of a skeleton $K$ with respect to $\calP$ is defined as the maximum topological load (with respect to $\calP$) among the spinal paths of $K$. The algebraic load of $K$ with respect to $\calP$ is defined analogously.

By a slight abuse of notation, throughout the rest of this section, we see every collection $\mathcal{D}$ of oriented subgraphs as mapping from arcs to sets of indices.
Precisely, 
 given a collection $\mathcal{D}=(D_1,\ldots,D_k)$ of oriented subgraphs of a graph $G$, for every arc $a\in \vec{E}(G)$ we use $\mathcal{D}(a) \sub \oSigma_\tcal$ to denote the set defined as follows:
 for each $i\in[k]$ we put $i\in\mathcal{D}(a)$ if $a$ is an arc in $D_i$.

With these definitions in place, we can state and prove our proposition.

\begin{proposition}\label{prop:niceSteiner}
    There exists an algorithm that, given a nice instance $(G,\tcal)$ of \pdsp with $k = |\tcal|$, runs in time $2^{\Oh(k)}\cdot n^{\Oh(1)}$, and either correctly concludes that $(G,\tcal)$ has no solution, or outputs a nice instance $(G',\tcal',\dcal)$ of \padpfull with $|\tcal'|=k$ and a skeleton $K$ of $(G',\calT',\dcal)$, with the following properties: 
    \begin{itemize}[nosep]
    \item If $(G,\tcal)$ has a solution, then $(G',\tcal',\dcal)$ has a solution $\calP$ such that the topological load of $K$ with respect to $\calP$ is bounded by $\Oh(k^3)$.
    \item If $(G',\tcal',\dcal)$ has a solution, then so does $(G,\tcal)$.
    \end{itemize}
\end{proposition}
\begin{proof}
    We first apply the algorithm of \cref{prop:primalSteiner}, thus obtaining either a report that the instance $(G,\tcal)$ has no solution, or a nice instance $(G^\circ,\tcal^\circ,\dcal^\circ)$ of \padpfull together with a tree $T$ spanning all terminals of $\tcal^\circ$, satisfying the assertions stated in \cref{prop:primalSteiner}. Our goal is to suitably modify the instance $(G^\circ,\tcal^\circ,\dcal^\circ)$ to an equivalent instance $(G',\tcal',\dcal)$ so that the tree $T$, which is a subgraph of $G$, becomes a skeleton $K$, which will be a subgraph of the dual $G^\star$.

\begin{figure}[t]
\centering
\includegraphics[scale=0.6]{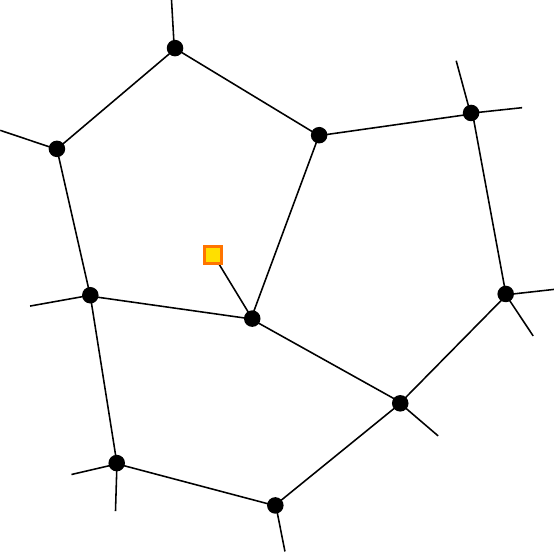}
\qquad\qquad\qquad
\includegraphics[scale=0.6]{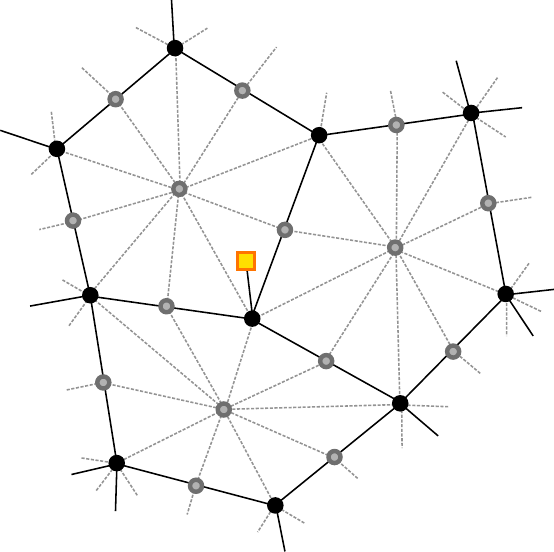}
\caption{Construction of the graph $G'$ (right) from the graph $G^\circ$ (left). Features present in $G'$ but not in $G^\circ$ are depicted in gray. The orange square is a terminal.}\label{fig:todual}
\end{figure}

    The graph $G'$ is obtained from $G^\circ$ as follows (see \Cref{fig:todual}):
    \begin{itemize}[nosep]
        \item Temporarily remove all the terminals $s_i$ and $t_i$, together with the incident edges. We let $f_{s_i}$ be the face of the obtained graph in which $s_i$ was embedded, and $u_{s_i}$ be the unique neighbor of $s_i$; and similarly for $t_i$. 
        \item Subdivide every edge $e$ of the obtained graph, using a new vertex $v_e$.
        \item For every face $f$ introduce a new vertex $v_f$. For every incidence between $f$ and a vertex $u$ lying on the boundary of $f$, introduce a new edge between $v_f$ and $u$. (Thus, if $u$ appears $\ell$ times on the facial walk of $f$, there are $\ell$ edges between $v_f$ and $u$.) Similarly, for every incidence between $f$ and an edge $e$ lying on the boundary of $f$, introduce a new edge between $v_f$ and $v_e$. (Again, if $f$ appears on both sides of $e$, there are two edges between $v_f$ and $v_e$.) The edges incident to $v_f$ are embedded naturally within the face $f$.
        \item Reintroduce the terminals $s_i$ and $t_i$ together with incident edges so that each $s_i$ is embedded within one of the two new faces contained in $f_{s_i}$ incident to $u_{s_i}$; and similarly for $t_i$.
    \end{itemize}
    The obtained graph shall be called $G'$.
    Thus, every face of length $\ell$ of $G^\circ$ with terminals removed gets divided into $2\ell$ faces of $G'$. With the exception of faces accommodating terminals, every face of $G'$ is a triangle whose one vertex is of the form $v_f$, and the opposing side is one half of a subdivided edge of $G^\circ$.
    
    Note that the terminals of $\tcal^\circ$ are also present in $G'$, so we can simply set $\tcal'=\tcal^\circ$.
    As for the mapping~$\dcal$, we define it as follows:
    \begin{itemize}[nosep]
        \item For every arc $a\in \Ev(G^\circ)$, we set $\dcal(a_1)=\dcal(a_2)=\dcal^\circ(a)$, where $a_1$ and $a_2$ are the two arcs of $G'$ resulting from the subdivision of $a$.
        \item For every arc $a\in \Ev(G')$ incident to a vertex of the form $v_f$, we set $\dcal(a)=\emptyset$. 
    \end{itemize}
    Thus, $(G',\tcal',\dcal)$ differs from $(G^\circ,\tcal^\circ,\dcal^\circ)$ only by subdividing the edges non-adjacent to the terminals and introducing some additional edges that anyway cannot accommodate any paths from the solution. It follows that the solutions to the instance $(G',\tcal',\dcal)$ correspond to the solutions to the instance $(G^\circ,\tcal^\circ,\dcal^\circ)$ in a one-to-one manner, where the correspondence involves subdividing every edge non-incident to a terminal.

    \newcommand{\fw}{\mathsf{follow}}

    It remains to define the skeleton $K$ and justify its properties. For this, consider the following construction. Let $\vec Q$ be an oriented path in $G^\circ$. We define an oriented path $\fw(\vec Q)$ in the dual $(G')^\star$ as follows:
    \begin{itemize}[nosep]
    \item Let $\vec Q'$ be the oriented path in $G'$ obtained from $\vec Q$ by replacing every arc $a$ traversed by $\vec Q$ by the two arcs in $G'$ into which $a$ gets subdivided.
    \item Let $a$ and $a'$ be the first and the last arc of $\vec Q'$, respectively. Let $f$ and $f'$ be the faces of $G'$ incident to $a$ and $a'$, respectively, lying on the clockwise side of $a$ and $a'$, respectively. Within the interiors of $f$ and $f'$, pick points $p$ and $p'$ at very close proximity to $a$ and $a'$, respectively.
    \item Draw a curve $\gamma$ that starts at $p$, ends at $p'$ and follows along $\vec Q$ in close proximity, but never touches it. See \Cref{fig:follow} for a visualization.
    \item We define $\fw(\vec Q)$ to be sequence of arcs of $(G')^\star$ dual to the consecutive edges of $G'$ crossed by $\gamma$, oriented so that $\fw(\vec Q)$ leads from $f$ to $f'$.
    \end{itemize}

\begin{figure}[t]
\centering
\includegraphics[scale=0.6]{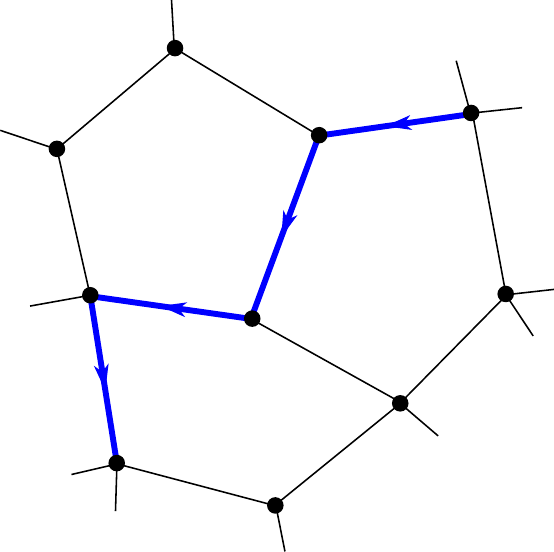}
\qquad\qquad\qquad
\includegraphics[scale=0.6]{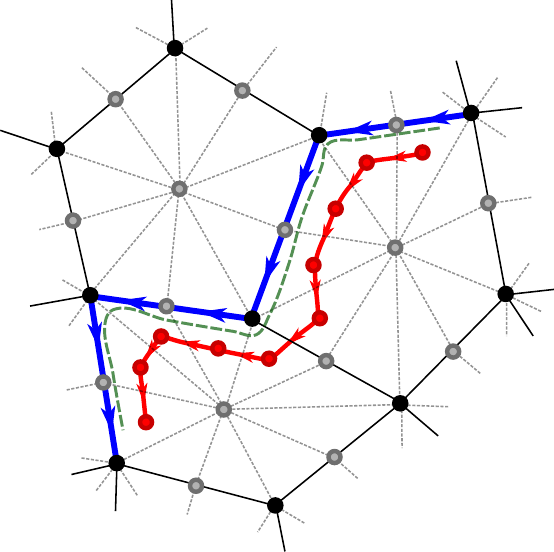}
\caption{An oriented path $\vec Q$ in $G$ (in blue) and the oriented path $\fw(\vec Q)$ in $G^\star$. The curve $\gamma$ is depicted in dashed green. }\label{fig:follow}
\end{figure}

    The next statement is implied directly by the construction; we leave the easy verification to the~reader.

    \begin{claim}\label{cl:subd-disjoint}
        Let $\vec Q, \vec R$ be two oriented paths in $G^\circ$ that are vertex-disjoint, except for possibly sharing an endpoint. Then $\fw(\vec Q)$ and $\fw(\vec R)$ are vertex-disjoint paths in $(G')^\star$.
    \end{claim}

    Let $K_0$ be the subgraph of $(G')^\star$ defined as the union of the following subgraphs:
    \begin{itemize}[nosep]
        \item For every spinal path $Q$ of $T$, the path $\fw(\vec Q)$ with orientation forgotten, for an arbitrary orientation $\vec Q$ of $Q$.
        \item For every vertex $v$ that has degree at least $3$ in $T$, the cycle $C_v$ in $(G')^\star$ consisting of the duals of the edges of $G'$ incident to $v$.
    \end{itemize}
    Finally, let $K$ be an inclusion-wise minimal connected subgraph of $K_0$ that contains every face incident to a terminal from $\calT$. Note that $K$ is a skeleton. Moreover, from \cref{cl:subd-disjoint} it follows that every spinal path $Q'$ of $K$ is either contained in some cycle $C_v$, or it arises from the path $\fw(\vec Q)$ for some orientation $\vec Q$ of a spinal path $Q$ of $T$ by possibly attaching parts of the cycles $C_v$ and $C_w$, where $v$ and $w$ are the endpoints~of~$Q$.

    It remains to argue that for every solution $\calP$ to $(G',\calT',\dcal)$, the topological load of $K$ with respect to $\calP$ is bounded by $\Oh(k^3)$. In the following, by a slight abuse of notation we will consider $\calP$ to be also a solution to $(G^\circ,\calT^\circ,\dcal^\circ)$, as solution paths in $(G^\circ,\calT^\circ,\dcal^\circ)$ correspond to solution paths in $(G',\calT',\dcal')$ through subdivision. Hence, for every spinal path $Q'$ of $K$, we need to prove that the topological load of $Q'$ with respect to $\calP$ is bounded by $\Oh(k^3)$.
    
    If $Q'$ is entirely contained in one cycle $C_v$, then all the edges of $G'$ dual to the edges of $Q'$ are incident to $v$. Hence, there are at most two crossings between $Q'$ and the paths of $\calP$. It follows that the load of $Q'$ with respect to $\calP$ is at most $2$, which implies that also the topological load of $Q'$ with respect to $\calP$ is at most $2$.

    Suppose now that $Q'$ is the concatenation of $R_v,Q'',R_w$, where
    \begin{itemize}[nosep]
        \item $Q''$ is $\fw(\vec Q)$ with orientation forgotten, where $\vec Q$ is an orientation of a spinal path $Q$ of $T$, and
        \item paths $R_v$ and $R_w$ (possibly empty) are contained in the cycles $C_v$ and $C_w$, where $v$ and $w$ are the endpoints of $Q$.
    \end{itemize}
    The same argument as the one used in the previous paragraph shows that both $R_v$ and $R_w$ have load at most $2$ with respect to $\calP$. Hence, it suffices to show that the topological load of $Q''$ with respect to $\calP$ is bounded by $\Oh(k^3)$. Recall here that $\minload(\calP,Q)\leq \Oh(k^3)$.

    Consider now a path $S$ homotopic to $Q$ such that $\load(\calP,S)=\minload(\calP,Q)$, and let $\vec S$ be the orientation of $S$ so that the start of $\vec S$ coincides with the start of $\vec Q$. Let $S''$ be $\fw(\vec S)$, with the orientation forgotten. Note that $\load(\calP,S'')\leq \load(\calP,S)\leq \Oh(k^3)$. We observe that one can find paths $S_v$, contained in $C_v$, and $S_w$, contained in $C_w$, so that the concatenation of $S_v$, $S''$, and $S_w$ is a path $S'''$ that is homotopic to $Q''$. Indeed, every pull operation applied to $S$ in $G$ can be relayed to a pull operation applied to $S'''$ in $G^\star$, followed by an appropriate splitting of $S'''$ into $S_v$, $S_w$, and $S''$ so that $S''$ stays equal to $\fw(\vec S)$ (with orientation forgotten). Again, the load of $S_v$ and $S_w$ with respect to $\calP$ is at most $2$, and moreover we have that $\load(\calP,S)\leq \Oh(k^3)$. By \cref{lem:topcomb}, we conclude that
    $$\topload(\calP,Q'')\leq \minload(\calP,Q'')\leq \load(\calP,S''')\leq 4+\load(\calP,S'')\leq 4+\load(\calP,S)\leq \Oh(k^3),$$
    as required.
\end{proof}

\subsection{Enumerating the homology classes}
\label{sec:homology}

In this section we describe how the algebraic framework of Schrijver~\cite{schrijver1994finding} can be used in combination with \cref{prop:niceSteiner} to solve the \pdspfull in fpt time, thereby proving \cref{thm:main}. First, we need to recall the said framework, which requires several auxiliary definitions.

\subsubsection{Homology Feasibility}

In the {\sc{Homology Feasibility}} problem we will be working with plane $\Gamma$-labelled multigraphs, as discussed in \cref{sec:group-prelims}, where $\Gamma=\langle \Sigma\rangle$ is a fixed free group over some finite set of generators $\Sigma$. Recall that this means that the elements of $\Gamma$ are words over $\oSigma=\{a,a^{-1}\colon a\in \Sigma\}$ that are {\em{reduced}}, that is, do not contain any symbol and its inverse in succession.
For two elements $a,b$ of $\Gamma$, we say that $a$ is a {\em{subword}} of $b$ if there exist $c,d\in \Gamma$ so that $b$ is the concatenation of $c$, $a$, $d$, in this order. Note that this entails $b=cad$ in the group $\Gamma$, but is not the same. We say that a set $A\subseteq \Gamma$ is {\em{hereditary}} if whenever $b\in A$, all the subwords of $b$ belong to $A$ as well. Also, we will denote $A^{-1}=\{a^{-1}\colon a\in A\}$. Note that if $A$ is hereditary, so is also $A^{-1}$.

Let $G$ be a connected plane multigraph. A {\em{$\Gamma$-shift}} of $G$ is a function $\psi\colon F(G)\to \Gamma$. (Recall here that $F(G)$ is the set of faces of $G$.)
Given a $\Gamma$-shift $\psi$ and a $\Gamma$-labelling $\lambda$ of $G$, we may apply $\psi$ to $\lambda$ to obtain a new $\Gamma$-labelling $\psi[\lambda]$ defined as follows: for all $a\in \Ev(G)$, if $a^\star=(f,f')$, then we set
$$\psi[\lambda](a)=\psi(f)\cdot \lambda(a)\cdot \psi(f')^{-1}\qquad\textrm{for every arc }a\in \Ev(G).$$
It is straightforward to verify that $\psi[\lambda]$ is again a $\Gamma$-labelling of $G$.

Next, for the purpose of setting restrictions on the sought labellings, we introduce the following definition. A {\em{domain mapping}} is a mapping that assigns to every arc $a\in \Ev(G)$ its {\em{domain}} $\Delta(a)$ and to each vertex $v\in V(G)$ its {\em{domain}} $\Delta(v)$ so that 
\begin{itemize}[nosep]
    \item $\Delta(a^{-1})=\Delta(a)^{-1}$ for each $a\in \Ev(G)$,
    \item $\Delta(v)=\Delta(v)^{-1}$ for each $v\in V(G)$, and
    \item hereditary sets $\Delta(x)$ for all $x\in V(G)\cup \Ev(G)$.
\end{itemize}  We say that a $\Gamma$-labelling $\lambda$ {\em{conforms}} to the domain mapping $\Delta$ if
\begin{itemize}[nosep]
    \item for every $a\in \Ev(G)$, we have $\lambda(a)\in \Delta(a)$, and
    \item for every vertex $v\in \Ev(G)$ and a sequence of arcs $a_1,\ldots,a_p\in \Ev(G)$ with heads $v$ that are consecutive in the cyclic order around $v$, we have $\lambda(a_1)\lambda(a_2)\ldots \lambda(a_p)\in \Delta(v)$.
\end{itemize}
For simplicity we assume that the domains are always given on input explicitly, as list of elements, even though in the framework of Schrijver one can even only assume oracle access to those sets.

Now, an instance of the {\sc{Homology Feasibility}} problem consists of the following:
\begin{itemize}[nosep]
 \item a connected plane graph $G$;
 \item a $\Gamma$-labelling $\lambda$ of $G$, where $\Gamma$ is a free group on $k$ given generators; and
 \item a domain mapping $\Delta$; and
 \item a set of {\em{fixed faces}} $R\subseteq F(G)$.
\end{itemize}
The question is whether there exists a $\Gamma$-shift $\psi$ of $G$ such that
\begin{itemize}[nosep]
    \item $\psi(f)=1_\Gamma$ for every fixed face $f\in R$; and
    \item $\psi[\lambda]$ conforms to $\Delta$.
\end{itemize}
Any $\psi$ satisfying the first condition above will be called {\em{$R$-stable}}.
The following statement summarizes the results of Schrijver that will be of use for us.

\begin{theorem}[follows from \cite{schrijver1994finding}]\label{thm:Schrijver}
The {\sc{Homology Feasibility}} problem can be solved in time $N^{\Oh(1)}$, where $N$ is the total length of the bit encoding of the input.
\end{theorem}

A few words of explanation are in place, as the work of Schrijver~\cite{schrijver1994finding} does not contain a precise statement as above. The method of Schrijver relies on a reduction of the {\sc{Homology Feasibility}} problem on a graph $G$ to the {\sc{Cohomology Feasibility}} problem on its dual graph $G^\star$. The planarity of input is only used in this reduction: Schrijver's polynomial-time algorithm for {\sc{Cohomology Feasibility}} works on general graphs. In Schrijver's setting, a priori there is no domain mapping $\Delta$ defined on vertices, but in the proof (see the proof of \cite[Proposition~5]{schrijver1994finding}) he extends the dual graph to an ``extended dual'', where he introduces additional arcs between any pair of faces incident to a common vertex $u$ and puts non-trivial constraints on those arcs in order to emulate the vertex-disjointness constraints for the sought paths. Setting those constraints to $\Delta(u)$ in our notation exactly corresponds to the restrictions expressed in our definition of {\sc{Homology Feasibility}}. Similarly, in Schrijver's setting, there is only one fixed face $r$ that is required to have value $1_\Gamma$ in the sought shift. However, we can again easily extend this to the possibility of requiring that every face from a given set $R$ is mapped to $1_\Gamma$ on the level of {\sc{Cohomology Feasibility}}, by picking any $r\in R$, asking it to be mapped to $1_\Gamma$ as described in the setting of Schrijver, and in the instance of {\sc{Cohomology Feasibility}} adding arcs $(r,r')$ for all $r'\in R\setminus \{r\}$ with $\lambda(r,r')=1_\Gamma$ and $\Delta(r,r')=\{1_\Gamma\}$. These additional arcs impose the condition that all faces of $R$ are assigned the same value in the shift.

\paragraph*{Clean vertices.} Before we continue, let us introduce an important invariant of labellings that is maintained under the application of shifts. Suppose $G$ is a connected plane multigraph and $\lambda$ is a $\Gamma$-labelling of $G$, where $\Gamma$ is a group. 

We call a vertex $u$ of $G$ {\em{clean}} under $\lambda$ if
$$\lambda(a_1)\cdot \lambda(a_2)\cdot \ldots \cdot \lambda(a_p)=1_\Gamma,$$
where $(a_1,\ldots,a_p)$ is the sequence of all arcs with head $u$, ordered clockwise around $u$. Note that there are $p$ such sequences depending on the choice of the arc to start the sequence, but the definition does not depend on this choice: the product above will be equal to $1_\Gamma$ for one of them if and only if it is equal to $1_\Gamma$ for all of them.

The observation is that the cleanliness of a vertex is preserved under applying a shift.

\begin{lemma}\label{lem:clean-invariance}
    Let $\Gamma$ be a group, $G$ be a connected plane multigraph, $\lambda$ be a $\Gamma$-labelling of $G$, and $\psi$ be a $\Gamma$-shift of $G$. Suppose a vertex $u$ is clean under $\lambda$. Then $u$ is also clean under $\psi[\lambda]$.
\end{lemma}
\begin{proof}
    Let $(a_1,\ldots,a_p)$ be the sequence of arcs with head $u$ ordered clockwise around $u$, and let $(f_1,\ldots,f_p)$ be the sequence of faces incident to $u$ so that $f_i$ lies between $a_i$ and $a_{i+1}$ in the sequence $(a_1,\ldots,a_p)$ (indices behave cyclically). Then
    \begin{align*}
        \psi[\lambda](a_1)\cdot \psi[\lambda](a_2)\cdot \ldots \cdot \psi[\lambda](a_p)& =\psi(f_p)\lambda(a_1)\psi(f_1)^{-1}\cdot \psi(f_1)\lambda(a_2)\psi(f_2)^{-1}\cdot \ldots \cdot \psi(f_{p-1})\lambda(a_p)\psi(f_p)^{-1} \\
        & = \psi(f_p)\cdot \lambda(a_1)\cdot \lambda(a_2)\cdot \ldots \cdot \lambda(a_p)\cdot \psi(f_p)^{-1}\\
        & =\psi(f_p)\cdot \psi(f_p)^{-1}=1_\Gamma,
    \end{align*}
    hence $u$ is clean under $\psi[\lambda]$.
\end{proof}

And we also note that in labellings originating from solutions to the \padpfull problem, all the non-terminal vertices are clean.

\begin{lemma}\label{lem:solution-clean}
    Suppose $(G,\calT,\dcal)$ is a nice instance of {\sc{Planar Annnotated Disjoint Paths}} and $\calP$ is a solution to $(G,\calT,\dcal)$. Then every non-terminal vertex of $G$ is clean with respect to $\lambda_\calP$.
\end{lemma}
\begin{proof}
    Suppose $u$ is a non-terminal vertex of $G$. Then under $\lambda_\calP$, either all the arcs with head $u$ are mapped to $1_\Gamma$, or all but two are mapped to $1_\Gamma$, and the remaining two are mapped to $i$ and $i^{-1}$ for some $i\in \Sigma_\calT$. In both cases, clearly $u$ is clean.
\end{proof}

Finally, we observe the following statement about solutions to carefully prepared instances of {\sc{Homology Feasibility}} implying the existence of solutions to \padpfull. We remark that the used trick to encode vertex-disjointness constraints is due to Schrijver, see~\cite[Proof of Proposition~5]{schrijver1994finding}.

\begin{lemma}\label{lem:solHom}
    Let $(G,\calT,\dcal)$ be a nice instance of \padpfull, where $|\calT|=k$, and $(G,\lambda,\Delta,R)$ be an instance of {\sc{Homology Feasibility}} over the group $\Gamma=\Gamma_\calT$, satisfying the following properties:
    \begin{itemize}[nosep]
        \item For every terminal pair $(s_i,t_i)\in \calT$, we have $\lambda(a(s_i))=\lambda(a(t_i))=i$, where $a(s_i)$ is the unique arc with tail $s_i$ and $a(t_i)$ is the unique arc with head $t_i$.
        \item Every non-terminal vertex of $G$ is clean with respect to $\lambda$.
        \item $\Delta(v)\subseteq \{1_{\Gamma_\calT}\}\cup \{i,i^{-1}\colon i\in \Sigma_\tcal\}$ for all $v\in V(G)$. 
        \item $\Delta(a)\subseteq \{1_{\Gamma_\calT}\}\cup \{i\colon i\in \dcal(a)\}\cup \{i^{-1}\colon i\in \dcal(a^{-1})\}$ for all $a\in \Ev(G)$.
        \item $F_\calT\subseteq R$.
    \end{itemize}
    Then if $(G,\lambda,\Delta,R)$ has a solution, then so does $(G,\calT,\dcal)$.
\end{lemma}
\begin{proof}
    Let $\psi$ be the $R$-stable $\Gamma$-shift that is a solution to $(G,\lambda,\Delta,R)$, and let $\lambda'=\psi[\lambda]$. Next, for each $i\in \Sigma_\tcal$, let $H_i$ be the directed graph consisting of all arcs $a\in \Ev(G)$ with $\lambda'(a)=i$. Note that for every arc $a\in \Ev(G)$, the word $\lambda'(a)$ has length at most $1$ by the assumption about the domain mapping $\Delta$. Therefore, for every edge $e$ of $G$, either both orientations of $e$ receive $1_\Gamma$ under $\lambda'$ and none them belongs to any of the graphs $H_i$, or these orientations receive $i$ and $i^{-1}$, in which case one of them belongs to $H_i$. 
    
    We first observe that by the assumed property of the domain mapping $\Delta$, the graphs $H_i$ are pairwise vertex-disjoint. Indeed, suppose that a vertex $v$ is incident to arcs belonging to at least two different graphs~$H_i$. Then there is a sequence of arcs $a_1,a_2,\ldots,a_p$, all with head at $v$ and consecutive in the cyclic order around $v$, such that for some $i\neq j$, we have $\lambda'(a_1)=i^{\pm 1}$, $\lambda'(a_p)=j^{\pm 1}$, and $\lambda'(a_2)=\lambda'(a_3)=\ldots=\lambda'(a_{p-1})=1_{\Gamma}$. Then we have $\lambda'(a_1)\lambda'(a_2)\ldots \lambda'(a_p)=i^{\pm 1}j^{\pm 1}$, which in all four cases cannot belong to $\Delta(v)$, because by assumption, all words in $\Delta(v)$ have length at most $1$. This is a contradiction with $\psi$ being a solution.

    Next, observe that since $\psi$ is $R$-stable and $F_\calT\subseteq R$, we have that $\lambda'$ and $\lambda$ agree on arcs incident to the terminals, so $a(s_i)$ and $a(t_i)$ belong to $H_i$, for each $i\in \{1,\ldots,k\}$. Further, as every non-terminal vertex is clean with respect to $\lambda$, by \cref{lem:clean-invariance} the same can be also said about $\lambda'$. From this it follows that every vertex of $H_i$ except for $s_i$ and $t_i$ is {\em{balanced}} in $H_i$: it has the same number of outgoing and incoming arcs. Therefore, every (weakly) connected component of $H_i$ has an Eulerian closed walk, except for one component that contains both $s_i$ and $t_i$, which contains an Eulerian walk that starts at $s_i$ and ends at $t_i$. By exhaustively shortcutting cycles on this walk, we obtain a path $P_i$ leading from $s_i$ to $t_i$ that is entirely contained in $H_i$. As the graphs $H_i$ are pairwise vertex-disjoint, the same also holds for the paths~$P_i$. Finally, by the way we defined $\Delta$ on the arcs of $\Ev(G)$, it follows that $i\in \dcal(a)$ for each arc $a$ traversed on~$P_i$ (oriented from $s_i$ to $t_i$). We conclude that $\calP=\{P_i\colon i\in \{1,\ldots,k\}\}$ is a solution to $(G,\calT,\dcal)$. 
\end{proof}

\subsubsection{Reducing {\sc{PDSP}} to {\sc{Homology Feasibility}}}

We now proceed to the main technical ingredient of this section: we show how given an instance of \padpfull together with the skeleton $K$ provided by \cref{prop:niceSteiner} and suitable characteristic words for every spinal path of $K$, the search for a solution yielding the prescribed characteristic words can be reduced to solving an instance of {\sc{Homology Feasibility}}. To a large extent this boils down to lifting the arguments of Schrijver~\cite[Section 3.3]{schrijver1994finding}, but we need to be much more careful, because we need to bound the number of obtained homology types using the bound on the topological load of the skeleton $K$.

In the following, a mapping $w\colon \vec{\calQ}(K)\to \Gamma$, where $\Gamma$ is a group, shall be called {\em{regular}} if $w(\vec Q)^{-1}=w(\vec Q^{-1})$ for each $\vec Q\in \vec{\calQ}(K)$.

\newcommand{\calK}{{\cal K}}
\newcommand{\calL}{{\cal L}}
\newcommand{\wL}{\widehat{L}}

\begin{proposition}\label{lem:reduceToHom}
    There exists an algorithm that, given a nice instance $(G,\tcal,\dcal)$ of \padpfull, a skeleton $K$ of $(G,\tcal,\dcal)$, and a regular mapping $w\colon \vec{\calQ}(K)\to \Gamma_\tcal$, in polynomial time either concludes that there is no solution $\calP$ satisfying $\charWord(\calP,\vec Q)=w(\vec Q)$ for each $\vec Q\in \vec{\calQ}(K)$, or constructs an instance $(G,\xi,\Delta,F_\calT)$ of {\sc{Homology Feasibility}} such that the following two implications hold:
    \begin{itemize}[nosep]
        \item If the instance $(G,\tcal,\dcal)$ has a solution $\calP$ such that for every $\vec Q\in \vec{\calQ}(K)$ we have $\charWord(\calP,\vec Q)=w(\vec Q)$, then the instance $(G,\xi,\Delta,F_\tcal)$ has a positive answer.
        \item If the instance $(G,\xi,\Delta,F_\tcal)$ has a positive answer, then the instance $(G,\tcal,\dcal)$ has a solution.
    \end{itemize}
\end{proposition}
\begin{proof}
    Throughout the proof we denote $\Gamma=\Gamma_\tcal$ and $\Sigma=\Sigma_\tcal$ for convenience.
    Recall that $G^\star$ is the dual of $G$ and the mapping $e\mapsto e^\star$ maps every edge of $G$ to the corresponding (dual) edge of $G^\star$. We may also apply this mapping to the edges of $G^\star$ by setting $(e^\star)^\star=e$. For a set of edges $X$ (of $G$ or $G^\star$), we denote $X^\star = \{e^\star\colon e\in X\}$. This notation can be naturally extended also to arcs of $\Ev(G)$ and $\Ev(G^\star)$; note that then $(a^\star)^\star=a^{-1}$. 

    Recall that $\calQ(K)$ is the set of (undirected) spinal paths of $K$ and $\vec{\calQ}(K)$ consists of the orientations of the paths from $\calQ(K)$, and we have $|\calQ(K)|\leq 4k-3$. For each path $Q\in \calQ(K)$, pick an arbitrary edge $e_0(Q)$ of $Q$, and let $e(Q)=(e_0(Q))^\star$. (Note that $e_0(Q)$ is an edge of $G^\star$, while $e(Q)$ is an edge of $G$.) If $\vec Q$ is any of the two orientations of $Q$, then we let $a(\vec Q)$ be the orientation of $e(Q)$ so that $a(\vec Q)^\star=a_0(\vec Q)$, where $a_0(\vec Q)$ is the orientation of $e_0(Q)$ that belongs to (matches the orientation of) $\vec Q$.
    
    Next, let $W$ be the set of vertices consisting of all the terminals of $V_\calT$, and all the endpoints of the edges $e(Q)$ for $Q\in \calQ(K)$; thus $|W|\leq 9k-6$.
    Let
    $$H=G-\bigcup_{Q\in \calQ(K)} E(Q)^\star.$$
    Observe that since $K$ is a tree in the dual graph $G^\star$, $H$ stays connected. Therefore, we may pick $L$ to be an inclusion-wise minimal connected subgraph of $H$ that contains all the vertices of $W$. Clearly, $L$ is a tree whose every leaf is a vertex of $W$. 
    
    Now, let $\wL$ be the graph obtained from the tree $L$ by adding all the edges $e(Q)$, for $Q\in \calQ(K)$. Thus, $\wL$ is a subgraph of $G$, hence every face of $\wL$ is the union of a subset of faces of $G$. We now observe various properties of $\wL$ in a series of claims. 
    
    For the first one, let $R$ be the set of principal vertices of $K$, or equivalently, the faces of $G$ that are the endpoints of the paths of $\calQ(K)$. Recall that $F_\calT\subseteq R$.

    \begin{claim}\label{cl:oneR}
        Every face of $\wL$ contains at most one face belonging to $R$.
    \end{claim}
    \begin{innerproof}
        Consider any distinct $f,f'\in R$ and let $Q$ be any path of $\calQ(K)$ such that $Q$ is contained in the unique $(f,f')$-path in $K$. Observe that $e(Q)$ together with the unique path in $L$ connecting the endpoints of $e(Q)$ forms a cycle in $\wL$ that has the faces $f$ and $f'$ on the opposite sides. Thus, $f$ and $f'$ cannot be contained in the same face of $\wL$.
    \end{innerproof}

    The next claim is crucial: we show that every solution to the instance $(G,\tcal)$ can be shifted so that only the edges of $\wL$ get non-trivial labels, and the orientations of the edge $e(Q)$ get labels corresponding to the characteristic words of the orientations of $Q$, for every $Q\in \calQ(K)$. We shall call a shift $\psi\colon F(G)\to \Gamma$ {\em{$R$-stable}} if $\psi(f)=1_\Gamma$ for all $f\in R$.

    \begin{claim}\label{cl:push}
        Suppose $\calP$ is a solution to the instance $(G,\calT)$. Then there exists an $R$-stable shift $\psi$ satisfying the following conditions:
        \begin{itemize}[nosep]
            \item $\psi[\lambda_\calP](a)=1_\Gamma$ whenever $a$ is an orientation of an edge that does not belong to $\wL$; and
            \item for every $\vec Q\in \vec{\calQ}(K)$, we have $\psi[\lambda_\calP](a(\vec Q))=\charWord(\calP,\vec Q)$.
        \end{itemize}
    \end{claim}
    \begin{innerproof}
        For every face $g$ of $\wL$, let $G^\star_g$ be the subgraph of $G^\star-E(\wL)^\star$ induced by all the faces of $G$ contained in $g$. Clearly, $G^\star_g$ is connected and, by \cref{cl:oneR}, at most one vertex of $G_g$ belongs to $R$. Let $r_g$ be this vertex in case it exists, and otherwise let $r_g$ be any vertex of $G^\star_g$.
        
        Further, since $K$ is a tree that is a subgraph of $G^\star$, $G^\star_g\cap K$ is a forest that is a subgraph of $G^\star_g$. Let then $S_g$ be any spanning tree of $G^\star_g$ that contains $G^\star_g\cap K$. We root $S_g$ in $r_g$.

        We now define the shift $\psi$ as follows. Consider any face $f$ of $G$. Let $g$ be the face of $\wL$ that contains~$f$; equivalently, $f\in V(G^\star_g)$. Let $\vec Z_f$ be the (unique) $(r_g,f)$-path contained in $S_g$, oriented from $r_g$ to $f$. Finally, we set
        $$\psi(f)=\charWord(\calP,\vec Z_f).$$
        It remains to show that the shift $\psi$ defined in this manner satisfies all the requested properties.

        First, observe that if $f\in R$, then $f=r_g$.  Hence, $\vec Z_f$ is a one-vertex path and $\psi(f)=\charWord(\calP,\vec Z_f)=1_\Gamma$, as required.

        Second, consider any edge $e$ of $G$ that does not belong to $\wL$, and any orientation $a$ of $e$. Then $a^\star=(f,f')$, where $f,f'$ are two faces of $G$ that belong to $G^\star_g$ for some single face $g$ of $\wL$. If $e^\star$ belongs to $S_g$, then $\vec Z_{f'}$ is $\vec Z_f$ with $(f,f')$ appended, or $\vec Z_f$ is $\vec Z_{f'}$ with $(f',f)$ appended;  by symmetry assume the former. Then $\charWord(\calP,\vec Z_{f'})=\charWord(\calP,\vec Z_f)\cdot \lambda(a)$, hence 
        $$\psi[\lambda_\calP](a)=\psi(f)\cdot \lambda(a)\cdot \psi(f')^{-1}=\charWord(\calP,\vec Z_f)\cdot \lambda(a)\cdot \charWord(\calP,\vec Z_{f'})^{-1}=1_\Gamma.$$
        Next, assume that $e^\star$ does not belong to $S_g$. Then letting $\vec Z$ be the longest common prefix of $\vec Z_f$ and $\vec Z_{f'}$, we can decompose $\vec Z_f$ as the concatenation of $\vec Z$ and $\vec Z_f'$, and $\vec Z_{f'}$ as the concatenation of $\vec Z$ and $\vec Z_{f'}'$, so that the concatenation of $\vec Z'_f$, $a$, and $(\vec Z'_{f'})^{-1}$ is an oriented cycle, say $\vec C$. Since $\vec C$ is entirely contained in $G_g$, it is a non-separating cycle, hence by \cref{lem:Green} we conclude the following:
        $$1_\Gamma=\charWord(\calP,\vec C)=\charWord(\calP,\vec Z'_f)\cdot \lambda(a)\cdot \charWord(\calP,\vec Z'_{f'})^{-1}.$$
        By left-multiplying both sides by $\charWord(\calP,Z)$ and right-multiplying by $\charWord(\calP,\vec Z)^{-1}$, we obtain that
        \begin{align*}
        1_\Gamma & =\charWord(\calP,\vec Z)\cdot \charWord(\calP,\vec Z'_f)\cdot \lambda(a)\cdot \charWord(\calP,\vec Z'_{f'})^{-1}\cdot \charWord(\calP,\vec Z)^{-1}\\
        & =\charWord(\calP,\vec Z_f)\cdot \lambda(a)\cdot \charWord(\calP,\vec Z_{f'})^{-1}=\psi(f)\cdot \lambda(a)\cdot \psi(f')^{-1}=\psi[\lambda_\calP](a),
        \end{align*}
        as required.

        Finally, consider any $\vec Q\in \vec{\calQ}(K)$. Then the beginning and the end of $\vec Q$ are two faces $r_g,r_{g'}\in R$, where $g,g'$ are distinct faces of $\wL$. Let $a_0(\vec Q)=(f,f')$; then $f\in V(G_g)$ and $f'\in V(G_{g'})$. Observe that removing $a_0(\vec Q)$ from $\vec Q$ breaks $\vec Q$ into two subpaths $\vec Q'$ (with beginning $r_g$ and end $f$) and $\vec Q''$ (with beginning $f'$ and end $r_{g'}$). By construction, $\vec Q'$ and $\vec Q''$ are contained in trees $S_g$ and $S_{g'}$, respectively (after forgetting orientations), hence
        $$\psi(f)=\charWord(\calP,\vec Q')\qquad\textrm{and}\qquad \psi(f')=\charWord(\calP,\vec Q'')^{-1}.$$
        We conclude that
        \begin{align*}\psi[\lambda_{\calP}](a(\vec Q)) & = \psi(f)\cdot \lambda_\calP(a(\vec Q))\cdot \psi^{-1}(f') = \charWord(\calP,\vec Q')\cdot \lambda^\star_\calP(a_0(\vec Q))\cdot \charWord(\calP,\vec Q')^{-1}\\
        & =  \charWord(\calP,\vec Q')\cdot \lambda^\star_\calP(a_0(\vec Q))\cdot \charWord(\calP,(\vec Q')^{-1})=\charWord(\calP,\vec Q),\end{align*}
        as required.
    \end{innerproof}

    Now, call a $\Gamma$-labelling $\xi$ of $G$ {\em{compact}} if the following conditions are satisfied:
    \begin{itemize}[nosep]
        \item $\xi(a)=1_\Gamma$ for every $a\in \Ev(G)$ that is an orientation of an edge not belonging to $\wL$;
        \item every non-terminal vertex of $G$ is clean with respect to $\xi$;
        \item for every terminal pair $(s_i,t_i)\in \calT$, we have $\xi(a_{s_i})=i^{-1}$ and $\xi(a_{t_i})=i$, where $a_{s_i}$ and $a_{t_i}$ are the unique arcs with head $s_i$ and $t_i$, respectively.
    \end{itemize}
    In terms of this definition, \cref{cl:push} implies that every labelling originating in a solution to $(G,\calT,\dcal)$ can be shifted to a compact labelling, as explained formally in the next claim.

    \begin{claim}\label{cl:compact}
        Suppose $\calP$ is a solution to the instance $(G,\calT,\dcal)$. Then the $R$-stable shift $\psi$ provided by \cref{cl:push} for $\calP$ is compact.
    \end{claim}
    \begin{innerproof}
        The first condition of compactness follows directly from \cref{cl:push}, the second condition follows from \cref{lem:clean-invariance} combined with \cref{lem:solution-clean}, and the third condition is ensured by the $R$-stability of~$\psi$.
    \end{innerproof}

    We next observe that the space of compact labellings is already very restricted: knowing the values on arcs $a(\vec Q)$ for $\vec Q\in \vec \calQ(K)$ --- which are fixed by the given mapping $w$ --- is enough to reconstruct the whole compact labelling.

    \begin{claim}\label{cl:reconstruction}
        There exists at most one compact $\Gamma$-labelling $\xi$ of $G$ such that $\xi(a(\vec Q))=w(\vec Q)$ for all $\vec Q\in \vec \calQ(K)$. Moreover, this compact labelling, or a conclusion about its non-existence, can be computed in polynomial time. 
    \end{claim}
    \begin{innerproof}
        Let $L'$ be $L$ with all the terminals of $V(\calT)$ removed; thus $L'$ is also a tree. Let $\xi$ be a labelling as in the lemma statement. 
        Observe that the restrictions imposed on $\xi$ fix the labels of all the arcs except for the orientations of the edges of $L'$:
        \begin{itemize}[nosep]
            \item For every arc $a$ that is an orientation of an edge outside of $\wL$, we have $\xi(a)=1_\Gamma$.
            \item For every arc of the form $a(\vec Q)$ for some $\vec Q\in \vec \calQ(K)$, we have $\xi(a(\vec Q))=w(\vec Q)$.
            \item For every arc incident to a terminal, say belonging to the $i$th pair of $\calT$, its value under $\xi$ is $i$ or $i^{-1}$, depending on the orientation and whether it is the source or the target terminal.
        \end{itemize}
        
        It remains to show that the values on the orientations of the edges of $L'$ can be also uniquely determined. Root $L'$ at any vertex $r$. We process the edges of $L'$ in a bottom-up manner, at each step determining the values of the orientations of an edge incident to a leaf, say $u$, followed by removing $u$ from $L'$. Thus, at the moment of processing, $u$ is incident to exactly one edge $e$ whose orientations have values under $\xi$ not yet defined. Since $u$ should be clean with respect to $\xi$ (because the terminals of $V(\calT)$ are not in $L'$), there is a unique way of fixing the values of $\xi$ on the orientations of $e$ so that $u$ becomes clean. We fix those values and proceed by removing $u$ and $e$ from $L'$.

        At the end, all the values of $\xi$ are uniquely determined, and all the non-terminal vertices are clean with respect to $\xi$,  possibly except for  the root $r$. Hence, we verify whether $r$ is clean. If this is the case, we may output $\xi$, and otherwise we may conclude that there is no labelling $\xi$ satisfying the required conditions.
    \end{innerproof}
    
    Call a solution $\calP$ to $(G,\calT,\dcal)$  {\em{compatible}} if $\charWord(\calP,\vec Q)=w(\vec Q)$ for each $\vec Q\in \vec{\calQ}(K)$.
    We are ready to present the algorithm to produce the instance of {\sc{Homology Feasibility}} promised in the lemma statement:
    \begin{itemize}[nosep]
        \item Run the algorithm of \cref{cl:reconstruction} to compute a compact $\Gamma$-labelling $\xi$ of $G$ such that $\xi(a(\vec Q))=w(\vec Q)$ for all $\vec Q\in \vec \calQ(K)$. If the algorithm concluded that there is no such labelling $\xi$, then terminate and return the conclusion that there is no compatible solution.
        \item Output the instance $(G,\xi,\Delta,F_\tcal)$, where $\Delta$ is defined by setting $\Delta(a)=\{1_{\Gamma_\calT}\}\cup \{i\colon i\in \dcal(a)\}\cup \{i^{-1}\colon i\in \dcal(a^{-1})\}$ for each $a\in \Ev(G)$ and $\Delta(v)=\{1_\Gamma\}\cup \{i,i^{-i}\colon i\in \Sigma\}$ for each $v\in V(G)$.
    \end{itemize}
    We are left with proving that if the algorithm returns the conclusion about the non-existence of a compatible solution $\calP$, then this conclusion is correct, and otherwise the computed instance $(G,\xi,\Delta,F_\tcal)$ has the required properties.

    First, suppose that $\calP$ is a compatible solution to $(G,\calT,\dcal)$. Let $\psi$ be the $R$-stable shift provided by \cref{cl:push} for $\calP$. Then denoting $\xi=\psi[\lambda_\calP]$, we have that $\xi$ satisfies $\xi(a(\vec Q))=w(\vec Q)$ for all $\vec Q\in \vec \calQ(K)$, by the conditions asserted by \cref{cl:push} and the compatibility of $\calP$. Further, $\xi$ is compact by \cref{cl:compact}. By \cref{cl:reconstruction}, we conclude that the algorithm of this claim actually succeeds in constructing $\xi$; this in particular shows that the conclusion about the non-existence of a compatible solution is always correct. Now as $\xi=\psi[\lambda_\calP]$, we have $\lambda_\calP=\psi^{-1}[\xi]$, where $\psi^{-1}$ is defined by $\psi^{-1}(f)=\psi(f)^{-1}$ for all $f\in F(G)$. Noting that $\psi^{-1}$ is also $R$-stable, and recalling that $F_\tcal\subseteq R$, we see that $\psi^{-1}$ witnesses that the instance $(G,\xi,\Delta,F_\tcal)$ of {\sc{Homology Feasibility}} has a solution.

    Second, suppose the instance $(G,\xi,\Delta,F_\tcal)$ of {\sc{Homology Feasibility}} has a solution. By the compactness of $\xi$ and the construction of $\Delta$, the instance $(G,\lambda,\Delta,F_\tcal)$ satisfies the prerequisites of \cref{lem:solHom}. Hence, from this lemma we infer that the instance $(G,\calT,\dcal)$ of {\sc{Planar Disjoint Annotated Paths}} also has a solution.
\end{proof}

The next proposition pipelines \cref{prop:niceSteiner} with \cref{lem:reduceToHom}, giving a way to reduce solving an instance of \pdspfull to solving a number of instance of {\sc{Homology Feasibility}}.

\begin{proposition}\label{lem:enumHom}
    There exists an algorithm that, given a nice instance $(G,\tcal)$ of \pdspfull with $|\tcal|=k$, in time $2^{\Oh(k\log k)}\cdot n^{\Oh(1)}$
    computes a family $\Phi$ of instance of {\sc{Homology Feasibility}} with the following properties:
    \begin{itemize}[nosep]
        \item $|\Phi|\leq 2^{\Oh(k\log k)}$.
        \item If the instance $(G,\tcal)$ has a solution, then at least one of the instances $I\in \Phi$ has a solution.
        \item If any of the instances $I\in \Phi$ has a solution, then so does $(G,\tcal)$.
    \end{itemize}
\end{proposition}

Let us note that \cref{lem:enumHom} is very easy if one is content with a worse bound of $2^{\Oh(k^4\log k)}$ on the size of $\Phi$ and the running time. Namely, do the following:
\begin{itemize}[nosep]
    \item Apply the algorithm of \cref{prop:niceSteiner}, obtaining an instance $(G',\tcal',\dcal)$ of \padpfull with $|\tcal'|=k$ and its skeleton $K$
    \item For every regular mapping $w\colon \vec{\calQ}\to \Gamma_{\tcal'}$  that has only words of length $\Oh(k^3)$ in its image, apply \cref{lem:reduceToHom} to $w$ and include the obtained instance of {\sc{Homology Feasibility}} in $\Phi$. As $|\vec{\calQ}(K)|\leq \Oh(k)$ and $|\overline{\Gamma}_\tcal|=2k$, there are only $2^{\Oh(k^4\log k)}$ such mappings.
\end{itemize}
Therefore, the main matter in the proof of \cref{lem:enumHom} will be how to reduce the number of relevant mappings to $2^{\Oh(k\log k)}$ by a more careful analysis. Importantly, we will use the fact that \cref{prop:niceSteiner} provides an upper bound on the topological load of the obtained skeleton, and not just on the algebraic~load.

\newcommand{\calR}{{\mathcal R}}

\begin{proof}[Proof of \cref{lem:enumHom}]
We first prepare some simple combinatorial tools.
For a positive integer $p$, we denote $[p]=\{1,\ldots,p\}$.
A {\em{matching}} in $[p]$ is a set $M$ consisting of pairwise disjoint subsets of $[p]$, each of size~$2$. We say that $a,b\in [p]$ are {\em{matched}} in $M$ if $\{a,b\}\in M$. A matching $M$ is {\em{non-crossing}} if there are no indices $a<b<c<d$ such that $a$ is matched with $c$ and $b$ is matched with $d$. Denoting $V(M)=\bigcup M$, we call $M$ {\em{perfect}} if $V(M)=[p]$.

An {\em{block}} in $[p]$ is a contiguous subset of $[p]$, that is, a subset $B$ such that for all $a\leq b\leq c$, that $a,c\in B$ implies also $b\in B$. A {\em{division}} of $[p]$ is a partition of $[p]$ into blocks. For a division $\calR$ of $[p]$, a matching $M$ in $[p]$ is called {\em{$\calR$-jumping}} if for all $a,b$ matched in $M$, $a$ and $b$ do not belong to the same block of $\calR$.
We will use the following combinatorial bound.

\begin{claim}\label{cl:matching-bound}
    Let $p$ be a positive integer and let $\calR$ be a division of $[p]$ with $|\calR|=r$. Then the number of different perfect non-crossing $\calR$-jumping matchings in $[p]$ is bounded by $p^{\Oh(r)}$, and they can be enumerated in time $p^{\Oh(r)}$.
\end{claim}
\begin{innerproof}
    Consider a perfect non-crossing $\calR$-jumping matching $M$ in $[p]$. 
    Call $a\in [p]$ {\em{left}} if $a$ is matched in $M$ with some $b>a$, and {\em{right}} otherwise. As $M$ is $\calR$-jumping, no pair of $M$ is contained in a single block of $\calR$. As $M$ is additionally non-crossing, every block $B$ of $\calR$ can be partitioned into a prefix $B'$ and a suffix $B''$, both possibly empty, so that elements of $B'$ are right and the elements of $B''$ are left. Let $\ell(B)$ be the length of $B'$.

    The observation now is that the mapping $\ell$ allows to uniquely reconstruct the matching $M$. Indeed, construct a word $w$ of length $p$ over the alphabet $\{(,)\}$ by taking $[p]$, replacing every left element by $($ and every right element by $)$, and writing the obtained symbols in the natural order. This can be done in a unique way knowing the mapping $\ell$, as this mapping entirely describes which elements of $[p]$ are left and which are right. Then $w$ must be a valid bracket expression, and $M$ is given by pairs of matching brackets. 

    Since the number of blocks is $r$, the number of possible mappings $\ell$ is bounded by $p^{\Oh(r)}$, which by the above observation translates to the same bound on the number of considered matchings. The enumeration statement follows by considering all possible mappings $\ell$ and outputting all suitable matchings corresponding to them.
\end{innerproof}

With these preliminaries settled, let us proceed with the proof of the lemma.  The first step is to apply \cref{prop:niceSteiner}, and thus obtain an instance $(G',\calT',\dcal)$ of \padpfull{} with $|\calT'|=k$ and its skeleton $K$. We denote $\overline{\Sigma}=\overline{\Sigma}_{\calT'}$ and $\Gamma=\Gamma_{\calT'}$ for brevity. Recall that by $\calQ(K)$ we denote the set of spinal paths of $K$, and $\vec \calQ(K)$ are the orientations of paths from $\calQ(K)$.

Let us consider any solution $\calP$ to $(G',\calT',\dcal)$ such that the topological load of $K$ with respect to $\calP$ is $\Oh(k^3)$. Each path $P_i\in \calP$ is considered to be oriented from $s_i$ to $t_i$. We consider the plane drawing consisting of the tree $K$ and paths from $\calP$. Note that each path $P_i\in \calP$ crosses every spinal path $Q\in \calQ(K)$ in a finite number of points. We will treat paths $P_i$ and spinal paths $Q\in \calQ(K)$ as curves in the plane.

\begin{figure}[t]
\centering
\includegraphics[scale=0.5]{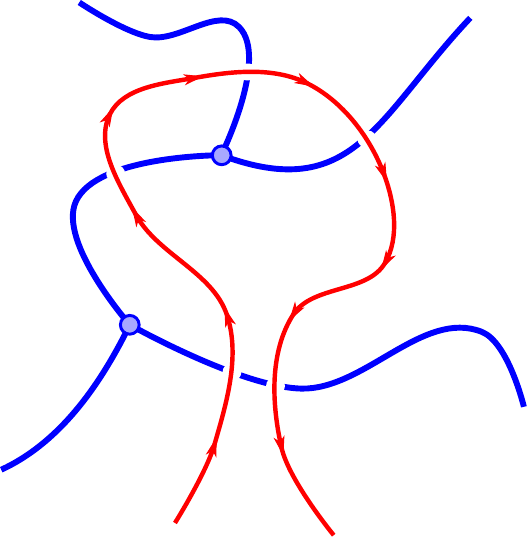}
\qquad\qquad\qquad
\includegraphics[scale=0.5]{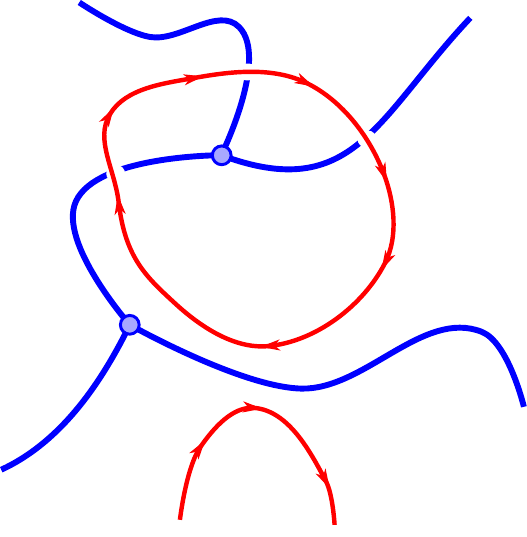}
\caption{Splitting: before (left) and after (right) the operation. The skeleton $K$ is in blue, the path $P_i$ is in~red.}\label{fig:splitting}
\end{figure}

We now apply the following modifications to the drawing; during these modifications, each path $P_i$ may be turned into an oriented path plus a number of oriented cycles. Consider a spinal path $Q\in \calQ(K)$ and the sequence of consecutive operations of simplification of the configuration $(\calP,Q)$ as in the definition of the topological load of $(\calP,Q)$, that is, leading from $H_{(\calP,Q)}$ to $\widetilde{H}_{(\calP,Q)}$. Now apply the same sequence of operations to the current drawing, but every operation of simplification is replaced by the operation of {\em{splitting}} depicted in \Cref{fig:splitting}: the two crossings of $Q$ with some $P_i$ are removed by removing short segments on $P_i$ around the crossing points, and the resulting four ``loose ends'' of $P_i$ are connected so that the connections are following closely $Q$ but not crossing it. Recall that the operation is applied only when the segment of $P_i$ between the two crossings forms an empty handle of $Q$, hence the operation ``splits off'' from $P_i$ a non-separating oriented cycle. 

We now apply this operation to all the spinal paths of $K$, in any order. Let $\calP^1=(\widehat{P}^1_1,\ldots,\widehat{P}^1_k)$ be the image of $\calP=(P_1,\ldots,P_k)$ after the operations. Thus, each $\widehat{P}^1_i$ consists of an oriented curve leading from $s_i$ to $t_i$, which we will call $P^2_i$, and a number of oriented cycles, each of them non-separating. Let also $\calP^2=(P_1^2,\ldots,P_k^2)$. Note that for each $Q\in \calQ(K)$, the number of crossings of $Q$ with the elements of $\calP^1$ is equal to the $\topload(\calP,Q)$, and hence it is bounded by~$\Oh(k^3)$.

\begin{figure}[t]
\centering
\includegraphics[scale=0.8]{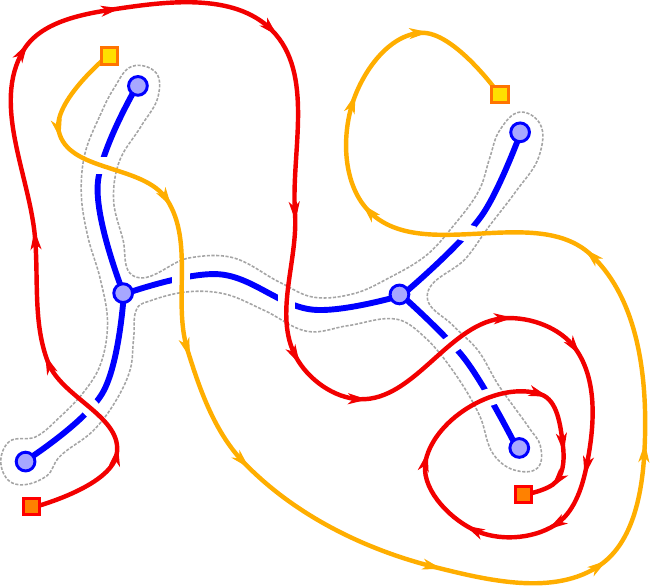}
\caption{Skeleton $K$, depicted in blue, and a solution $\calP^2=\calP=(P_1,P_2)$ with $P_1$ in orange and $P_2$ in red. The dotted closed curve is $\gamma$, the Euler tour of $K$. The word $u$ for this situation is $\underline{1}|1|1^{-1}2^{-1}|1|\underline{1^{-1}}|1^{-1}|22|\underline{2^{-1}}|2^{-1}2^{-1}|21|2^{-1}|\underline{2}|2|1^{-1}$, where the underlines denote the terminal symbols and the $|$ markers delimit the division $\calR$.}\label{fig:euler}
\end{figure}

Consider now the {\em{Euler tour}} of $K$: a closed curve $\gamma$ that follows the tree $K$ around in very close proximity, in the clockwise order; see \Cref{fig:euler}. Choosing a starting point of $\gamma$ next to $s_1$, we construct a word $u$ by following along $\gamma$ and recording symbols from $\overline{\Sigma}$ as follows:
\begin{itemize}[nosep]
    \item When $\gamma$ passes next to $s_i$, we record $i$. (This happens right at the beginning for $s_1$.)
    \item When $\gamma$ passes next to $t_i$, we record $i^{-1}$.
    \item When a path $P_i$ crosses $\gamma$ we record $i$ or $i^{-1}$, depending whether the crossing is from right to left or from left to right.
\end{itemize}
Naturally, $u$ can be partitioned into {\em{terminal symbols}} --- symbols $i,i^{-1}$ added for the sake of terminals $s_i,t_i$ --- and subwords $u(\vec Q)$ for $\vec Q\in \vec \calQ(K)$, recorded while $\gamma$ was following a spinal path $Q$ along its orientation~$\vec Q$. Letting $p$ be the length of $u$, by $\calR$ we denote the following division of $[p]$:
\begin{itemize}[nosep]
    \item for every terminal symbol, create a separate block of length $1$ consisting of the position of this symbol; and
    \item for each subword $u(\vec Q)$ for $\vec \calQ\in \vec \calQ(K)$, create one block consisting of positions on which $u(\vec Q)$ appears.
\end{itemize}
Note that since there are $2|\calT|=2k$ terminal symbols and $|\vec Q|\leq 8k-6$, we have $|\calR|\leq 10k-6$.

For $t\in \{1,2\}$, we analogously define the word $u^t$, subwords $u^t(\vec Q)$ for $\vec Q\in \vec \calQ(K)$, length $p^t$, and division $\calR^t$ of $[p^t]$, using the family $\calP^t$ instead of $\calP$. By the construction of $\calP$ and $\calP^1$, we have
$$\reduce(u(\vec Q))=\reduce(u^1(\vec Q))=\charWord(\calP,\vec Q)\qquad\textrm{for each }\vec Q\in \vec\calQ(K).$$
Define the mapping $w\colon \vec \calQ(K)\to \Gamma$ that assigns to each $\vec Q\in \vec \calQ(K)$ the common value above:
$$w(\vec Q)=\charWord(\calP,\vec Q)\qquad \textrm{for all }\vec Q\in \vec \calQ(K).$$
Observe that the solution $\calP$ witnesses that the instance $(G',\xi,\Delta,F_{\tcal'})$ obtained by applying \cref{lem:reduceToHom} to the mapping $w$ has a solution.

Now, define the mapping $w'\colon \vec \calQ(K)\to \Gamma$ as follows:
$$w'(\vec Q)=\reduce(u^2(\vec Q))\qquad \textrm{for all }\vec Q\in \vec \calQ(K).$$
Note that 
$$|w'(\vec Q)|\leq |u^2(\vec Q)|\leq |u^1(\vec Q)|\leq \Oh(k^3).$$
Consider applying \cref{lem:reduceToHom} to $w'$ instead of $w$, and let $(G',\xi',\Delta',F_{\tcal'})$ be the obtained instance. 
By the same argument as in the proof of \cref{lem:reduceToHom} we get that $\Delta'=\Delta$ and the labelling $\xi'$ differs from $\xi$ by applying an $F_{\tcal'}$-stable shift. Indeed, because $\calP^2$ differs from $\calP^1$ only by removal of some non-separating oriented cycles, each removal of such a cycle $C$, say belonging to $P^1_i$, can be modelled by applying a shift that assigns $i$ or $i^{-1}$ (depending on the orientation of $C$) to the faces of $G$ corresponding to the inside of $C$, and $1_{\Gamma}$ to the other faces; applying those shifts for all the cycles present in $\calP^1$ and not in $\calP^2$ transforms $\xi$ to $\xi'$. Hence, as $(G',\xi,\Delta,F_{\tcal'})$ has a solution, so does $(G',\xi',\Delta',F_{\tcal'})$ as well.

The observation now is that there are only $2^{\Oh(k\log k)}$ candidates for the word $u^2$. For this, consider construction of the following two matchings $M$ and $N$ in $[p^2]$:
\begin{itemize}[nosep]
    \item Every crossing of a path $P^2_i$ with a spinal path $Q\in \calQ(K)$ corresponds to two positions in $u^2$, for the two orientations of $Q$. Match those two positions in $M$.
    \item For every segment of $P^2_i$ between two crossings with $K$ (or between $s_i/t_i$ and the first/last crossing), the endpoints of this segment naturally correspond to two positions in $u^2$. Match those two positions in $N$.
\end{itemize}
Thus, $N$ is a perfect matching, $M$ matches all  the positions except for those corresponding to the terminal symbols, and $M\cup N$ consists of $k$ disjoint paths. The $i$-th path of $M\cup N$ connects the positions of $[p]$ corresponding to $s_i$ and $t_i$ (which hence contain symbols $i$ and $i^{-1}$, respectively) and alternately visits positions with symbols $i$ and $i^{-1}$.
Moreover, since curves $P^2_1,\ldots,P^2_k$ are pairwise disjoint, it follows that both $M$ and $N$ are non-crossing. Finally, the iterative removal of empty handles in the construction of $\calP^2$ implies that both $M$ and $N$ are $\calR^2$-jumping.

Concluding, the whole word $u^2$ and the subwords $u^2(\vec Q)$ for $\vec Q\in \mathcal{Q}$ can be uniquely reconstructed from the following information:
\begin{itemize}[nosep]
    \item The length $p^2$ of $u^2$. As $|u^2(\vec Q)|\leq \Oh(k^3)$ for each $\vec Q\in \vec \calQ(K)$, we have $p^2\leq \Oh(k^4)$, hence there are $\Oh(k^4)$ ways to choose $p^2$.
    \item The division $\calR^2$ of $[p^2]$. Since $|\calR|\leq 10k-6$, there are at most $(p^2+1)^{10k-6}=2^{\Oh(k\log k)}$ ways to choose $\calR^2$.
    \item The labelling of $S^2$ with pairwise different elements of $\overline{\Sigma}$; there are $(2k)!=2^{\Oh(k\log k)}$ ways to choose this labelling.
    \item The matchings $M$ and $N$ that are non-crossing and $\calR$-jumping, and such that $N$ is perfect and $M$ matches all positions except for those corresponding to the terminal symbols. By \cref{cl:matching-bound}, there are $2^{\Oh(k\log k)}$ candidates for $N$ and  $2^{\Oh(k\log k)}$ candidates for $M$, where in the second bound we consider the word with the blocks corresponding to the terminal symbols removed.
    \item Once $M$ and $N$ are fixed, we verify that $M\cup N$ consists of $k$ disjoint paths, and the remaining piece of information is a labelling of the endpoints of those paths with distinct elements of $\overline{\Sigma}$ so that the endpoints of every path are always labelled with a pair of opposite elements ($i$ and $i^{-1}$). There are $2^k\cdot k!=2^{\Oh(k\log k)}$ such labellings.
\end{itemize}
Thus, we can enumerate $2^{\Oh(k\log k)}$ candidates for the mapping $u^2\colon \vec \calQ(K)\to \overline{\Sigma}^\star$, which gives us also $2^{\Oh(k\log k)}$ candidates for the mapping $w'\colon \vec \calQ(K)\to \Gamma$ defined through $w'(\vec Q)=\reduce(u^2(\vec Q))$. 

Now, define $\Phi$ to be the set of all instances of {\sc{Homology Feasibility}} obtained by applying algorithm of \cref{lem:reduceToHom} to the mappings $w'$ constructed above. From the discussion above it follows that provided $(G,\tcal)$ has a solution, one of the instances of $\Phi$ will have a solution. On the other hand, from \cref{lem:solHom} it follows that if any of the instances of $\Phi$ has a solution, then so does $(G',\tcal',\dcal)$, and consequently also~$(G,\tcal)$.
\end{proof}

\subsubsection{Wrap-up: Proof of \texorpdfstring{\cref{thm:main}}{Theorem 1.1}}

With all ingredients in place, we can finish the proof of our main result.

\begin{proof}[Proof of \cref{thm:main}]
Due to \Cref{lem:nice-reduction} we can assume that the  input instance $(G,\tcal)$ of \pdspfull is nice.
Then we simply apply the algorithm from \cref{lem:enumHom} to $(G,\tcal)$ and solve each of the obtained instances of {\sc{Homology Feasibility}} using the algorithm of \cref{thm:Schrijver}.
\end{proof}

\section{Conclusion}

We have classified the \pdspfull problem as fixed-parameter tractable by presenting a $2^{\Oh(k\log k)}\cdot n^{\Oh(1)}$-time algorithm.
An intriguing question is whether  this problems is indeed easier than {\sc Planar Disjoint Paths}, or can one improve the running time $2^{\Oh(k^2)}\cdot n^{\Oh(1)}$ for the latter.
Our algorithm involves a careful preparation of a Steiner tree spanning the terminals which facilitates enumerating the relevant homology classes of a solution.
Is it possible to build such a Steiner tree for {\sc Planar Disjoint Paths} that would accommodate only $2^{o(k^2)}$ relevant homology classes?

\bibliographystyle{plain}

\end{document}